\documentclass[11pt]{article}

\newif\ifnotes
\notestrue

\title{Quantum Codes with Addressable and Transversal \\ Non-Clifford Gates}

\author{Zhiyang He\thanks{Email: \texttt{szhe@mit.edu}. Supported by the MIT Department of Mathematics and the NSF Graduate Research Fellowship Program under Grant No. 2141064.}\\MIT
\and Vinod Vaikuntanathan\thanks{Email: \texttt{vinodv@mit.edu}. Supported in part by NSF grant CNS-2154149 and a Simons Investigator Award.}\\MIT
\and Adam Wills\thanks{Email: \texttt{a\_wills@mit.edu}. Affiliated with the MIT Center for Theoretical Physics - a Leinweber Institute. Supported by the MIT Department of Physics. This preprint is assigned number MIT-CTP/5813.}\\MIT
\and Rachel Yun Zhang\thanks{Email: \texttt{rachelyz@mit.edu}. Supported by NSF Graduate Research Fellowship 2141064. Supported in part by NSF grant CNS-2154149.}\\MIT}
\date{\today}

\usepackage[hyperindex,breaklinks]{hyperref}
\usepackage[margin=1in]{geometry}
\usepackage{mathptmx}
\usepackage{amssymb}
\usepackage{amsmath}
\usepackage{amsthm}
\usepackage{amsfonts}
\usepackage[T1]{fontenc}
\usepackage{bm}
\usepackage{color}
\usepackage{comment}
\usepackage[capitalize]{cleveref}
\usepackage[dvipsnames]{xcolor}
\usepackage{float}
\usepackage{todonotes}
\usepackage{asymptote}
\usepackage{mdframed}
\usepackage[most]{tcolorbox}
\usepackage{hyperref}
\usepackage[shortlabels]{enumitem}
\usepackage{framed}
\usepackage{mdframed}
\usepackage{scrextend}
\usepackage{multirow}
\usepackage{ifthen}
\usepackage{bbm}
\usepackage{frcursive}
\usepackage{thm-restate}
\usepackage{booktabs}
\usepackage{array}
\usepackage{graphicx}
\usepackage{physics}
\usepackage{mathtools}

\definecolor{denim}{rgb}{0.08, 0.38, 0.74}
\definecolor{classicrose}{rgb}{0.98, 0.8, 0.91}
\definecolor{darkpastelblue}{rgb}{0.47, 0.62, 0.8}
\definecolor{dogwoodrose}{rgb}{0.84, 0.09, 0.41}

\usepackage{hyperref}
\hypersetup{
    colorlinks=true,
    linkcolor=dogwoodrose,
    filecolor=dogwoodrose,
    citecolor=denim,
    urlcolor=denim,
}
\usepackage[hyperpageref]{backref}

\newtheorem{theorem}{Theorem}[section]
\newtheorem{lemma}[theorem]{Lemma}
\newtheorem*{lemma*}{Lemma}

\newtheorem{proposition}[theorem]{Proposition}
\newtheorem{corollary}[theorem]{Corollary}

\newtheorem{assumption}[theorem]{Assumption}
\newtheorem{definition}[theorem]{Definition}
\newtheorem{fact}[theorem]{Fact}

\theoremstyle{definition}
\newtheorem{remark}[theorem]{Remark}

\Crefname{theorem}{Theorem}{Theorems}
\Crefname{claim}{Claim}{Claims}
\Crefname{lemma}{Lemma}{Lemmas}
\Crefname{proposition}{Proposition}{Propositions}
\Crefname{corollary}{Corollary}{Corollaries}
\Crefname{definition}{Definition}{Definitions}

\renewcommand{\le}{\leqslant}
\renewcommand{\leq}{\leqslant}
\renewcommand{\ge}{\geqslant}
\renewcommand{\geq}{\geqslant}

\newcommand{\CCZ}{\mathsf{CCZ}}
\newcommand{\CCCZ}{\mathsf{CCCZ}}
\newcommand{\CLZ}{\mathsf{C^{(\ell-1)}Z}}
\newcommand{\CSS}{\text{CSS}}
\newcommand{\bigN}{N}

\newcommand{\bbF}{\mathbb{F}}

\newcommand{\cC}{\mathcal{C}}

\newcommand{\cG}{\mathcal{G}}
\newcommand{\cH}{\mathcal{H}}

\newcommand{\cL}{\mathcal{L}}
\newcommand{\cM}{\mathcal{M}}

\newcommand{\cP}{\mathcal{P}}

\makeatletter
\newcommand{\customlabel}[2]{%
   \protected@write \@auxout {}{\string \newlabel {#1}{{#2}{\thepage}{#2}{#1}{}} }%
   \hypertarget{#1}{#2}
}
\makeatother

\newcounter{datacounter}

\newcounter{casenum}

\newcommand{\case}[2]{
    \refstepcounter{casenum}
    \ifthenelse{\equal{\value{casenum}}{0}}{
    \vskip.5\baselineskip\par\noindent
    }{}
    \noindent {\it Case \arabic{casenum}:} {\it #1}
    \vskip0.1\baselineskip
    \begin{addmargin}[1.5em]{1em}
    #2
    \end{addmargin}
}

\newcounter{subcasenum}

\newcounter{casenumb}

\newcounter{subcasenumb}

\usepackage{cite}

\newcommand{\sm}{\setminus}
\newcommand{\al}{\alpha}
\newcommand{\bal}{\boldsymbol{\alpha}}

\newcommand{\one}{\boldsymbol{1}}

\newcommand{\bx}{\mathbf{x}}
\newcommand{\bz}{\mathbf{z}}

\newcommand{\qcode}{\mathcal{Q}}
\newcommand{\qoutcode}{\mathcal{Q}_{\text{out}}}
\newcommand{\coutcode}{\mathcal{C}_{\text{out}}}
\newcommand{\enc}{\phi}
\usepackage[normalem]{ulem}

\newcommand{\CC}{\mathbb{C}}
\newcommand{\FF}{\mathbb{F}}
\newcommand{\KK}{\mathbb{K}}
\newcommand*{\ol}{\overline}  
\usepackage{braket}
\usepackage[normalem]{ulem}
\DeclareMathAlphabet{\mathcal}{OMS}{cmsy}{m}{n}

\newcommand{\K}{K}
\newcommand{\N}{N}

\newcommand{\punc}[2]{#1_{\sm #2}}
\newcommand{\puncT}[1]{\punc{#1}{T}}

\newcommand{\shrT}[1]{#1_{T = 0}}

\DeclareMathOperator{\im}{im}
\DeclareMathOperator{\RS}{RS}
\DeclareMathOperator{\GRS}{GRS}

\newcommand{\B}{\mathcal{B}}

\newcommand{\ME}{\mathsf{ME}_{q\rightarrow 2}}
\newcommand{\MEinv}{\mathsf{ME}_{2 \rightarrow q}}

\newcommand{\phyC}[1]{C(#1)}

\begin{document}

\sloppy
\maketitle
\renewcommand{\thefootnote}{\fnsymbol{footnote}}
\setcounter{footnote}{5}
\footnotetext{{The authors are listed in alphabetical order.}}
\setcounter{footnote}{0}
\renewcommand{\thefootnote}{\arabic{footnote}}

\begin{abstract}

The development of quantum codes with good error correction parameters and useful sets of transversal gates is a problem of major interest in quantum error correction. Abundant prior works have studied transversal gates which are restricted to acting on all logical qubits simultaneously.
In this work, we study codes that support transversal gates which induce \textit{addressable} logical gates, i.e., the logical gates act on logical qubits of our choice.
As we consider scaling to high-rate codes, the study and design of low-overhead, addressable logical operations presents an important problem for both theoretical and practical purposes.

Our primary result is the construction of an explicit qubit code for which \textit{any} triple of logical qubits across one, two, or three codeblocks can be addressed with a logical $\CCZ$ gate via a depth-one circuit of physical $\CCZ$ gates, and whose parameters are asymptotically good, up to polylogarithmic factors. The result naturally generalizes to other gates including the $\mathsf{C}^{\ell} Z$ gates for $\ell \neq 2$.

Going beyond this, we develop a formalism for constructing quantum codes with {\em addressable and transversal} gates.
Our framework, called \textit{addressable orthogonality}, encompasses the original triorthogonality framework of Bravyi and Haah (Phys. Rev. A 2012), and extends this and other frameworks to study addressable gates.
We demonstrate the power of this framework with the construction of an asymptotically good qubit code for which {\em pre-designed}, pairwise disjoint triples of logical qubits within a single codeblock may be addressed with a logical $\CCZ$ gate via a physical depth-one circuit of $\mathsf{Z}$, $\mathsf{CZ}$ and $\CCZ$ gates. 
In an appendix, we show that our framework extends to addressable and transversal $T$ gates, up to Clifford corrections.

\end{abstract}
\thispagestyle{empty}
\newpage

\enlargethispage{1cm}
\tableofcontents
\pagenumbering{roman}
\newpage
\pagenumbering{arabic}

\section{Introduction}

    It is a major challenge in fault-tolerant quantum computation to lower the overhead of quantum error correction~\cite{shor1995scheme,gottesman1997stabilizer,aharonov1997fault}, that is, the amount by which a quantum circuit must be made larger, both in space and in time, in order to be made fault-tolerant. Significant progress has been made in this area in many directions, although we still notably lack a well-developed theory of quantum codes supporting transversal gates. 
    
    The most standard notion of transversality is as follows. We say that a code on $n$ physical qubits encoding $k$ logical qubits supports the single-qubit gate $U$ \emph{transversally} if, by acting on the physical qubits with the gate $U^{\otimes n}$, the induced effect on the logical qubits is to execute the gate $U$ on all $k$ encoded qubits, i.e., the gate $\overline{U^{\otimes k}}$ is executed. The point is that acting transversally is an inherently fault-tolerant operation because errors cannot spread within a codeblock. This definition can be extended in a number of directions, particularly when $U$ is a multi-qubit gate. 
    In general, the broadest notion of transversality asks us to perform a non-trivial logical action via a constant-depth circuit.

    While transversal gates offer a  convenient and low-overhead approach to executing logical gates fault-tolerantly,
    unfortunately no quantum code can support a universal set of transversal gates~\cite{eastin2009restrictions}. 
    It is therefore natural to ask: \textit{how can we construct quantum codes with the most flexible and powerful transversal gates, for use in universal fault-tolerant computations?}

    Substantial effort has been invested in this problem.
    A pivotal line of research studies transversal gates on topological codes, notably the color codes on cellulated surfaces. 
    The two-dimensional color codes~\cite{Bombin2006topological}, of which the classic 7-qubit Steane code~\cite{calderbank1996good,steane1996multiple} is the smallest instance, support the entire set of Clifford gates transversally.
    Higher-dimensional color codes support transversal non-Clifford gates in higher levels of the Clifford hierarchy, notably the $T$ gate~\cite{bombin2015gauge,kubica2015universal} and the $\CCZ$ gate~\cite{Paetznick2013universal,Kubica2015unfolding}. A transversal $\CCZ$ gate has also been constructed in 3D surface codes~\cite{vasmer2019three}.
    These results led to fault-tolerant schemes which perform universal gates by code-switching between 2D and 3D topological codes~\cite{bombin2016dimensional,vasmer2019three}.
    
    Another line of research studies algebraic constructions of quantum codes with transversal non-Clifford gates for use in magic state distillation (MSD)~\cite{bravyi2005universal} protocols. 
    MSD is a central yet expensive technique that consumes low fidelity magic states to produce high fidelity magic states, which can be used in gate teleportation protocols~\cite{gottesman1999demonstrating} to perform universal fault-tolerant computation.
    By improving the error-correction parameters of the code used, we can obtain more efficient MSD schemes~\cite{bravyi2005universal,campbell2012magic,bravyi2012magic,krishna2019towards}. 
    In recent works, asymptotically good quantum codes supporting non-Clifford transversal gates were constructed~\cite{wills2024constant,golowich2024asymptoticallygoodquantumcodes,nguyen2024goodbinaryquantumcodes}, which implies constant-overhead MSD protocols for qubits~\cite{wills2024constant}.\footnote{Here, the meaning of the word ``overhead'' refers only to the ratio of inputted to outputted magic states, not to the spacetime-overhead of the scheme, as is relevant in a full fault-tolerant quantum computing protocol. We note that~\cite{nguyen2024quantum} constructed a MSD protocol with low spacetime overhead using algebraic codes.}
    
    Complementary to developments through algebraic methods, a novel line of recent research, inspired by the line of work which led to asymptotically good quantum low-density parity-check (qLDPC) codes~\cite{hastings2021fiber,panteleev2021quantum,breuckmann2021balanced,panteleev2022asymptotically,leverrier2022quantum,dinur2023good}, studies construction of qLDPC codes with transversal non-Clifford gates through homological methods~\cite{zhu2023non,scruby2024quantum,golowich2024quantum,lin2024transversal,hsin2024classifying,breuckmann2024cups,zhu2025topological}.
    The codes constructed in these works are not asymptotically good, but have constant or poly-logarithmic stabilizer check weight, 
    which is amenable to potential applications in fault-tolerant schemes.

    \paragraph{A new goal: Addressability.} 
    Almost all transversal gates constructed in aforementioned works act globally on the logical spaces: a constant depth physical circuit induces a non-trivial logical action on \emph{all or a large portion} of the logical qubits.
    For codes which encode a single logical qubit, such as the surface code or simple instances of 2D color codes, these gates are exactly what's needed to build fault-tolerant quantum computation protocols.
    For high-rate codes which encode many logical qubits, such control over the logical information is too coarse for applications outside of MSD;
    there are only limited situations in which it is algorithmically useful to execute the same gate on every qubit (except perhaps the Hadamard gate). 
    Instead, what we really want is for the code to support transversal gates that are also \emph{addressable}. 
    What this means is that, given any logical qubit indexed by $A$, there is some constant depth physical circuit, perhaps acting on a subset of the physical qubits, such that the induced logical gate acts only on the $A$'th logical qubit.
    Such fine-grained logical control is especially useful for implementation of generic algorithms, e.g. arithmetic circuits like the quantum adder~\cite{Gidney2018halvingcostof}.
    
    We observe that this is a common limitation that goes even beyond the works discussed. For high-rate quantum codes, notably quantum LDPC codes, it is often difficult to act on an individual logical qubit without acting on the rest of the logical space. 
    In existing fault-tolerant schemes based on constant-rate QLDPC codes~\cite{gottesman2013fault,fawzi2020constant,tamiya2024polylog,nguyen2024quantum}, logical gates are performed by full-block gate teleportation~\cite{gottesman1999demonstrating, Knill2005}. 
    For specific codes, existing studies on transversal or constant-depth gates, notably those based on automorphisms~\cite{calderbank1998quantum,grassl2013leveraging, bravyi2024high,eberhardt2024logical,sayginel2024fault} and other code symmetries~\cite{rengaswamy2019unifying,hu2021climbing,webster2023transversal,Quintavalle2023partitioningqubits,Breuckmann2024foldtransversal}, also primarily concern operations on the full logical space. 
    Two general techniques which inherently address individual logical qubits are pieceable fault-tolerance~\cite{yoder2016universal} and generalized code surgery~\cite{cohen2022low};\footnote{Generalized code surgery is a form of code deformation~\cite{bombin2009quantum,vuillot2022quantum}, which has been applied to design addressable logical action on high rate codes~\cite{breuckmann2017hyperbolic,krishna2021fault}.
    A different yet similar technique is homomorphic measurement~\cite{huang2023homomorphic}, which was utilized in a recent work~\cite{xu2024fast} to design addressable logical action on hypergraph product codes.
    As in the case of surgery, these proposals incur more overhead than constant depth unitary circuits, i.e., transversal circuits.} the overhead of the latter was significantly improved in recent works~\cite{cross2024improved,cowtan2024ssip,williamson2024low,ide2024fault,swaroop2024universal,zhang2024time}.
    Nevertheless, these techniques are still much more costly when compared to constant-depth circuits, i.e., transversal circuits.
    
    For these reasons, we believe the study and design of low-overhead, addressable logical operations on high-rate quantum codes is an interesting problem for both theoretical and practical purposes.
    Theoretically, recent works on constructing codes with transversal gates have opened many new possibilities for code and fault-tolerant scheme designs. 
    Practically, recent experimental progress has demonstrated sub-threshold scaling of the surface code~\cite{google2023suppressing,acharya2024quantumerrorcorrectionsurface} and logical operations on 2D color codes~\cite{rodriguez2024experimental,lacroix2024scaling}. 
    As we consider scaling to codes that encode a growing number of qubits in every block, notably quantum LDPC codes with promising finite-length performances~\cite{tillich2013quantum,leverrier2015quantum,breuckmann2021quantum,panteleev2021degenerate,bravyi2024high,xu2024constant,lin2024quantum}, addressable logic is not only a natural question to study, but also critical for practical efficiency.

    In this paper, we construct the first quantum codes to support transversal, addressable non-Clifford gates, such as the $\CCZ$ gate, and we obtain such codes that are near-asymptotically good (for the strongest sense of addressability) or asymptotically good (for weaker senses of addressability). As far as we know, transversal and addressable non-Clifford gates have not been established for any quantum code, except in a particular weak sense that we discuss below and in Section~\ref{sec:prelim_non_clifford_addressability}. While we construct codes that are (near) asymptotically good, our codes are not LDPC.\footnote{Note that transversal, addressable {\em Clifford} operations have been studied on LDPC codes in~\cite{zhu2023non,hsin2024classifying}.} We believe that the results demonstrated in the present paper can provide important ideas for similar developments on LDPC codes.

    \subsection{Our Results}

    \paragraph{Near-asymptotically good code with addressable non-Clifford gates.}
    
    Our primary result will be concerned with the strongest (and most computationally useful) form of addressability. Here, we will look to construct codes for which \textit{any} triple of logical qudits within one, two, or three blocks of the code may be addressed with the $\CCZ$ gate via a depth-one circuit of physical $\CCZ$ gates. We say that gates on codes satisfying this strong form of addressability are \emph{addressable}.
    \begin{theorem}\label{thm:main_result_1}
        There exists an explicit family of quantum CSS codes over qubits with parameters
        \begin{equation}
            \left[\left[n,\Omega\left(\frac{n}{\text{polylog}(n)}\right),\Omega\left(\frac{n}{\text{polylog}(n)}\right)\right]\right]_2
        \end{equation}
        supporting a transversal, addressable $\CCZ$ gate. Specifically, any three logical qubits across one, two, or three blocks of the code may be addressed with a logical $\CCZ$ gate via a depth-one circuit of physical $\CCZ$ gates.
    \end{theorem}
    We note that this code has asymptotically optimal error correction parameters up to polylogarithmic factors, and this statement will be proved using the results of Sections~\ref{sec:intrablock_GRS} and~\ref{sec:concatenation}. 
    We also remark that the same result may be obtained for any $\CLZ$ gate, not just for $\ell = 3$, and even for more complicated diagonal gates with entries $\pm 1$ on the diagonal, as we will describe. 
    For practical purposes, we note that a depth-1 physical circuit of $\CCZ$ gates can be naturally realized on platforms based on neutral atoms, 
    which have undergone exciting recent development~\cite{bluvstein2024logical}, by arranging atoms into the desired triples and applying a single laser pulse to induce the required multi-qubit interaction~\cite{levine2019parallel,evered2023high}.

    There is an important piece of folklore that a code supporting a transversal gate in the $\ell$-th level of the Clifford hierarchy supports a transversal and addressable gate in the $(\ell-1)$-th level of the Clifford hierarchy~\cite{hsin2024classifying,lin2024transversal,breuckmann2024cups}. This notion is indeed true for single-qubit gates, but in the case of multi-qubit gates such as $\CCZ$, the resulting form of addressability is limited to certain triples of qubits and is thus much weaker and less computationally useful than what we establish in Theorem~\ref{thm:main_result_1}. We discuss this in more detail in Section~\ref{sec:prelim_non_clifford_addressability}. 
    
    Given the stringent structure required on the codes to achieve the above result, we prove this using very specialized properties of certain algebraic codes, namely classical Reed-Solomon (RS) codes. Indeed, using the ``matrix-to-quantum code'' construction~\cite{bravyi2012magic,krishna2019towards,wills2024constant} to construct a quantum code from the generator matrix of a classical code, in this case an RS code, one obtains the above result for qudits of dimension $q=2^t$, as will be done in Section~\ref{sec:intrablock_GRS}. Then, using the qudit-to-qubit conversions of Section~\ref{sec:concatenation}, which are  an application of results in~\cite{nguyen2024goodbinaryquantumcodes,golowich2024asymptoticallygoodquantumcodes},\footnote{Note that the qudit-to-qubit conversions of~\cite{wills2024constant} are too weak for the present work, and we exclusively use those of~\cite{nguyen2024goodbinaryquantumcodes,golowich2024asymptoticallygoodquantumcodes}.} we obtain Theorem~\ref{thm:main_result_1} for qubits. It is conceptually important to note that, being made from Reed-Solomon codes, the qudit codes in this construction are asymptotically good, but only over qudits of growing size, i.e., $2^t = q \sim n$, where $n$ is the length of the code. This means the polylogarithmic factors are lost in the conversion to qubit codes in Section~\ref{sec:concatenation}. By contrast, constructions from algebraic geometry (AG) codes in this paper may be made asymptotically good over a constant field size, i.e., $q = 2^t = \Theta(1)$, meaning that the corresponding qubit codes are also asymptotically good. Obtaining the addressability of the $\CCZ$ gate with AG codes appears to be considerably more involved than with RS codes, however, this does appear to be a promising avenue to obtain Theorem~\ref{thm:main_result_1} with asymptotically good parameters; this direction is deferred to future work.

    \paragraph{A framework for transversal addressable gates.} 
    Much of the progress on constructing codes with transversal non-Clifford gates stems from the establishment of frameworks that reduce the task of constructing codes with non-Clifford gates to the task of constructing classical codes satisfying certain properties. The first such framework was the notion of \emph{triorthogonal matrices}, with the corresponding triorthogonal quantum codes supporting transversal $T$ gates up to Clifford correction~\cite{bravyi2012magic}. This framework has gone on to be highly influential in quantum coding theory~\cite{eczoo_quantum_triorthogonal}, and has been extended for prime-dimensional qudits~\cite{krishna2019towards} and for prime-power-dimensional qudits~\cite{wills2024constant}.\footnote{While the basis of knowledge on a quantum code's properties that are \textit{necessary} (not only sufficient) for supporting various transversal gates is incomplete, it is known~\cite{rengaswamy2020optimality} that triorthogonality is also necessary for a quantum code to support a transversal $T$ gate. We feel that such necessity statements on triorthogonality-like frameworks should be generalizable to other gates and other qudits.} 
    In Section~\ref{sec:addressable_orthogonality_big_section}, we develop a natural generalization of Bravyi and Haah's triorthogonality framework, and those that follows it, to addressable non-Clifford gates which we call \emph{addressable orthogonality} (see Section~\ref{sec:addressable_orthogonality} for a precise definition). One consequence of this framework is that it allows us to construct codes supporting the transversal addressability of a large family of diagonal gates. The first demonstration of the power and flexibility of this formalism comes in the form of our second main result, which is as follows.
    \begin{theorem}\label{thm:main_result_2}
        There exists an explicit family of (asymptotically good) qubit CSS codes with parameters
        \begin{equation}
            \left[\left[n,\Theta(n),\Theta(n)\right]\right]_2,
        \end{equation}
        supporting a addressable and transversal  $\CCZ$ gate on pre-designed (pairwise disjoint) intra-block triples of logical qubits. Specifically, the logical qubits in one block of the code are partitioned into sets of three, and one may address any such set with the $\CCZ$ gate via a depth-one circuit of physical $\mathsf{Z}$, $\mathsf{CZ}$ and $\mathsf{CCZ}$ gates.
    \end{theorem}
    We call this notion \emph{pre-designed addressability} since the code supports gate actions on arbitrary, but {\em pre-designed}, intra-block triples of qubits. This result is proved in Sections~\ref{sec:addressable_orthogonality_big_section},~\ref{sec:instantiations_orthogonality}, and~\ref{sec:concatenation} and instantiates part of our addressable orthogonality framework with algebraic geometry codes. As mentioned, since these codes may be instantiated to be asymptotically good over fixed-size qudits, i.e., over a fixed field of size $q=2^t$ with $t=O(1)$, the qubit codes resulting from the conversions in Section~\ref{sec:concatenation} are also asymptotically good.

    We provide two further demonstrations of the  flexibility of our addressable orthogonality formalism. 
    
    \begin{enumerate}
        \item In Appendix~\ref{sec:transversal_T}, we show that one of our addressable orthogonality notions can lead to the construction of quantum codes with transversal, addressable $T$ gates, up to the application of corrective Cliffords. Since this so naturally extends the original framework of Bravyi and Haah, we give it the name \textit{addressable triorthogonality}. However, we do not have an instantiation in mind. %

        \item In other parts of Section~\ref{sec:addressable_orthogonality_big_section} and~\ref{sec:instantiations_orthogonality}, we go beyond the $\CCZ$ gate and thoroughly explore other gates with entries $\pm 1$ on the diagonal in the Clifford hierarchy,\footnote{Note that the qualification ``in the Clifford hierarchy'' is technically redundant for such gates.} including, but not limited to $\CLZ$ gates for $\ell > 3$. 
        It is convenient to describe such gates using the Galois qudit formalism, where a $q=2^t$-dimensional qudit is described by computational basis states labeled by elements of the corresponding finite field $\left\{\ket{\eta}\right\}_{\eta \in \mathbb{F}_q}$. Note that, despite being larger qudits, they are still of interest, since a $q=2^t$-dimensional qudit is equivalent to a set of $t$ qubits, as will be shown in the preliminary material of Section~\ref{sec:qudit_to_qubit_prelim}.\footnote{Further relevant preliminaries are discussed in Sections~\ref{sec:prelim_qudit} and~\ref{sec:prelim_non_clifford_addressability}.} Then, the set of gates $U_{q,\ell}^\beta$, where $\ell \in \mathbb{N}$ and $\beta \in \mathbb{F}_q$ defined via
    \begin{equation}
        U_{q,\ell}^\beta\ket{\eta} = \exp\left(i\pi\tr(\beta\eta^\ell)\right)\ket{\eta},
    \end{equation}
    where $\tr:\mathbb{F}_q \to \mathbb{F}_2$ is the trace linear map (see Section~\ref{sec:prelim_qudit}), is an important set of gates since they can generate any diagonal gate with $\pm 1$ on the diagonal for $q=2^t$-dimensional qudits, and therefore also for qubits, as will be discussed in Appendix~\ref{sec:diagonal_Clifford}. 
    
    We will discuss how quantum codes may be designed supporting these gates (and their qubit versions) addressably, leading to quantum codes whose addressability properties may be highly tailored to particular situations.  
    We also show that the methods following from our addressable orthogonality framework are actually strictly stronger than what may be achieved using the folklore intuition described above, in that they may lead to strictly stronger parameters.
    In particular, one of the surprising results we obtain is that, by using Reed-Solomon codes, one can obtain as good asymptotic code parameters with transversal and addressable gates as with merely transversal gates; this result is captured by Theorem~\ref{thm:instantiation:GRS}.
    We will discuss this more in the introduction to Section~\ref{sec:instantiations_orthogonality}.
    \end{enumerate}

    \subsection{Discussion and Further Directions}

    In this work, we have initiated the study of quantum codes supporting transversal, addressable non-Clifford gates.
    This naturally leads to a large number of well-motivated research directions of theoretical and practical importance. %
    \paragraph{Open Problem 1:} {\em Does there exist an asymptotically good qubit code supporting a transversal, addressable $\CCZ$ gate}?
    As previously mentioned, a natural approach to this would be to leverage the internal structure of algebraic geometry codes in a similar way to how we have leveraged the internal structure of Reed-Solomon codes in this work. We look forward to pursuing this in future work. 
    \\
    
    The next natural question concerns the stabilizer weight of codes with transversal, addressable gates. 
    \paragraph{Open Problem 2:} {\em What are the best achievable asymptotic (and practical) parameters for a quantum LDPC code supporting a multi-qubit transversal non-Clifford gate, such as the $\CCZ$ gate? Can we achieve addressability in addition?} It is interesting to ask if the ideas on addressability contained herein apply to recent transversality constructions such as~\cite{golowich2024quantum}, or to any future constructions. Of course, one may ask the same questions for non-LDPC codes, where the interest would be in the actual constants of rate and relative distance.
    \\
    
    It is worth emphasizing that the recent works mentioned primarily focus on diagonal non-Clifford gates with entries $\pm 1$ on the diagonal. Quantum codes directly supporting other non-Clifford gates, such as the $T$ gate, transversally (addressably), or even transversally up to Clifford corrections as in the triorthogonal codes~\cite{bravyi2012magic} are still not known to exist with particularly good parameters.

    \paragraph{Open Problem 3:} {\em What are the best achievable asymptotic parameters for a quantum code supporting a single-qubit non-Clifford gate, such as the $T$ gate, transversally}? What if we require that the code is LDPC, the gate is addressable, or if we allow Clifford corrections?
    \\

    In this direction,
    we present the notion of addressable triorthogonality (from which one can obtain transversal, addressable $T$ gates up to Clifford corrections) in Appendix~\ref{sec:transversal_T}, which naturally generalizes Bravyi and Haah's original notion of triorthogonality~\cite{bravyi2012magic}.
    However,
    we do not have an instantiation. We hope that further work will help to address the following problem.
    
    \paragraph{Open Problem 4:} {\em Can we instantiate the addressable triorthogonality notion presented in Appendix~\ref{sec:transversal_T} to make quantum codes with addressable, transversal $T$ gates, up to Clifford corrections}? Is it possible to obtain well-performing small examples of such codes in a similar way to how~\cite{jain2024high} does with the usual triorthogonal (as well as doubly even and weak triple even) quantum codes?
    \\
    
    Besides code constructions,
    it would also be useful to understand and  characterize the necessary, as well as the sufficient conditions for quantum codes to support interesting addressable and transversal gates.
  
    \paragraph{Open Problem 5:} {\em Can we characterize necessary, as well as sufficient conditions for certain codes to support general (diagonal) gates in the Clifford hierarchy (addressably) transversally, beyond just the $T$ gate}~\cite{rengaswamy2020optimality,Hu2022designingquantum}? Could the addressable orthogonality framework in this paper offer insight to this?
    \\

    Furthermore, the main motivation behind studying powerful transversal gates is to lower the overhead of fault-tolerant quantum computation. 
    Given the addressable non-Clifford gates in our constructions, there are many potential approaches to realize universal fault-tolerance. 
    One idea is to use surgery techniques to measure Pauli observables on logical qubits. 
    These logical measurements can be used to implement Clifford gates directly on a code, or prepare stabilizer resource states which when teleported into a code translates to a Clifford circuit.
    In particular, the parallel surgery system of~\cite{cowtan2025parallel} and the extractor system of~\cite{he2025extractors}, which builds upon the works of~\cite{cross2024improved,williamson2024low,swaroop2024universal}, are both capable of preparing arbitrary Clifford resource states on arbitrary CSS codes. 
    This enables us to teleport a potentially deep and convoluted Clifford circuit into a code with addressable and low-overhead non-Clifford gates, completing a proposal for universal fault-tolerant computation. 
    We leave the cost analysis of this surgery approach, as well as the exploration of other approaches, as interesting open problems for future works.

    \paragraph{Open Problem 6:} {\em How can we utilize codes with transversal and addressable non-Clifford gates in a fault-tolerant, universal quantum computation scheme?}
    \\

    Finally, a question in a related but somewhat different direction: this paper focuses on implementing diagonal gates with entries $\pm 1$ on the diagonal (which are exactly all diagonal Clifford hierarchy gates with entries $\pm 1$ on the diagonal).
    Now, the diagonal Clifford hierarchy has been classified~\cite{cui2017diagonal} for prime-dimensional qudits, and given the equivalence between a $p^t$-dimensional qudit and a set of $t$ $p$-dimensional qudits, the diagonal Clifford hierarchy for prime-power-dimensional qudits has effectively also been classified. However, it would be very useful and interesting to further characterize the diagonal Clifford hierarchies for prime-power-dimensional qudits in terms of notation that is native to them. This would be of use in the present line of work~\cite{wills2024constant,golowich2024asymptoticallygoodquantumcodes,nguyen2024goodbinaryquantumcodes,golowich2024quantum} for which algebraic constructions for prime-power qudits based on finite fields enable the construction of qubit codes supporting transversal gates. For example, as we discuss in Appendix~\ref{sec:diagonal_Clifford}, the diagonal Clifford hierarchy for a $p$-dimensional qudit may be generated from gates of the form
    \begin{equation}
        \sum_{j \in \mathbb{Z}_p}\exp\left[\frac{2\pi i}{p^m}\text{poly}(j)\right]\ket{j}\bra{j}
    \end{equation}
    and the level of the Clifford hierarchy always increases with the degree of the polynomial~\cite{cui2017diagonal}. On the other hand, the diagonal Clifford hierarchy for $2^t$-dimensional qudits with $\pm 1$ on the diagonal may be generated by gates of the form
    \begin{equation}\label{eq:diagonal_classification}
        \sum_{\eta \in \mathbb{F}_q}\exp\left[i\pi\tr(\text{poly}(\eta))\right]\ket{\eta}\bra{\eta},
    \end{equation}
    (see Appendix~\ref{sec:diagonal_Clifford}). The essential point here is that the level of the Clifford hierarchy does not necessarily increase with the degree of the polynomial here, and behaves somewhat erratically. We believe that while the diagonal Clifford hierarchy has been classified, it is both theoretically interesting and important for the design of codes supporting interesting transversal and addressable gates to classify the prime-power-dimensional diagonal Clifford hierarchy in terms of notation akin to~\eqref{eq:diagonal_classification}.
    
    \paragraph{Open Problem 7:} {\em Can we understand the diagonal Clifford hierarchy of prime-power-dimensional qudits in a way that is more native to these qudits}?

    \subsection{Outline}
    
    Preliminary material will be presented in Section~\ref{sec:prelim}. The construction of a qudit code via an addressable $\CCZ$ gate via Reed-Solomon codes will be made in Section~\ref{sec:intrablock_GRS}. In Section~\ref{sec:addressable_orthogonality_big_section}, we will introduce the framework of addressable orthogonality and show how it can be used to construct codes with single-qudit addressable gates or single-index addressable multi-qudit gates. Instantiations of the addressable orthogonality framework will be presented in Section~\ref{sec:instantiations_orthogonality}. Finally, in Section~\ref{sec:concatenation}, the main results of the paper, which are the qubit code constructions, will be proved via certain qudit-to-qubit conversions.
    
    In Appendix~\ref{sec:transversal_T}, we discuss how the addressable orthogonality framework can encapsulate a transversal, addressable $T$ gate; a notion which we call addressable triorthogonality as it directly extends the original triorthogonality framework of Bravyi and Haah~\cite{bravyi2012magic}. In Appendix~\ref{sec:diagonal_Clifford}, we provide a discussion of the gates primarily considered in this paper in light of the diagonal Clifford hierarchy for qudits of prime-power dimension.

\section{Preliminaries}\label{sec:prelim}

\subsection{Finite Fields, Qudits and CSS Codes}\label{sec:prelim_qudit}

Given a prime power $q = p^t$, there is exactly one field of order $q$, up to isomorphism, denoted $\mathbb{F}_q$, whose set of non-zero elements in denoted $\mathbb{F}_q^*$. Throughout the paper, arithmetic will take place in the relevant finite field $\mathbb{F}_q$, unless otherwise specified. The field $\mathbb{F}_{q_1}$, where $q_1 = p_1^{t_1}$ contains a copy of $\mathbb{F}_{q_2}$, where $q_2 = p_2^{t_2}$, as a subfield if and only if $p_1 = p_2$ and $t_2$ divides $t_1$. In particular, $\mathbb{F}_q$ contains a copy of the field $\mathbb{F}_p$ as a subfield. 
There is a canonical $\mathbb{F}_p$-linear map $\tr_{\mathbb{F}_q/\mathbb{F}_p}:\mathbb{F}_q \to \mathbb{F}_p$ defined as
\begin{equation}
    \tr_{\mathbb{F}_q/\mathbb{F}_p}(\gamma) = \sum_{i=0}^{t-1}\gamma^{p^i}.
\end{equation}
It is a non-obvious yet standard fact that the right-hand side must be in $\mathbb{F}_p$~\cite{Lidl_Niederreiter_1996}.
This map is known as the \textit{trace} and may simply be denoted $\tr$ when the domain and range are clear.

Finite fields give us a convenient way to describe qudits of prime-power dimension, as well as their Pauli operations and Clifford hierarchies. Letting $q=p^t$, the state space of a $q$-dimensional qudit is $\CC^q$. Denoting computational basis states as $\{\ket{\eta}:\eta \in \mathbb{F}_q\}$,
we define Pauli operators $X^\beta$ and $Z^\beta$ for every $\beta \in \mathbb{F}_q$ via
\begin{align}
    X^\beta \ket{\eta} = \ket{\eta + \beta} \hspace{.1in} \mbox{and} \hspace{.1in} Z^\beta = \omega^{\tr(\beta\eta)}\ket{\eta},
\end{align}
where $\omega$ is a $p$-th root of unity. In these equations, arithmetic takes place over $\FF_q$. One may check the commutation relation
\begin{equation}
    X^\beta Z^\gamma = \omega^{\tr(\beta\gamma)} Z^\gamma X^\beta.
\end{equation} 
The Pauli group over one qudit is defined as 
\begin{align}
    \cP_{1,q} = \{\omega^a X^\beta Z^\gamma, a\in [p], \beta,\gamma\in\FF_q\}
\end{align}
and the $n$-qudit Pauli group is defined as $\cP_{n,q} = \cP_{1,q}^{\otimes n}$. The Clifford hierarchy was first introduced in~\cite{gottesman1999demonstrating} and has an important role in the field of fault-tolerant quantum computation, being intimately related to the fault-tolerant implementability and distillability of certain gates. It is constructed as follows.
\begin{align}
    \mathcal{C}_{n,q}^{(1)} &= \mathcal{P}_{n,q}\\
    \cC_{n,q}^{(k)} &= \{U\in U((C^q)^{\otimes n}): UPU^{\dagger}\in \cP_{n,q} \text{ for all } P\in \mathcal{C}_{n,q}^{(k-1)}\}\text{ for }k > 1.
\end{align}
Notice in particular that $\mathcal{C}_{n,q}^{(2)}$ denotes the Clifford group, and may be written more briefly as simply $\mathcal{C}_{n,q}$.

The Calderbank-Shor-Steane (CSS) code formalism extends to qudits as well. 
For a vector $\Gamma\in \FF_q^n$, we denote
$X^\Gamma := \bigotimes_{i=1}^n X^{\Gamma_i}$ and $Z^\Gamma := \bigotimes_{i=1}^n Z^{\Gamma_i}$.
\begin{definition}[CSS Codes]
Given two classical $\mathbb{F}_q$-linear codes $\cL_X, \cL_Z\subseteq \FF_q^n$ satisfying $\cL_X\subseteq \cL_Z^\perp$, we define the quantum code space $\CSS(X, \cL_X; Z, \cL_Z) \subseteq (\CC^q)^{\otimes n}$ as
\begin{align}
    \qcode = \CSS(X, \cL_X; Z, \cL_Z)
    = \{\ket{\psi}\in (\CC^q)^{\otimes n}: X^{\bx}\ket{\psi} = \ket{\psi}, Z^{\bz}\ket{\psi} = \ket{\psi} 
    \text{ for all } \bx\in \cL_X, \bz\in \cL_Z\}.
\end{align}
\end{definition}
We cite some standard facts about CSS codes, and refer readers to proofs in~\cite{gottesman2016surviving}.
The CSS codespace $\CSS(X, \cL_X; Z, \cL_Z)$ has dimension $k = n - \dim(\cL_X) - \dim(\cL_Z)$, and it encodes $k$ qudits into $n$ qudits, i.e., it identifies the state space of $k$ qudits, $(\mathbb{C}^q)^{\otimes k}$, with some subspace of $(\mathbb{C}^q)^{\otimes n}$, called the codepsace. Given the stabilisers specified by $\mathcal{L}_X$ and $\mathcal{L}_Z$, the logical $X$ operators are established by making some choice of basis $\{g^1, \cdots, g^k\}\subseteq \FF_q^n$ for the space $\cL_Z^\perp/\cL_X$.\footnote{Normally, when we say that the set $\{g^1, \ldots, g^k\}$ forms a basis for a subspace, we mean that the vectors are in that subspace, they span the space, and no non-trivial linear combination of them sums to zero, i.e., they are linearly independent. Here, saying $\{g^1, \ldots, g^k\}$ form a basis for the space $\mathcal{L}_Z^\perp/\mathcal{L}_X$ means that these vectors are in $\mathcal{L}_Z^\perp$, any coset in this quotient space may be written as a linear combination of cosets of the form $g^i + \mathcal{L}_X$, and that no non-trivial linear combination of the $g^i$ sums to an element of $\mathcal{L}_X$, i.e., no non-trivial linear combination of cosets of the form $g^i+\mathcal{L}_X$ sums to the trivial coset.} 
Then, given a vector $u\in \FF_q^k$, the logical encoding of the $k$-qudit computation basis state labelled by $u$ in the code $\qcode$ is, up to normalization,
\begin{align}
    \ol{\ket{u}} \propto \sum_{g\in \cL_X} \Ket{\sum_{a=1}^k u_ag^a + g}. 
\end{align}
From here, it can be checked that $X^{\beta g^a}$ is the logical operator $X^\beta$ acting on the $a$'th logical qudit.
We can also find $\{ h^b\in \cL_X^\perp\setminus \cL_Z \}_{b \in [k]}$ such that $g^a\cdot h^b = \one_{a=b}$, in which case
the $Z^\gamma$ logical operator of the $a$-th logical qubit would be $Z^{\gamma h^a}$. 
We can check that these operators have the desired commutation relations.
\begin{align}
    X^{\beta g^a}Z^{\gamma h^b} 
    &= \omega^{\sum_{i=1}^n\tr(\beta\gamma g^a_ih^b_i)}
    Z^{\gamma h^b}X^{\beta g^a} \\ 
    &= \omega^{\tr(\sum_{i=1}^n \beta\gamma g^a_ih^b_i)}
    Z^{\gamma h^b}X^{\beta g^a}\label{eq:XZcommutation}  \\ 
    &= \begin{cases}
        \omega^{\tr(\beta\gamma)} Z^{\gamma h^b}X^{\beta g^a} & a = b \\
        Z^{\gamma h^b}X^{\beta g^a} & a \not= b.
    \end{cases} 
\end{align}
We note explicitly that the summation can be brought inside the trace in equation~\eqref{eq:XZcommutation} because $\tr(\cdot)$ is a $\FF_p$-linear function and $\omega$ is a $p$-th root of unity. 

For a vector $g\in \FF_q^n$, we denote its Hamming weight as $|g|$. 
For a classical linear code $\cL$, we let $d(\cL)$ denote its distance, which is 
$d(\cL) = \min_{g\in \cL\setminus\{0\}} |g|$.

\begin{definition}
For the quantum CSS code $\qcode = \CSS(X, \cL_X; Z, \cL_Z)$, we define its $X$ distance to be $d_X = d(\cL_Z^\perp/\cL_X)$,\footnote{Technically, the quotient space $\cL_Z^\perp/\cL_X$ contains elements of the form $g + \cL_X$, and we say that $|g + \cL_X| = \min_{g'\in g + \cL_X} |g'|$.} and its $Z$ distance to be $d_Z = d(\cL_X^\perp/\cL_Z)$. 
The distance of $\qcode$ is $d(\qcode) = \min(d_X, d_Z)$. 
\end{definition}
\begin{corollary}~\label{cor:CSS_distance}
    Note that for a CSS code, $d_X = d(\cL_Z^\perp/\cL_X)\ge d(\cL_Z^\perp)$, and similarly $d_Z\ge d(\cL_X^\perp)$. 
    Therefore $d(\qcode)\ge \min(d(\cL_X^\perp), d(\cL_Z^\perp))$.
\end{corollary}
We say that $\qcode$ is a $[[n, k, d]]_q$ quantum code when it encodes $k$ logical qudits (of dimension $q$) into $n$ physical qudits, and has distance $d$.
A family of quantum codes with a divergent number of physical qudits $n$ is called \textit{asymptotically good} if $k, d = \Theta(n)$ as $n\rightarrow \infty$. 

\subsection{Non-Clifford Gates and Addressability}\label{sec:prelim_non_clifford_addressability}

In this work, we consider the case $p = 2$ (although our results naturally generalise to larger primes $p$) and study diagonal non-Clifford gates over qubits and qudits of dimension $q=2^t$. Given that qubits are of much greater practical relevance than larger qudits, we find it necessary to emphasise here that $2^t$-dimensional qudits are still practically interesting because such a qudit is equivalent to a set of $t$ qubits, as will be shown in Section~\ref{sec:qudit_to_qubit_prelim}. Packaging multiple qubits into a single prime-power qudit in this way is a convenient means to describe Clifford hierarchy gates on multiple qubits.  In addition, by building quantum codes over qudits, we will convert them to qubit codes in Section~\ref{sec:concatenation}.

The first and most important non-Clifford gate for our purposes is the $\CCZ$ gate which acts on 3 qudits of dimension $q$. This is defined via 
\begin{align}
    \CCZ_q \ket{\eta_1}\ket{\eta_2}\ket{\eta_3} 
    &= \exp(i\pi\tr(\eta_1\eta_2\eta_3))\ket{\eta_1}\ket{\eta_2}\ket{\eta_3},
\end{align}
or, more generally, given $\beta \in \mathbb{F}_q$,
\begin{align}    
    \CCZ_q^\beta \ket{\eta_1}\ket{\eta_2}\ket{\eta_3} 
    &= \exp(i\pi\tr(\beta\eta_1\eta_2\eta_3))\ket{\eta_1}\ket{\eta_2}\ket{\eta_3}.
\end{align}
We note that, here and further on, while the trace is a canonical linear map from $\mathbb{F}_q$ to $\mathbb{F}_2$, no generality is lost by using it exclusively, because any linear map $\theta : \mathbb{F}_q \to \mathbb{F}_2$ may be written as $\theta(\eta) = \tr(\beta\eta)$ for some $\beta \in \mathbb{F}_q$~\cite{Lidl_Niederreiter_1996}.
We also consider the more general $\CLZ$ gates, which act on $l$ qudits as
\begin{align}
    \CLZ_q \ket{\eta_1}\ket{\eta_2}\cdots\ket{\eta_\ell} 
    &= \exp(i\pi\tr(\eta_1\eta_2\cdots \eta_\ell))\ket{\eta_1}\ket{\eta_2}\cdots\ket{\eta_\ell} , \\
    \CLZ_q^\beta \ket{\eta_1}\ket{\eta_2}\cdots\ket{\eta_\ell} 
    &= \exp(i\pi\tr(\beta\eta_1\eta_2\cdots \eta_\ell))\ket{\eta_1}\ket{\eta_2}\cdots\ket{\eta_\ell}.
\end{align}
It may easily be checked that $\CLZ_q^\beta$ is in exactly the $\ell$-th level of the Clifford hierarchy, by which we mean that it is in the $\ell$-th level, but not the $(\ell-1)$-th level, for all $\beta \neq 0$. In particular, the gate $\CCZ_q$ is in exactly the third level.

Our second class of gates are the following single-qudit gates:
\begin{align}\label{eq:Uql_defn}
    U_{q,\ell}\ket{\eta} &= \exp(i\pi\tr(\eta^\ell))\ket{\eta}, \\
    U_{q,\ell}^\beta\ket{\eta} &= \exp(i\pi\tr(\beta\eta^\ell))\ket{\eta}.
\end{align}
We note that when $q = 2^t$ and $t\ge 3$, the gate $U_{q,7}$ is in exactly the third level of the Clifford hierarchy; see Lemma~3.1.1 of~\cite{wills2024constant} for a proof.\footnote{In fact,~\cite{wills2024constant} only provides a proof for $t \ge 5$ where the proof is most convenient, but the same fact may be checked for $t=3,4$ also.}

The third type of gates we consider is a generalization of the above gates to arbitrary polynomials.
Let $P$ be a degree-$\ell$ polynomial over $\mathbb{F}_q$ in $e$ variables.

We define the $e$-qudit gate $U_{q,P}$ to be
\begin{align}
    U_{q,P}\ket{\eta_1}\cdots \ket{\eta_e}
    &= \exp(i\pi\tr(P(\eta_1, \cdots, \eta_e))) \ket{\eta_1}\cdots \ket{\eta_e},
\end{align}
and, even though is is technically redundant,
\begin{align}               
    U_{q,P}^\beta\ket{\eta_1}\cdots \ket{\eta_e}
    &= \exp(i\pi\tr( \beta P(\eta_1, \cdots, \eta_e))) \ket{\eta_1}\cdots \ket{\eta_e}.
\end{align}
When the dimension of the qudits is clear, we will often abbreviate the subscript $q$ from our gates. Appendix~\ref{sec:diagonal_Clifford} contains a brief discussion on the significance of the gates $U_{q,P}$ in the context of the diagonal Clifford hierarchy.

Continuing with defining notation, given a large number of qudits $n$, it will be helpful to have a way to denote gates acting on particular qudits. Given a single qudit gate such as $U_{q,\ell}$ and some physical qudit label $A \in [n]$, $U_{q,\ell}[A]$ will denote the gate $U_{q,\ell}$ acting on the $A$'th (physical) qudit. Furthermore, it is convenient to be able to denote the action of a single-qudit gate on multiple qudits at once. For example, given a vector $\Gamma\in \FF_q^n$, we denote
\begin{equation}
    U_{q,\ell}^\Gamma = \bigotimes_{i=1}^n U_{q,\ell}^{\Gamma_i}[i],
\end{equation}
so that, for $h \in \mathbb{F}_q^n$,
\begin{equation}
    U_{q,\ell}^\Gamma\ket{h} = \prod_{i=1}^n \exp(i\pi\tr(\Gamma_i h_i^\ell))\ket{h}.
\end{equation}

We adopt similar notation for multi-qudit gates such as $\CLZ_q$, where given $A_1, \ldots, A_\ell \in [n]$, $\CLZ_q[A_1, \ldots, A_\ell]$ denotes the gate $\CLZ_q$ acting on the qudits labelled by $A_1, \ldots, A_\ell$.

The primary goal of this paper is to construct quantum codes which allow one to act with the non-Clifford gates above on a particular logical qudit (in the case of single-qudit gates) or on any set of logical qudits (in the case of multi-qudit gates). In such a case, we say that we can \textit{address} this logical qudit, or set of logical qudits. Moreover, we wish to be able to do so by acting transversally on the physical qudits, by which we mean that the implementing circuit is of constant depth (in the strongest form of transversality, the implementation is possible with depth one). In this case, we say that the code supports the gate transversally and addressably. Acting with different transversal operations on the physical qudits will allows us to address different logical qudits (in the case of single-qudit gates), or sets of logical qudits (in the case of multi-qudit gates).

Notationally, as is standard, we use an overline to indicate a logical (encoded) gate. For example, given a logical qudit labelled by $A \in [k]$, $\overline{U_{q,\ell}[A]}$ will denote the logical $U_{q,\ell}$ gate acting on the $A$-th logical qudit. Similarly, given three logical qudits labelled by $A,B,C \in [k]$, $\overline{\CCZ[A,B,C]}$ denotes the logical $\CCZ$ gate acting on these three logical qudits. Furthermore, given multiple code blocks, we may wish to execute a multi-qudit gate between logical qudits in different blocks (a so-called \textit{interblock} gate). Given logical qudits $A,B,C \in [k]$, and code blocks labelled by $i,j,k$, $\overline{\CCZ_{ijk}[A,B,C]}$ denotes the logical $\CCZ$ gate acting on the $A$'th, $B$'th, and $C$'th logical qudits in the code blocks $i$, $j$ and $k$ respectively. Such notation naturally exists also for physical, rather than logical gates, and it also extends to gates on more than three qudits such as $\CLZ$ for $l > 3$. Note that the code block labels $i,j,k$ need not be pairwise distinct. Indeed, in the case that $i=j=k$, the gate is called an \textit{intrablock} gate, and the notation reduces to $\overline{CCZ[A,B,C]}$ when the code block is clear.

Certain facts are already known about transversal, addressable non-Clifford gates, primarily as folklore in the community, or as alluded to in~\cite{lin2024transversal, hsin2024classifying,breuckmann2024cups}. The main one is that a code supporting a transversal gate in the $t$-th level of the Clifford hierarchy implies that it supports a transversal and addressable gate in the $(t-1)$-th level of the Clifford hierarchy. 
As we discuss below, this fact implies the correct sense of addressability for single-qudit gates, but a much weaker sense of addressability for multi-qudit gates in comparison to our main result, Theorem~\ref{thm:main_result_1}.
Let us enunciate this in the case of $\CCCZ$ and $\CCZ$ now.

Consider an $[[n,k,d]]$ quantum code that supports a transversal (but not addressable) $\CCCZ$ gate. This means the following. Suppose one has four code blocks of the code and one acts with $\CCCZ$ across the first physical qubits, the second physical qubits, the third physical qubits, and so on, meaning one acts with the physical gate
\begin{equation}\label{eq:CCCZ_fixed_interblock}
    \CCCZ_{1234}^{\otimes n} \coloneq \prod_{i=1}^n\CCCZ_{1234}[i],
\end{equation}
where
\begin{equation}
    \CCCZ_{1234}[i]\coloneq \CCCZ_{1234}[i,i,i,i].
\end{equation}
Then, given any code state of the four code blocks, the logical gate
\begin{equation}\label{eq:CCCZ_fixed_interblock_logical}
    \overline{\CCCZ_{1234}^{\otimes k}} \coloneq \prod_{A=1}^k\overline{\CCCZ_{1234}[A]}
\end{equation}
will be executed. In words, acting transversally across the same physical qubits in each code block executes the logical gate transversally across the same logicla qubits in each code block. This is the most common sense of transversality for multi-qudit gates and, in particular, is the sense in which recent works~\cite{golowich2024asymptoticallygoodquantumcodes,nguyen2024goodbinaryquantumcodes,golowich2024quantum} support transversal $\CLZ$ gates.

Now, noting that logical Pauli operators such as $\overline{X[A]}$, for $A \in [k]$, are always transversal, suppose that one acted transversally via
\begin{equation}\label{eq:CCZ_fixed_interblock_first}
    \overline{X_1[A]}\CCCZ_{1234}^{\otimes n}\overline{X_1[A]}\CCCZ_{1234}^{\otimes n},
\end{equation}
which reduces to a depth-one circuit of $\CCZ$ gates.
One can check that this combination will execute exactly
\begin{equation}
    \overline{\CCZ_{234}[A]}.
\end{equation}
Therefore, a code supporting a transversal $\CCCZ$ gate does indeed allow for addressable $\CCZ$, but only in this weak sense. In addressing the logical qudits with $\CCZ$, we are restricted to addressing only triples of logical qudits $[A,A,A]$ in three different code blocks, for $A \in [k]$. 
We call this behaviour where we can address only triples of logical qudits $[A,A,A]$ \textit{single-index addressability}, reserving the term \textit{addressability} for what we really want --- which is to be able to act on \textit{any} triple of logical qudits, in one, two, or three different code blocks, with some transversal operation.

The above thoughts easily extend to other diagonal Clifford hierarchy gates. If $U$ is a diagonal gate in exactly the $\ell$-th level of the Clifford hierarchy acting on $s$ qudits, then there must be some $i \in [s]$ such that $X[i]UX[i]U^\dagger = V$, where $V$ is a diagonal gate in exactly the $(\ell-1)$-th level of the Clifford hierarchy. This means that any code supporting a transversal gate in exactly the $\ell$-th level of the Clifford hierarchy (in the usual sense of Equations \eqref{eq:CCCZ_fixed_interblock} to \eqref{eq:CCCZ_fixed_interblock_logical}) supports a \textit{single-index addressable} gate in exactly the $(\ell-1)$-th level of the Clifford hierarchy. Of course, single-index addressability is the same thing as addressability for single-qudit gates, but we want more for multi-qudit gates like $\CCZ$. The main result of this work is the first construction of quantum codes (with non-trivial parameters) with addressable $\CCZ$ gates, with an easy generalisation to other multi-qudit gates. This result is established in Section~\ref{sec:intrablock_GRS} for qudits, and we convert these into qubit codes to obtain our main result in Section~\ref{sec:concatenation}.

\subsection{Converting Prime-Power Qudits to Qubits}\label{sec:qudit_to_qubit_prelim}

In this section, we will justify the statement that a $q$-dimensional qudit, where $q=2^t$, is really the same as a set of $t$ qubits~\cite{gottesman2016surviving}. To do this, one constructs an isomorphism $\psi$ between the states of the former object, which lie in $\mathbb{C}^q$, and the states of the latter object, which lie in $(\mathbb{C}^2)^{\otimes t}$. We will also obtain an isomorphism $\theta$ between the unitaries acting on the former space, and those on the latter space, which is compatible with the isomorphism $\psi$ in the sense that
\begin{equation}
    \theta(U)\left(\psi(\ket{v}\right) = \psi\left(U\ket{v}\right), \; \forall U \in U(\mathbb{C}^q) \text{ and } \ket{v} \in \mathbb{C}^q.
\end{equation}
Moreover, because $\theta$ specialises to an isomorphism between every level of the Clifford hierarchy, the $q$-dimensional qudit and the set of $t$ qubits really may be treated as the same thing. Recalling the computational basis and Pauli group for a qudit of dimension $q=2^t$ (see Section \ref{sec:prelim_qudit}), in particular that the computational basis states are labelled by the elements of the finite field $\mathbb{F}_q$, consider the following. One may view $\FF_q$, where $q = 2^t$, as a vector space over $\FF_2$. Naturally, then a basis for $\FF_q$ over $\FF_2$ is a collection of elements $\al_1,\cdots, \al_t\in \FF_q$ such that any element $\gamma\in \FF_q$ has an unique decomposition as 
\begin{align}
    \gamma = \sum_{i=1}^t s_i \al_i,\;\; s_i\in \FF_2.
\end{align}
We can write this decomposition as a bijection $\B:\FF_q\rightarrow \FF_2^t$.

\begin{definition}[Self-Dual Basis]
    A basis $\{\al_1, \cdots, \al_t\}$ is called self-dual if for all $i,j$,
    \begin{align}
        \tr(\al_i\al_j) = \delta_{ij}.
    \end{align}
\end{definition}
It is known that self-dual bases exist for all $\FF_{2^t}$ over $\mathbb{F}_2$~\cite{seroussi1980factorization}. Choosing a self-dual basis $\{\al_i\}$, fix $\B$ to be its decomposition map as above. 
For a vector $g\in \FF_q^n$, it is convenient to write $\B(g) = (\B(g_1), \cdots, \B(g_n))$. 
For a set of vectors $C\subseteq \FF_q^n$, we write $\B(C) = \{\B(v):v\in C\}\subseteq (\FF_2^t)^n$.

The isomorphism $\psi$ between $\mathbb{C}^q$ and $(\mathbb{C}^2)^{\otimes t}$ is then simply defined via a map between their computational bases which extends by linearity to the whole space.
Indeed,
\begin{align}
     \psi\ket{\eta} \coloneq \ket{\B(\eta)}, \forall \eta\in \FF_q.
\end{align}
Given this association between the computational basis states of $\mathbb{C}^q$ and $(\mathbb{C}^2)^{\otimes t}$, the map $\theta$ from unitaries on the former space to those on the latter space is defined, in words, as the map taking a matrix (with respect to the computational basis $(\ket{\eta})_{\eta \in \mathbb{F}_q}$) to the same matrix (with respect to the computational basis $(\ket{\B(\eta)})_{\eta \in \FF_q}$). Symbolically, $\theta$ is defined as
\begin{align}
    \theta : U(\CC_q) &\to U\left((\CC^2)^{\otimes t}\right)\\
    \theta(U)\left(\ket{v}\right) &= \psi U \psi^{-1}\left(\ket{v}\right)
\end{align}
which is easily seen to be an isomorphism. The map $\theta$ turns out to operate on Paulis in the following convenietly written way:
\begin{align}
    X^\beta \mapsto X^{\B(\beta)}, \; Z^\gamma \mapsto Z^{\B(\gamma)},
\end{align}
where we note that a single-qudit Pauli is mapped to a $t$-qubit Pauli. 
Importantly, this map preserves commutation relations, because
\begin{align}
    X^{\B(\beta)}Z^{\B(\gamma)}
    &= (-1)^{\B(\beta)\cdot \B(\gamma)}Z^{\B(\gamma)}X^{\B(\beta)}
    = (-1)^{\tr(\beta\gamma)}Z^{\B(\gamma)}X^{\B(\beta)}.
\end{align}
$\theta$ therefore restricts to an isomorphism between the Pauli groups on the respective spaces. Furthermore, it is possible to show that $\theta$ restricts to a bijection between every level of the Clifford hierarchy.

An important step in Section~\ref{sec:concatenation}, when we convert codes over qudits to codes over qubits, will be the following lemma. 
\begin{lemma}\label{lem:self-dual-embedding}
Given a quantum qudit CSS code $\qcode_0 = \CSS(X,\cL_X;Z, \cL_Z)$ with parameters $[[n_0,k_0,d_0]]_q$, we can obtain a qubit CSS code $\qcode_1$ as follows. $\qcode_1$ may be defined equivalently as either the CSS code $\CSS(X, \B(\cL_X); Z, \B(\cL_Z))$, or as the image of the code $\qcode_0$ under $\psi^{\otimes n_0}$.
\begin{itemize}
    \item If $\qcode_0$ encodes the standard basis state $\ket{u}, u\in \FF_q^{k_0}$, as $\ol{\ket{u}} \in (\mathbb{C}_q)^{n_0}$, then $\qcode_1$ encodes $\ket{\B(u)}$ as $\psi^{\otimes n_0}\left(\overline{\ket{u}}\right)$. $\qcode_1$ has parameters $[[n_1 = n_0t, k_1 = k_0t, d_1\ge d_0]]_2$. In particular, every logical Pauli operator of $\qcode_1$ must be supported on at least $d_0$ many distinct blocks of $t$-qubits.
    \item The set of stabilisers for $\qcode_1$ may be seen as the set of stabilisers of $\qcode_0$ operated on by $\theta^{\otimes n_0}$. Consider some $\mathbb{F}_q$-basis for $\mathcal{L}_X$, call it $\hat{\mathcal{L}}_X$. Then, a natural choice of generating set for the $X$-stabilisers of $\qcode_1$ is given by
    \begin{equation}
        \left\{X^{\mathcal{B}(v)}: v \in \tilde{\mathcal{L}}_X\right\}, \text{ where } \tilde{\mathcal{L}}_X \coloneq \left\{\alpha_i\cdot w, \text{ for }i \in [t] \text{ and } w \in \hat{\mathcal{L}}_X\right\}.
    \end{equation}
    The same statement holds for the $Z$-type stabilisers.
    \item Consider the $X$-type logical operator $\overline{X^\beta_a}$ of $\qcode_0$, which is the logical operator $X^\beta$ acting on the $a$-th logical qudit, for $\beta \in \mathbb{F}_q$ and $a \in [k_0]$. Then $\theta^{\otimes n_0}\left(\overline{X^\beta_a}\right)$ forms a natural choice for the logical operator $\overline{X^{\B(\beta)}_a}$, which is the $X$-operator $X^{\B(\beta)}$ acting on the $a$-th block of logical qubits in $\qcode_0$. The same statement holds for the $Z$-type logical operators.
\end{itemize}
\end{lemma}
We refer readers to Sections~2.2 and~4.3 of~\cite{wills2024constant} for detailed proofs and discussions of this Lemma.

\subsection{Quantum Codes from Punctured and Shortened Classical Codes}\label{sec:qcode_from_puncture}

The introduction of triorthogonal quantum codes in~\cite{bravyi2012magic} gave a unified framework for building quantum codes supporting transversal $T$-gates, called the triorthogonal framework, where quantum codes were defined out of a matrix with nice properties, called a ``triorthogonal matrix''. Later,~\cite{hastings2018distillation} became the first work to establish the possibility of magic state distillation with a sub-logarithmic overhead, meaning that the magic state distillation exponent $\gamma$\footnote{This means that the protocol's ratio of noisy inputted magic states to output magic states of error rate below $\epsilon$ is $\mathcal{O}(\log^\gamma(1/\epsilon))$.} could be made smaller than 1.
This was done via an instantiation of Bravyi and Haah's framework to build a quantum code supporting a transversal $T$ gate with good enough parameters that sub-logarithmic overhead distillation became possible. The triorthogonal matrix in question was constructed as a generator matrix of a punctured classical code that had certain nice properties (multiplication properties, as we will describe) to make the triorthogonality possible. This process has been expanded upon~\cite{krishna2019towards,wills2024constant,nguyen2024quantum} and is central to the present work.\footnote{The reader should not have in mind that these frameworks for constructing quantum codes with transversal diagonal Clifford hierarchy gates are in some way \textit{ad hoc}. Indeed, it was proved~\cite{rengaswamy2020optimality} that any quantum code supporting a transversal $T$ gate is a triorthogonal code. In turn, it is easily seen that any triorthogonal matrix may be realised as the generator matrix of a punctured multiplication-friendly classical code. While the triorthogonality framework has been generalised~\cite{krishna2018magic,wills2024constant,nguyen2024quantum} for the construction of quantum codes supporting other diagonal Clifford hierarchy gates, similar statements saying that quantum codes supporting such gates must be of this form have not been proved. Nevertheless, we believe similar statements on the necessity of the generalised triorthogonality frameworks for these purposes should hold.} Therefore, we present the necessary preliminaries on puncturing and shortening classical codes now.

\begin{definition}[Shortened and Punctured Codes]
    For a set of $t$ indices $T\subseteq [n]$, and a vector $c\in \FF_q^n$,
    we let $c|_T \in \FF_q^{t}$ denote the restriction of $c$ to indices in $T$, and $\puncT{c} \in \FF_q^{n-t}$ denote the vector $c$ with the indices in $T$ removed. 
    For a classical code $C\subseteq \FF_q^n$, 
    \begin{enumerate}[itemsep = 0pt]
        \item We define the $T$-punctured code $\puncT{C}$ to be the linear subspace $\{\puncT{c}: c\in C\}\subseteq \FF_q^{n-t}$.
        \item We define the $T$-shortened code $\shrT{C}$ to be the linear subspace $\{\puncT{c}: c\in C, c|_T = \mathbf{0}_t\}\subseteq \FF_q^{n}$, where $\mathbf{0}_t$ is the all zero vector of length $t$.
    \end{enumerate}  
\end{definition}
We refer readers to Chapter~1 of~\cite{Fundamentals_of_ECC} for a detailed introduction to these topics, as well as a proof of the following lemma.
\begin{lemma}[Theorem~1.5.7 of~\cite{Fundamentals_of_ECC}]
\label{lem:short_punc}
    For a set of $t$ indices $T\subseteq [n]$, and a $[n, k, d]$ classical code $C\subseteq \FF_q^n$, we have
    \begin{enumerate}[itemsep = 0pt]
        \item $\shrT{(C^{\perp})} = (\puncT{C})^\perp$, and $\puncT{(C^\perp)} = (\shrT{C})^\perp$. 
        \item If $t < d$, then $\puncT{C}$ has dimension $k$, and $\shrT{(C^\perp)}$ has dimension $n-t-k$. 
    \end{enumerate}  
\end{lemma}

\begin{fact}\label{fact:punctured_distance}
    $\shrT{C}$ has distance $d$, $\puncT{C}$ has distance at least $d-t$. 
\end{fact}

We may now describe how punctured and shortened classical codes may be used to construct quantum CSS codes. Consider an $[\N, \K, D]$ classical code $C$.

Fixing $k < \K$, take a generator matrix $\tilde G$ of $C$ with the following form,
\begin{align}
    \tilde{G} = \begin{pmatrix}
        I_k & G_1 \\
        0 & G_0
    \end{pmatrix}
\end{align}
where $I_k$ is the $k\times k$ identity matrix, $0$ is the $(\K-k) \times k$ all $0$'s matrix, and $G_1$ and $G_0$ are $k \times n$ and $(\K-k) \times n$ matrices respectively, for $n = N-k$. 
Let 
\begin{align}
    G = \begin{pmatrix}
        G_1 \\ G_0
    \end{pmatrix}
\end{align}
be the last $n$ columns of $\tilde{G}$. 
We denote the $a$'th row of $G$ by $g^a \in \bbF_q^n$. 
Let $\cG_0, \cG_1$ and $\cG$ be the row spaces of $G_0, G_1$ and $G$ respectively (over $\mathbb{F}_q$). 
We will define our qudit code to be 
\begin{align}\label{eq:basic_matrix_to_qcode_defn}
    \qcode = \CSS(X, \cG_0; Z, \cG^\perp),
\end{align}
where we note that $\mathcal{G}_0 \subseteq \mathcal{G} = (\mathcal{G}^\perp)^\perp$, so the definition makes sense. Note that in many of the constructions that follow, we will simply start from a certain matrix $G$, and talk about quantum code $\qcode$ associated to $G$, i.e., we may bypass the step of defining and puncturing $\tilde{G}$, although it is shown here for completeness. Note that the quantum code $\qcode$ associated to $G$ will always be given by Equation~\eqref{eq:basic_matrix_to_qcode_defn}.

While additional assumptions on $G$ will be needed to make $\qcode$ support interesting transversal, addressable gates, the following assumption is the minimal one to define a sensible quantum code $\qcode$. This assumption will always hold for our later constructions of matrices $G$.
\begin{assumption}\label{assump:independence}
    We assume that $\cG_0\cap \cG_1 = 0$, and the rows of $G_1$ are linearly independent.
\end{assumption}
Given this assumption, we can analyze the logical structure of our code.
Suppose $G$ has rank $r$. 
It follows from the above assumption that $\dim(\cG_1) = k$ and $\dim(\cG_0) = r-k$.
The number of logical qudits is $n - \dim(\cG^\perp) - \dim(\cG_0) = n - (n-r) - (r-k) = k$. 
For a vector $u\in \FF_q^k$, the logical encoding of the corresponding $k$-qudit computational basis state in the code $\qcode$ is, up to normalization,
\begin{align}~\label{eq:logical_state}
    \ol{\ket{u}} \propto \sum_{g\in \cG_0} \Ket{\sum_{a=1}^k u_ag^a + g}. 
\end{align}
Here, we have made the choice to let the logical $X$ operator for the $a$-th logical qudit be $\overline{X[a]} = X^{\beta g^a}$ for $a \in [k]$.
We can also find $\{ h^b\in \cG_0^\perp\setminus\cG^\perp \}_{b \in [k]}$ such that $g^a\cdot h^b = \one_{a=b}$, in which case
the $Z^\gamma$ logical operator of the $a$-th logical qubit would be $Z^{\gamma h^a}$.
The distance of our quantum code is $d(\qcode)$, which from Corollary~\ref{cor:CSS_distance} we know is at least $\min\left\{d(\cG_0^\perp), d(\cG)\right\}$. 
Let $T$ denote the first $k$ indices of $\tilde G$, which is punctured from $G$.
Since $\cG_0 = \shrT{C}$, we have from Lemma~\ref{lem:short_punc} that $\cG_0^\perp = \puncT{(C^\perp)}$. 
Moreover, $\cG = \puncT{C}$.
Therefore, we have
\begin{align}\label{eq:dist_qcode}
    d(\qcode) \ge \min\left\{d(\mathcal{G}_0^\perp),d(\mathcal{G})\right\}=\min\left\{d(\puncT{(C^\perp)}), d(\puncT{C})\right\} \ge \min\left\{d(C^\perp), d(C)\right\}-k. 
\end{align}
To get a non-trivial bound on the distance of $\mathcal{Q}$, it will therefore be sensible to not only choose $k$ with $k < K$, but also a $k$ with $k < d(C^\perp)$ and $k < d(C)$.

\subsection{Classical Algebraic Codes}

We first review standard definitions and constructions of Reed-Solomon codes.     
For a finite field $\FF_q$, let $\alpha_1, \cdots, \alpha_{\bigN}$ be a set of $\bigN$ distinct points in $\FF_q$, denoted $\bal = (\al_1, \cdots, \al_{\bigN})$.
For an integer $0< m < \bigN$, the Reed-Solomon code $\RS_{\bigN,m}(\bal)$ is defined as follows.

\begin{definition}[Reed-Solomon Codes]
    \begin{align}
        \RS_{\bigN,m}(\bal) = \{(f(\alpha_1), \cdots, f(\al_{\bigN})): f\in \FF_q[X]_{<m}\}.
    \end{align}
\end{definition}
Here we use $\FF_q[X]_{<m}$ to denote all polynomials in $\FF_q[X]$ of degree less than $m$. 
Note that for $m_1 < m_2< \bigN$, we have $\RS_{\bigN,m_1}(\bal) \subset \RS_{\bigN,m_2}(\bal)$.

Most commonly (and most simply) in the literature, the evaluation points of a Reed-Solomon code are taken to be the whole field. 
In this case, Reed-Solomon codes have very elegant properties; for example, their dual is another Reed-Solomon code: $\RS_{\bigN,m}(\bal)^\perp = \RS_{\bigN,\bigN-m}(\bal)$. 
This property essentially boils down to the fact that for any polynomial $f$ of degree at most $q-2$, $\sum_{\alpha \in \FF_q} f(\alpha) = 0$. For our purposes, we must also consider the case where the evaluation points are a strict subset of the field. In this case, we have the following similarly elegant property.

\begin{fact} \label{fact:GRS-coeff} 
    For any $\alpha_1, \dots, \alpha_{\bigN} \in \FF_q$, there are some $\nu_1, \dots, \nu_{\bigN} \in \FF_q \backslash \{ 0 \}$ such that for any polynomial $f$ of degree $< \bigN-1$, it holds that
    \begin{align}
        \sum_{i \in [\bigN]} \nu_i f(\alpha_i) = 0.
    \end{align}
\end{fact}

Fix the evaluation points $\bal = (\alpha_1, \dots, \alpha_{\bigN})$ and the corresponding weights $\mathbf{\nu} = (\nu_1, \dots, \nu_{\bigN})$. The dual of a Reed-Solomon code on evaluation points $\bal$ can be described by a \emph{generalised Reed-Solomon code}.

\begin{theorem}\label{thm:GRS_dual}
    Duals of Reed-Solomon codes are generalised Reed-Solomon codes. Specificially,
    \begin{align}
        \RS_{\bigN,m}(\bal)^\perp = \GRS_{\bigN,\bigN-m}(\bal, \mathbf{\nu}),
    \end{align}
    where 
    \begin{align}
        \GRS_{\bigN, \bigN-m}(\bal, \mathbf{\nu}) := \{ \nu_1 f(\al_1), \dots, \nu_{\bigN} f(\al_{\bigN}) : f \in \FF_q[X]_{<m} \}. 
    \end{align}
\end{theorem}
Note that in the case that $\nu = \mathbf{1}^{\bigN}$, the all-1's vector, $\GRS_{\bigN, m}(\bal, \nu) = \RS_{\bigN, m}(\bal)$. 
Lastly, we state the rate and distance of (generalised) Reed-Solomon codes.
\begin{theorem} \label{thm:GRS-rate-distance}
    The generalised Reed-Solomon code $\GRS_{\bigN, m}(\mathbf{\alpha})$ has dimension $m$ and distance $\bigN-m+1$. 
\end{theorem}
For a proof of these theorems and facts, we refer readers to lectures notes of Hall~\cite{hall2003notes}. Aside from (generalised) Reed-Solomon codes, instantiations of our addressable orthogonality framework in Section~\ref{sec:instantiations_orthogonality} will require classical algebraic geometry (AG) codes. For the relevant preliminaries on AG codes, we refer the reader to Section 2.3 of~\cite{wills2024constant}.

\section{Transversal, Addressable CCZ Gates over Qudits}\label{sec:intrablock_GRS}

In this section, we will build quantum codes over qudits that support addressable CCZ gates. Our qudits will be over the field $\bbF_q$, where $q = 2^t$ is a power of $2$. At a high level, they will be built from punctured Reed Solomon codes over $\bbF_q$. Later, in Section~\ref{sec:concatenation}, we will show how to convert these qudit codes with addressable $\CCZ$ gates into qubit codes with addressable $\CCZ$ gates.

Let $\KK \subseteq \FF_q$ be a subfield\footnote{Strictly speaking, one only needs the shifts $\Delta_{A,B}$ (defined below to be the difference between $\beta_A$ and $\beta_B$) to preserve $\alpha_1, \dots, \alpha_n$. For instance, one could instead take $\KK$ to be a linear subspace of $\mathbb{F}_q$.} of order $n$, and let $\zeta \in \FF_q \backslash \KK$.\footnote{One can think about $\zeta$ as a field element connecting the addresses of the physical qudits to those of the logical qudits.} Let $\alpha_1, \dots, \alpha_n$ be the elements of $\KK$, and let $\beta_1, \dots, \beta_k \in \FF_q$ be distinct elements of $\bbF_q$ such that $\beta_i - \zeta \in \KK$. We will also choose $m < \frac{n}{3} + 1$ and $k \in [m]$. 

We begin with a Reed-Solomon code of dimension $m$ over the $\bigN = k + n$ evaluation points $\mathbf{\beta} \cup \mathbf{\alpha}$, where $\mathbf{\beta} = \{ \beta_1, \dots, \beta_k \}$ and
$\mathbf{\alpha} = \{ \alpha_1, \dots, \alpha_n \}$. By Fact~\ref{fact:GRS-coeff}, 
for each $A \in [k]$, by considering the evaluation points $\{ \beta_A \} \cup \{ \alpha_1, \dots, \alpha_n \}$, there are coefficients $\Gamma^A_1, \dots, \Gamma^A_n$ such that 
\begin{align}
    f(\beta_A) + \sum_{i \in [n]} \Gamma^A_i f(\alpha_i) = 0 \label{eqn:A-addressable}
\end{align}
for all polynomials $f$ of degree $< n$, where the coefficient of $f(\beta_A)$ has been normalised to 1.

Let us choose a generating matrix of the Reed-Solomon code $C \coloneq \RS_{k+n, m}(\mathbf{\beta} \cup \mathbf{\alpha})$ over the field $\mathbb{F}_q$. By performing row operations, we can set the generating matrix to be of the form 
\begin{align}\label{eq:partially_systematic_gen_mat}
    \tilde{G} = \begin{pmatrix}
        I_k & G_1 \\
        0 & G_0,
    \end{pmatrix}
\end{align}
Here, the first $k$ columns of $\tilde{G}$ correspond to the evaluation points $\beta_1, \dots, \beta_k$, and the last $n$ columns of $\tilde{G}$ correspond to the evaluation points $\alpha_1, \dots, \alpha_n$. 
This is possible because Reed-Solomon codes are maximum distance separable.\footnote{Indeed, one can consider any $m \times (k+n)$ generator matrix $\tilde{G}$ for $C = \RS_{k+n,m}(\beta\cup\alpha)$. Because the distance of $C$ is then $k+n-m+1$, any set of at most $m$ columns of $\tilde{G}$ is linearly independent, meaning thay any $k$ columns is certainly linearly indepedent. This allows one to bring $\tilde{G}$ into the given form using row operations.} As usual, the rows of $\tilde{G}$ may be denoted by $\tilde{g}^a$ for $a \in [m]$. Each row of $\tilde{G}$ corresponds to a polynomial over $\mathbb{F}_q$ of degree less than $m$, and it is convenient to denote by $\tilde{g}^a$ both the row of $\tilde{G}$ and this polynomial.

Let us consider the quantum code $\qcode$ built as in Section~\ref{sec:qcode_from_puncture}. 
First, we check that assumption~\ref{assump:independence} is satisfied. 
Writing the usual $G = \begin{pmatrix}
    G_1\\ G_0
\end{pmatrix}$, we denote the rows of $G$ by $g^a, a\in [m]$. 
Suppose there exist coefficients $c_a\in \mathbb{F}_q$ such that 
$\sum_{a\in [m]}c_ag^a = 0.$
From equation~\eqref{eqn:A-addressable}, we must have that for all $A\in [k]$,
\begin{align}
    c_A = \left\langle \sum_{a\in [m]}c_ag^a, \Gamma^A \right\rangle = 0,
\end{align}
where $\Gamma^A = (\Gamma^A_1, \ldots, \Gamma^A_n)$. Therefore, any linear dependence can only exist among rows of $G_0$, which proves Assumption~\ref{assump:independence}.
Now recall that $\qcode$ encodes $k$ logical qudits and every $u\in \FF_q^k$ is encoded as
\begin{align}
    \ol{\ket{u}} \propto \sum_{g\in \cL_X} \Ket{\sum_{a=1}^k u_ag^a + g}. 
\end{align}
The distance of $\qcode$ is, by equation~\eqref{eq:dist_qcode}, Theorem~\ref{thm:GRS_dual} and~\ref{thm:GRS-rate-distance},
\begin{align}
    d(\qcode)
    \ge \min(d(C), d(C^\perp)) - k 
    &= \min(n+k-m+1, m+1) - k
    = \min(n-m+1, m-k+1).
\end{align}
This construction gives our first main result.
\begin{theorem}\label{thm:RS_addressable_CCZ}
    Let $n$ be a power of two, and let $q = n^2$.
    For any $m < \frac{n}{3} + 1$ and $k \in [m]$,
    there exists a qudit CSS code $\qcode$ with parameters
    $[[n, k, d\ge m-k+1]]_q$ such that for any three logical indices $A, B, C\in [k]$, for any $\gamma \in \mathbb{F}_q$,
    \begin{enumerate}[itemsep = 0pt]
        \item On a single code block of $\qcode$, $\ol{\CCZ^\gamma[A, B, C]}$ may be implemented by a depth-4 circuit of physical $\CCZ$ gates;
        \item On three distinct code blocks of $\qcode$, $\ol{\CCZ_{123}^\gamma[A, B, C]}$ may be implemented by a depth-1 circuit of physical $\CCZ$ gates.
    \end{enumerate}
    By choosing $m = \lceil\frac{n}{3}\rceil$ and $k = \frac{m}{2}$, we see that $\qcode$ is asymptotically good with rate $1/6$ and relative distance at least $1/6$.
\end{theorem}
\begin{remark} 
    For completeness, we remark that on two distinct code blocks of $\qcode$, $\ol{\CCZ_{112}^\gamma[A,B,C]}$ may be implemented by a depth-2 circuit of physical $\CCZ$ gates. All three cases are similar enough that we only prove the two stated in  Theorem \ref{thm:RS_addressable_CCZ}.
\end{remark}
\begin{remark}\label{rmk:flatten-depth}
    The depth-4 implementation of the intra-block gate may be reduced trivially to a strongly transversal implementation, i.e., a depth-1 implementation, by duplicating every physical qudit 4 times. However, in our case, it will turn out that each qudit is only used three times in such an implementation, so we can achieve strong transversality by duplicating every physical qudit 3 times.\footnote{By ``duplicate every physical qudit 3 times'' we mean ``encode every physical qudit into a length-3 repetition code in the $Z$-basis'', i.e., $\ket{\eta} \mapsto \ket{\eta}^{\otimes 3}$.}
\end{remark}
\begin{remark}\label{rmk:high_degree_poly_gates}
    We only state and prove these results for the $\CCZ$ gate for brevity. However, these constructions may be easily generalised to the $\CLZ_q^\gamma$ gate, and even further to the gates $U_{q,P}^\gamma$ (by decomposing the polynomial into monomials).
\end{remark}
\begin{remark}
    Later on in Section~\ref{sec:concatenation}, we will show how to convert these qudit codes into qubit codes which also have addressable $\CCZ$ gates, up to a small loss of parameters, via the concatenation scheme of~\cite{nguyen2024goodbinaryquantumcodes}. The resulting qubit code will have both intra-block and inter-block $\CCZ$ gates implementable by a depth one physical circuit of $\CCZ$ gates. Now, because the qudit codes in this case have parameters $[[n,\Theta(n),\Theta(n)]]_q$, where $q = n^2$ (and so the size of the field grows with the code length), the resulting qubit code will have parameters $\left[\left[n,\Omega\left(\frac{n}{\text{polylog}(n)}\right),\Omega\left(\frac{n}{\text{polylog}(n)}\right)\right]\right]_2$.
\end{remark}
\begin{remark}\label{rmk:sets_triples}
    The given code allows us to transversally address $\CCZ$ on any triple of logical qudits in one, two, or three different code blocks. One may also wonder about the possibility of executing the $\CCZ$ gate on a \textit{set} of triples of logical qudits with one transversal operation. This is possible for some such sets, as will be discussed in Section \ref{sec:sets_triples_CCZ}.
\end{remark}
We will prove the two claims in Theorem \ref{thm:RS_addressable_CCZ} in the subsequent sections.

\subsection{Intra-Block CCZ Gates}

We will now define a collection of $\CCZ$ gates that, when applied to a physical code state, will apply a logical $\CCZ^\gamma$ on logical qudits $A$, $B$, and $C$ within the same code block, i.e., the operation $\ol{\CCZ^\gamma[A,B,C]}$. Note that, since the physical    qudits of the above code are in one-to-one correspondence with the set $\mathbf{\alpha}$, we may identify the set of physical qudits with this set. For example, a $\CCZ^\gamma$ gate on three physical qudits may be unambiguously written as $\CCZ^\gamma[\alpha_i, \alpha_j, \alpha_k]$.

Recall that we chose the evaluation points $\beta_i$ to lie in $\zeta + \KK$, where $\zeta \in \bbF_q \backslash \KK$. Thus, for every $A,B \in [k]$, there exists a field element $\Delta_{A,B} \in \KK$ such that 
\begin{align}
    \beta_A + \Delta_{A,B} &= \beta_B.
\end{align}
With this, we may define the logical $\CCZ$ operation in terms of physical $\CCZ$'s as follows.

\begin{theorem} \label{thm:intra-block}
    For any $\ket{\psi} \in \qcode$, it holds that
    \begin{align}
        \ol{\CCZ^\gamma[A,B,C]} \ket{\psi} 
        = \prod_{i \in [n]} \CCZ^{\gamma\;\Gamma^A_i}[\alpha_i, \alpha_i + \Delta_{A,B}, \alpha_i + \Delta_{A,C}] \ket{\psi}.
    \end{align}
    Hence we can define
    \begin{align}
        \ol{\CCZ^\gamma[A,B,C]} 
        = \prod_{i \in [n]} \CCZ^{\gamma\;\Gamma^A_i}[\alpha_i, \alpha_i + \Delta_{A,B}, \alpha_i + \Delta_{A,C}]. \label{eqn:intra-CCZ-def}
    \end{align}
\end{theorem}
Note that because $\{ \alpha_1, \dots, \alpha_n \} = \KK$ is closed under additions by elements of $\KK$, all three coordinates are well defined for each $\alpha_i \in \KK$. 

\begin{proof}
    Let us check the effect of applying $\prod_{i \in [n]} \CCZ^{\gamma\;\Gamma^A_i}[\alpha_i, \alpha_i + \Delta_{A,B}, \alpha_i + \Delta_{A,C}]$ on a code state. By linearity, it is sufficient to check the effect of applying this gate to a logical computational basis state $\ol{\ket{u}}$ as defined in~\eqref{eq:logical_state}. In turn, we check the effect of this gate on $\ket{\sum_{a \in [m]} u_a g^a}$, and again linearity will ensure the correct action on $\overline{\ket{u}}$ also.

    Let $\tilde{g}^{(u)} := \sum_{a \in [m]} u_a \tilde{g}^a$, so that the state $\ket{\sum_{a \in [m]} u_a g^a}$ simply consists of the evaluations of $\tilde{g}^{(u)}$ at $\alpha_1, \dots, \alpha_n$. For short, we also write $g^{(u)} = \sum_{a \in [m]} u_a g^a$.
    
    We can write:
    \begin{align}
        &\prod_{i \in [n]} \CCZ^{\gamma\;\Gamma^A_i}[\alpha_i, \alpha_i + \Delta_{A,B}, \alpha_i + \Delta_{A,C}] \ket{g^{(u)}} \\ 
        = &\prod_{i \in [n]} \exp(i\pi \tr(\gamma\;\Gamma^A_i \cdot \tilde{g}^{(u)}(\alpha_i) \cdot \tilde{g}^{(u)}(\alpha_i + \Delta_{A,B}) \cdot \tilde{g}^{(u)}(\alpha_i + \Delta_{A,C}) )) \ket{g^{(u)}} \\
        = &\exp(i\pi \tr(\gamma\sum_{i \in [n]} \Gamma^A_i \cdot \tilde{g}^{(u)}(\alpha_i) \cdot \tilde{g}^{(u)}(\alpha_i + \Delta_{A,B}) \cdot \tilde{g}^{(u)}(\alpha_i + \Delta_{A,C}) )) \ket{g^{(u)}}.
    \end{align}
    By~\eqref{eqn:A-addressable} and because $\tilde{g}^{(u)}(x) \cdot \tilde{g}^{(u)}(x + \Delta_{A,B}) \cdot \tilde{g}^{(u)}(x + \Delta_{A,C})$ is a degree $\le 3m-3 < n$ polynomial in $x$, the summation in the trace in the exponent is equal to
    \begin{align}
        \sum_{i \in [n]}\Gamma^A_i \cdot \tilde{g}^{(u)}(\alpha_i) \cdot \tilde{g}^{(u)}(\alpha_i + \Delta_{A,B}) \cdot \tilde{g}^{(u)}(\alpha_i + \Delta_{A,C})
        &= \tilde{g}^{(u)}(\beta_A) \cdot \tilde{g}^{(u)}(\beta_A + \Delta_{A,B}) \cdot \tilde{g}^{(u)}(\beta_A + \Delta_{A,C}) \\
        &= \tilde{g}^{(u)}(\beta_A) \cdot \tilde{g}^{(u)}(\beta_B) \cdot \tilde{g}^{(u)}(\beta_C) \\
        &= u_A u_B u_C,
    \end{align}
    where the last line follows from the structure of the chosen generator matrix in Equation~\eqref{eq:partially_systematic_gen_mat}.
    Thus, we get the effect of applying $\prod_{i \in [n]} \CCZ^{\Gamma^A_i}[\alpha_i, \alpha_i + \Delta_{A,B}, \alpha_i + \Delta_{A,C}]$ on our state $\ket{g^{(u)}}$ is 
    \begin{align}
        \prod_{i \in [n]} \CCZ^{\gamma\;\Gamma^A_i}[\alpha_i, \alpha_i + \Delta_{A,B}, \alpha_i + \Delta_{A,C}] \ket{g^{(u)}} 
        &= \exp(i\pi \tr(\gamma u_A u_B u_C)) \ket{g^{(u)}}.
    \end{align}
    Referring to the form of $\overline{\ket{u}}$ in~\eqref{eq:logical_state}, we see that this implies
    implying that 
    \begin{align}
        \prod_{i \in [n]} \CCZ^{\gamma\;\Gamma^A_i}[\alpha_i, \alpha_i + \Delta_{A,B}, \alpha_i + \Delta_{A,C}] \ol{\ket{u}} 
        &= \ol{\CCZ^\gamma[A,B,C]} \ol{\ket{u}},
    \end{align}
    as claimed.
\end{proof}

\paragraph{Circuit Depth.}
We've shown that we can apply a logical $\ol{\CCZ^\gamma[A,B,C]}$ gate by applying $n$ different $\CCZ$ gates on the physical qudits. We claim that one can arrange these $n$ different $\CCZ$ gates into a depth $4$ circuit.

For any $\alpha \in \KK$, consider the four $\CCZ$ gates that touch the following triples of qudits:
\begin{center}
\begin{tabular}{crrrc}
    $[$ & $\alpha,$ & $\alpha + \Delta_{A,B},$ & $\alpha + \Delta_{A,C}$ & $]$ \\
    $[$ & $\alpha + \Delta_{A,B},$ & $\alpha,$ & $\alpha + \Delta_{A,B} + \Delta_{A,C}$ & $]$ \\
    $[$ & $\alpha + \Delta_{A,C},$ & $\alpha + \Delta_{A,B} + \Delta_{A,C},$ & $\alpha$ & $]$ \\
    $[$ & $\alpha + \Delta_{A,B} + \Delta_{A,C},$ & $\alpha + \Delta_{A,C},$ & $\alpha + \Delta_{A,B}$ & $]$.
\end{tabular}
\end{center}
Note that each of the qudits $\alpha, \alpha + \Delta_{A,B}, \alpha + \Delta_{A,C}, \alpha + \Delta_{A,B} + \Delta_{A,C}$ occur exactly once in each of the three positions. Since there is exactly one physical $\CCZ$ gate with a given qudit in a given position, this means that there are no other $\CCZ$ gates that share a qudit with any of these four $\CCZ$ gates. Here, we've used that our field has characteristic $2$, but a field of a higher characteristic would simply lead to a larger constant circuit depth.

This means that we can arrange the $n$ physical $\CCZ$ gates into a circuit of depth-$4$ by placing one triple from each set of four gates into each layer. As discussed in Remark~\ref{rmk:flatten-depth}, one could instead duplicate each qudit $3$ times to obtain a code with a fully transversal implementation of $\ol{\CCZ^\gamma[A,B,C]}$, i.e., a circuit of depth one, because each qudit is only used three times in the depth-4 implementation.

\subsection{Inter-Block CCZ Gates}

In the last section, we showed how to implement a logical $\CCZ$ gate on three qudits \emph{within} the same code state, meaning an intra-block $\CCZ$ gate. In this section, we will show how to implement a logical $\CCZ$ gate on any three logical qudits in three \emph{different} code blocks. Here, we will be able to get full transversality in the strongest sense, meaning that the circuit will be of depth one.

We will construct the logical operator that applies a $\CCZ$ gate between qudits $A$ in the first code state, $B$ in the second code state, and $C$ in the third code state. We will denote this logical operator by $\ol{\CCZ_{123}^\gamma[A, B, C]}$. Likewise, a physical $\CCZ$ gate on qudits $\alpha_i$ in code state $1$, $\alpha_j$ in code state $2$, and $\alpha_k$ in code state $3$ will be denoted by $\CCZ_{123}[\alpha_i, \alpha_j, \alpha_k]$.

Again, we will take field elements $\Delta_{A,B} \in \KK$ such that
\begin{align}
    \beta_A + \Delta_{A,B} &= \beta_B.
\end{align}

\begin{theorem} \label{thm:inter-block}
    For any $\ket{\psi_1}, \ket{\psi_2}, \ket{\psi_3} \in \qcode$, it holds that
    \begin{align}
        \ol{\CCZ_{123}^\gamma[A,B,C]} \ket{\psi} 
        = \prod_{i \in [n]} \CCZ_{123}^{\gamma\;\Gamma^A_i}[\alpha_i, \alpha_i + \Delta_{A,B}, \alpha_i + \Delta_{A,C}] \ket{\psi_1}, \ket{\psi_2}, \ket{\psi_3}.
    \end{align}
    Hence we can define
    \begin{align}
        \ol{\CCZ^\gamma_{123}[A,B,C]} 
        = \prod_{i \in [n]} \CCZ_{123}^{\gamma\;\Gamma^A_i}[\alpha_i, \alpha_i + \Delta_{A,B}, \alpha_i + \Delta_{A,C}]. \label{eqn:inter-CCZ-def}
    \end{align}
\end{theorem}

The proof of Theorem~\ref{thm:inter-block} will follow via essentially the same calculation as in the proof of Theorem~\ref{thm:intra-block}. 

\begin{proof}
    We will check the effect of applying $\prod_{i \in [n]}\CCZ_{123}^{\gamma\;\Gamma^A_i}[\alpha_i, \alpha_i + \Delta_{A,B}, \alpha_i + \Delta_{A,C}]$ on $\ket{g^{(u)}}\ket{g^{(v)}}\ket{g^{(w)}}$ where $g^{(u)} := \sum_{a \in [m]} u_a g^a$, $g^{(v)} := \sum_{a \in [m]} v_a g^a$, and $g^{(w)} := \sum_{a \in [m]} w_a g^a$, which will be enough to check its behaviour on $\ol{\ket{u}}\ol{\ket{v}}\ol{\ket{w}}$, again by linearity. Let $\tilde{g}^{(u)} := \sum_{a \in [m]} u_a \tilde{g}^a$, $\tilde{g}^{(v)} := \sum_{a \in [m]} v_a \tilde{g}^a$, and $\tilde{g}^{(w)} := \sum_{a \in [m]} w_a \tilde{g}^a$ be the corresponding polynomials of degree $< m$.

    We have that
    \begin{align}
        &\prod_{i \in [n]}\CCZ_{123}^{\gamma\;\Gamma^A_i}[\alpha_i, \alpha_i + \Delta_{A,B}, \alpha_i + \Delta_{A,C}] \ket{g^{(u)}}\ket{g^{(v)}}\ket{g^{(w)}} \\
        = &\prod_{i \in [n]} \exp({i\pi\tr(\gamma\;\Gamma^A_i \cdot \tilde{g}^{(u)}(\alpha_i) \cdot \tilde{g}^{(v)}(\alpha_i + \Delta_{A,B}) \cdot \tilde{g}^{(w)}(\alpha_i + \Delta_{A,C}) )}) \ket{g^{(u)}}\ket{g^{(v)}}\ket{g^{(w)}} \\
        = &\exp({i\pi\tr(\gamma\sum_{i \in [n]} \Gamma^A_i \cdot \tilde{g}^{(u)}(\alpha_i) \cdot \tilde{g}^{(v)}(\alpha_i + \Delta_{A,B}) \cdot \tilde{g}^{(w)}(\alpha_i + \Delta_{A,C}) )}) \ket{g^{(u)}}\ket{g^{(v)}}\ket{g^{(w)}}.
    \end{align}
Because $\sum_{i \in [n]} \Gamma^A_i \cdot \tilde{g}^{(u)}(x) \cdot \tilde{g}^{(v)}(x + \Delta_{A,B}) \cdot \tilde{g}^{(w)}(x + \Delta_{A,C})$ is a degree $\le 3m-3 < n$ polynomial in $x$, by Equation~\eqref{eqn:A-addressable} we have that 
\begin{align}
    \sum_{i \in [n]} \Gamma^A_i \cdot \tilde{g}^{(u)}(\alpha_i) \cdot \tilde{g}^{(v)}(\alpha_i + \Delta_{A,B}) \cdot \tilde{g}^{(w)}(\alpha_i + \Delta_{A,C})
    &= \tilde{g}^{(u)}(\beta_A) \cdot \tilde{g}^{(v)}(\beta_A + \Delta_{A,B}) \cdot \tilde{g}^{(w)}(\beta_A + \Delta_{A,C}) \\
    &= \tilde{g}^{(u)}(\beta_A) \cdot \tilde{g}^{(v)}(\beta_B) \cdot \tilde{g}^{(w)}(\beta_C) \\
    &= u_A v_B w_C,
\end{align}
so 
\begin{align}
    \CCZ_{123}^{\gamma\;\Gamma^A_i}[\alpha_i, \alpha_i + \Delta_{A,B}, \alpha_i + \Delta_{A,C}] \ket{g^{(u)}}\ket{g^{(v)}}\ket{g^{(w)}}
    &= \exp({i\pi\tr(\gamma u_A v_B w_C)}) \ket{g^{(u)}}\ket{g^{(v)}}\ket{g^{(w)}},
\end{align}
as desired.
\end{proof}

\paragraph{Circuit Depth.} 
Note that each physical qudit appears in exactly one of the $n$ physical $\CCZ$ gates, so our definition of $\ol{\CCZ_{123}^\gamma[A,B,C]}$ in~\eqref{eqn:inter-CCZ-def} is a strongly transversal implementation, i.e., implementable via a depth-1 circuit.

\subsection{Addressing Sets of Triples of Logical Qudits}\label{sec:sets_triples_CCZ}

As was mentioned in Remark~\ref{rmk:sets_triples}, while this code can address any triple of logical qudits with a particular transversal operation, it is also interesting to ask to what extent it can address arbitrary \textit{sets} of triples of logical qudits with the $\CCZ$ gate via a transversal action. Of course, by combining the operations above a constant number of times, one may ensure a transversal action, and act on certain sets of triples of logical qudits with the $\CCZ$ gate. The interesting question, however, is to ask which sets may be acted on without increasing the depth at all (from 4 in the intra-block case, and 1 in the inter-block case).

In either the intra-block, or inter-block case, a triple of logical qudits $[A,B,C]$ may be specified by a label $A \in [k]$ as well as field elements $\Delta_{A,B}, \Delta_{A,C} \in \mathbb{K}$ in the obvious way. It is possible to address a set of triples of logical qudits without increasing the depth if these qudits share the same field elements $\Delta_{A,B}, \Delta_{A,C}$, but the label $A \in [k]$ is allowed to vary. For example, in the intra-block case, if two triples $[A_1, B_1, C_1]$ and $[A_2, B_2, C_2]$ have $\Delta_{A_1, B_1} = \Delta_{A_2, B_2}$ and $\Delta_{A_1, C_1} = \Delta_{A_2, C_2}$, then we have
\begin{equation}
    \ol{\CCZ^{\gamma_1}[A_1,B_1,C_1]\CCZ^{\gamma_2}[A_2,B_2,C_2]} = \prod_{i \in [n]}\CCZ^{\gamma_1\,\Gamma_i^{A_1}+\gamma_2\,\Gamma_i^{A_2}}[\alpha_i, \alpha_i+\Delta_{A,B}, \alpha_i+\Delta_{A,C}],
\end{equation}
which is again implementable via a depth 4 circuit. As a further example, the code can address arbitrary sets of interblock logical qudit triples of the form $[A,A,A]$ (since these all have $\Delta_{A,B} = \Delta_{A,C} = 0$) with a depth one circuit as in the case of single-index addressability derived from a transversal $\CCCZ$ gate --- see Section \ref{sec:prelim_non_clifford_addressability} --- but we emphasise that this code's capabilities to address (sets of) logical qudits are much stronger.

\section{The Addressable Orthogonality Framework and Phase Gates over Qudits}\label{sec:addressable_orthogonality_big_section}
Section~\ref{sec:qcode_from_puncture} described how the generator matrix $G$ for a punctured classical code could be built into a quantum code $\qcode$, assuming a certain basic assumption (\ref{assump:independence}) was satisfied. The intention behind this section is to show that by assuming further conditions on the matrix, encapsulated in our \textit{addressable orthogonality} conditions described in Section~\ref{sec:addressable_orthogonality}, the resulting quantum code can be shown to support desirable (addressable) transversal gates. This framework of addressable orthogonality naturally encapsulates all the previous definitions made in this area for building codes with desirable transversal gates~\cite{bravyi2012magic,krishna2019towards,wills2024constant}, including the original notion of triorthogonality introduced by Bravyi and Haah, and also extends these frameworks very naturally to build codes supporting addressable transversal gates.

In Section~\ref{sec:addressable_orthogonality}, we will begin by presenting the definitions relevant to this framework, before proving some necessary, basic results on these definitions. Then, we will go on to show in Sections~\ref{sec:address_ortho_univariate_gates} to~\ref{sec:address_ortho_poly_gates} that assuming various addressable orthogonality properties on the matrix $G$ guarantees various transversal, addrssable gates on the quantum code $\qcode$. In particular, recalling the definitions made in Section~\ref{sec:prelim_non_clifford_addressability}, in Section~\ref{sec:address_ortho_univariate_gates}, we will discuss the addressability of the univariate gates $U_\ell^\beta$, and in Section~\ref{sec:address_ortho_single_index_CCZ}, we will discuss the single-index addressability of $\CCZ$ gates. In Section~\ref{sec:designed_intrablock} we will show, given a matrix with a certain addressable orthogonality property, how to construct a quantum code such that fixed, pairwise disjoint triples of intra-block logical qudits may be addressed with a non-Clifford three qudit gate via a transversal physical circuit. While this form of adressability is much weaker than the addressability in Section~\ref{sec:intrablock_GRS}, we will be able to instantiate this construction with algebraic geometry codes, and so the corresponding qubit codes constructed in Section~\ref{sec:concatenation} will be asymptotically good, rather than near-asymptotically good. Finally, in Section~\ref{sec:address_ortho_poly_gates} we will discuss the single-index addressability of general polynomial gates $U_P^\beta$. We will also show in Appendix~\ref{sec:transversal_T} that one of our addressable orthogonality properties can provide quantum codes for qubits supporting a transversal, addressable $T$ gate, up to some Clifford correction. This directly extends Bravyi and Haah's original notion of triorthogonality, and so we call it ``addressable triorthogonality''. However, we relegate this discussion to the appendices since we do not have an instantiation in mind.
 
\subsection{Addressable Orthogonality Framework}\label{sec:addressable_orthogonality}

Consider an $m\times n$ matrix $G$ over the field $\bbF_q$. 
Fixing $k < m$, we may partition the matrix so that the first $k$ rows form the submatrix $G_1$  and the remaining $m-k$ rows form the submatrix $G_0$,
\begin{equation}
    G = 
    \begin{bmatrix}
        G_1\\
        G_0
    \end{bmatrix}.
    \label{eq:def_G}
\end{equation}
Denoting the rows of $G$ as $g^a, a\in [m]$, we let $\mathcal{G}_0$, $\mathcal{G}_1$ and $\mathcal{G}$ be the row spaces of $G_0, G_1$ and $G$, respectively (over $\mathbb{F}_q$). Assuming Assumption~\ref{assump:independence}, Section~\ref{sec:qcode_from_puncture} describes how the code
\begin{equation}\label{eq:repeated_qcode_defn}
    \mathcal{Q} = \CSS(X,\mathcal{G}_0;Z,\mathcal{G}^\perp)
\end{equation}
forms a quantum code with dimension $k$ and distance at least $\min\left\{d(\mathcal{G}_0^\perp), d(\mathcal{G})\right\}$. To be able to capture, in a unified way, how the code $\qcode$ supports desirable (addressable) transversal gates, we now present the framework of addressable orthogonality.

\begin{enumerate}[(1)]
    \item Given fixed vectors $\Gamma \in \mathbb{F}_q^n$ and $\tau \in (\mathbb{F}_q^*)^k$, $G$ has the \textbf{$\ell$-orthogonality property with respect to} $\boldsymbol{(\Gamma,\tau)}$ if, for all $a_1, \ldots, a_\ell \in [m]$, 
    \begin{equation}~\label{eq:ortho_prop}
        \sum_{i=1}^n \Gamma_i{g}^{a_1}_i {g}^{a_2}_i\cdots  {g}^{a_\ell}_i = 
        \begin{cases}
            \tau_{a_1} &\text{ if $1\le a_1 = a_2 = \cdots = a_\ell \le k$} \\
            0 &\text{ otherwise}
        \end{cases}.
    \end{equation}
    Often, when the identity of the vectors $\Gamma$ and $\tau$ is not important to us, we will abbreviate their usage.

    \item Given fixed vectors $\Gamma \in \mathbb{F}_q^n,$ and $\tau \in (\mathbb{F}_q^*)^k$, $G$ has the \textbf{strong $\ol\ell$-orthogonality property with respect to }$\boldsymbol{(\Gamma,\tau)}$ if, for all $\ell \le \ol\ell$, for all $a_1, \ldots, a_\ell \in [m]$,
    \begin{equation}~\label{eq:strong_ortho_prop}
        \sum_{i=1}^n \Gamma_i {g}^{a_1}_i {g}^{a_2}_i\cdots  {g}^{a_{\ell}}_i = 
        \begin{cases}
            \tau_{a_1} &\text{ if $1\le a_1 = a_2 = \cdots = a_{\ell} \le k$} \\
            0 &\text{ otherwise}
        \end{cases}.
    \end{equation}
    The same abbreviation will often apply to $\Gamma$ and $\tau$.
    
    \item Given a fixed vector $\tau \in (\mathbb{F}_q^*)^k$, $G$ has the \textbf{addressable $\ell$-orthogonality property with respect to }$\boldsymbol{\tau}$ if, for every $A\in [k]$, there exists a vector $\Gamma^A\in \FF_q^n$ such that, for all $a_1, \ldots, a_\ell \in [m]$,
    \begin{equation}~\label{eq:address_ortho_prop}
        \sum_{i=1}^n \Gamma^A_{i}{g}^{a_1}_i {g}^{a_2}_i\cdots  {g}^{a_\ell}_i = 
        \begin{cases}
            \tau_{a_1} &\text{ if $a_1 = a_2 = \cdots = a_\ell = A$} \\
            0 &\text{ otherwise}
        \end{cases}.
    \end{equation}
    A similar abbreviation will apply when the identity of $\tau$ is unimportant.

    \item Given a fixed vector $\tau \in (\mathbb{F}_q^*)^k$, $G$ has the \textbf{addressable strong $\ol\ell$-orthogonality property with respect to }$\boldsymbol{\tau}$ if, for every $A\in [k]$, there exists a vector $\Gamma^A\in \FF_q^n$ such that, for all $\ell \le \ol\ell$, for all $a_1, \ldots, a_\ell \in [m]$,
    \begin{equation}~\label{eq:address_strong_mult_prop}
        \sum_{i=1}^n \Gamma^A_{i}{g}^{a_1}_i {g}^{a_2}_i\cdots  {g}^{a_{\ell}}_i = 
        \begin{cases}
            \tau_{a_1} &\text{ if $a_1 = a_2 = \cdots = a_{\ell} = A$} \\
            0 &\text{ otherwise}
        \end{cases}.
    \end{equation}
\end{enumerate}

These definitions are made to create a general framework for constructing quantum codes supporting (addressable) transversal gates. The first two definitions encapsulate all similar definitions made before on transversal gates, as we now describe.
\begin{enumerate}
    \item In the seminal work of Bravyi and Haah~\cite{bravyi2012magic}, the notion of a \textit{triorthogonal matrix} was defined. In their language, this is a binary $m \times n$ matrix, for which any distinct pair of rows has even overlap, any distinct triple of rows has even overlap, the first $k$ rows all have odd weight, and the latter $m-k$ rows all have even weight. Such a matrix necessarily arises from a punctured generator matrix of a classical code with the 3-multiplication property; see Section~\ref{sec:classical_code_multiplication}. Via the usual method of turning the matrix $G$ into the quantum code $\qcode$, see Section~\ref{sec:qcode_from_puncture}, or simply Equation~\eqref{eq:repeated_qcode_defn}, the resulting quantum code supports a transversal $T$-gate, up to Clifford corrections, and in fact such codes are a complete class of codes with this behaviour~\cite{rengaswamy2020optimality}.
    
    Bravyi and Haah's defintion of a triorthogonal matrix can be checked to be exactly our notion of 3-orthogonality for a matrix over $\mathbb{F}_2$, or indeed our notion of strong 3-orthogonality for a binary matrix.\footnote{One can check that for a binary matrix, each non-strong notion is equivalent to its strong counterpart. This follows from the fact that all elements of $\mathbb{F}_2$ satisfy $x^2 = x$.}

    \item Krishna and Tillich~\cite{krishna2019towards} were the first authors to extend Bravyi and Haah's definition to a definition over a field larger than $\mathbb{F}_2$, doing so for fields of prime order, $\mathbb{F}_p$. Indeed, their definition of a triorthogonal matrix over $\mathbb{F}_p$ is equivalent to the matrix being $2$-orthogonal, and $3$-orthogonal with respect to $(\mathbf{1}^n,\tau)$, for any $\tau$. Via the usual matrix-to-quantum code construction, again see Equation~\eqref{eq:repeated_qcode_defn}, the resulting quantum code over qudits of dimension $p$ supports the transversal gate
    \begin{equation}
        U = \sum_{j \in \mathbb{F}_p}\exp\left(\frac{2\pi i}{p}j^3\right)\ket{j}\bra{j},
    \end{equation}
    which is in exactly the third level of the Clifford hierarchy for primes $p \geq 5$~\cite{cui2017diagonal}.
    \item \cite{wills2024constant} extends Krishna and Tillich's notion to fields of order $q=2^t$. The resulting definition for a matrix over $\mathbb{F}_q$ is exactly our notion of $7$-orthogonality, as well as $2$-orthogonality with respect to any $(\Gamma,\tau)$. That definition was made with the intention to construct a quantum code with a transversal $U_{q,7}$ gate (see Equation~\eqref{eq:Uql_defn}), which is in exactly the third level of the Clifford hierarchy for all $t \geq 3$.
\end{enumerate}
The latter two of the four definitions given, namely the addressability conditions, do not encapsulate earlier definitions in the literature, but rather aim to extend this previous framework for building quantum codes with transversal gates to a framework for building quantum codes with addressable transversal gates too.

To begin developing intuition for these definitions, we establish the following basic facts. First, we show two lemmas proving that addressable orthogonality is a stronger property than orthogonality, but that orthogonality also implies addressable orthogonality  with a smaller value of $\ell$.

\begin{lemma}\label{lem:address_to_ortho}
    If $G$ has addressable (strong) $\ol\ell$-orthogonality with respect to $\tau$, then $G$ also has (strong) $\ol\ell$-orthogonality with respect to $(\Gamma,\tau)$ for some $\Gamma$.
\end{lemma}
\begin{proof}
    In both the strong and non-strong cases, the choice
    \begin{equation}
        \Gamma = \sum_{A \in [k]}\Gamma^A
    \end{equation}
    is sufficient, where $(\Gamma^A)_{A \in [k]}$ are the vectors given by the addressable orthogonality.
\end{proof}

\begin{lemma}\label{lem:ortho_to_address}
    If $G$ has $\ol{\ell}$-orthogonality with respect to some $(\Gamma,\tau)$, then $G$ also has addressable $\ell$-orthogonality, with respect to $\tau$, for all $\ell < \ol{\ell}$.
    Moreover, if $G$ has strong $(\ol\ell+1)$-orthogonality with respect to some $(\Gamma,\tau)$, then $G$ also has addressable strong $\ol\ell$-orthogonality with respect to $\tau$.
\end{lemma}

\begin{proof}
    If $G$ is $\bar{\ell}$-orthogonal with respect to $(\Gamma,\tau)$, then for $A\in [k], \ell < \ol\ell$, let
    \begin{equation}
        \Gamma^A\coloneq \Gamma \star \underbrace{g^a \star \ldots \star g^a}_{\bar{\ell}-\ell \text{ copies }},
    \end{equation}
    where $\star$, as usual, denotes the componentwise product of vectors. This choice establishes the first claim.
    To prove the second claim, we can take $\Gamma^A = \Gamma\star g^A$.
\end{proof}

\begin{remark}
    Note that the result of Lemma~\ref{lem:ortho_to_address} is quite reminiscent of the folklore that a transversal gate in the $\ell$-th level of the Clifford hierarchy implies an addressable gate in the $(\ell-1)$-th level of the Clifford hierarchy which, as discussed, is only completely true for single-qubit gates. One may wonder if Lemma~\ref{lem:ortho_to_address} is secretly equivalent to this folklore. As we will discuss in the introduction to Section~\ref{sec:instantiations_orthogonality}, it is in fact strictly stronger, in that it can lead to better parameters on the quantum code.
\end{remark}

\begin{remark}
    We note that it is possible for a matrix $G$ to have addressable $\ol\ell$-orthogonality for all $\ell\le \ol\ell$, but not to have the addressable strong $\ol\ell$-orthogonality property, since the latter condition requires that the same vector $\Gamma^A$ satisfies equation~\eqref{eq:address_strong_mult_prop} for all $\ell\le \ol\ell$. 
\end{remark}

Next, we establish the minimal property required to make~\eqref{eq:repeated_qcode_defn} a sensible definition of a quantum code, by establishing Assumption~\ref{assump:independence}.
\begin{lemma}\label{lem:ortho_to_assump}
    If $G$ is $\ell$-orthogonal with respect to some $(\Gamma,\tau)$, for any $\ell \ge 2$, then it satisfies Assumption~\ref{assump:independence}.
\end{lemma}
\begin{proof}
By Lemma~\ref{lem:ortho_to_address}, $G$ has addressable $1$-orthogonality.
Suppose there exists coefficients $c_a$ such that 
$\sum_{a\in [m]}c_ag^a = 0.$
From equation~\eqref{eq:address_ortho_prop}, we must have for all $A\in [k]$,
\begin{align}
    c_A\tau_A = \left\langle \sum_{a\in [m]}c_ag^a, \Gamma^A \right\rangle = 0.
\end{align}
Because $\tau_A \neq 0$, $c_A 
= 0$, and so any linear dependence can only exist among rows of $G_0$, proving Assumption~\ref{assump:independence}.
\end{proof}

\subsection{Univariate Gates}\label{sec:address_ortho_univariate_gates}
Let us recall the gates $U_\ell^\beta$ from Section~\ref{sec:prelim_non_clifford_addressability}, which we call the ``univariate gates''. These are a natural class of single-qudit diagonal Clifford hierarchy gates, and in fact can be used to generate all such gates with entries $\pm 1$ on the diagonal. When decomposed into qubit gates as in Section~\ref{sec:qudit_to_qubit_prelim}, the gates $U_\ell^\beta$ become sequences of qubit $\CLZ$ gates, and so understanding their transversality and addressability leads to flexible constructions of qubit codes allowing for these types of sequences to be supported in a transversal and addressable way. Note that, because these are single-qudit gates, there is no distinction between addressability and single-index addressability. We now show that addressable $\ell$-orthogonality leads to the addressability of the $U_\ell^\beta$ gate.

\begin{theorem}\label{thm:degree_ell_gate}
    Suppose $G$ has the addressable $\ell$-orthogonality property with respect to $\tau$. Then for all $A \in [k]$, $\beta \in \mathbb{F}_q$, we have
    \begin{equation}
        \ol{U_\ell^{\beta\tau_A}[A]} = U_\ell^{\beta\Gamma^A},
    \end{equation}
    with $\Gamma^A$ as defined in equation~\eqref{eq:address_ortho_prop}, in the associated quantum code $\qcode$, as defined in Equation~\eqref{eq:repeated_qcode_defn}. In words, the $A$-th qudit may be addressed with the gate $U_\ell^{\beta\tau_A}$ by executing the physical transversal gate $U_\ell^{\beta\Gamma^A}$.
\end{theorem}
\begin{proof}
    We wish to check that the physical gate $U_\ell^{\beta\Gamma^a}$ executes the logical gate $U_\ell^\beta[A]$ on any code state. By linearity, it suffices to check this on any logical computational basis states $\ol{\ket{u}}$ for $u \in \mathbb{F}_q^k$,
    \begin{equation}\label{eq:qcode_comp_basis_logical}
        \ol{\ket{u}} \propto \sum_{g \in \mathcal{G}_0}\Ket{\sum_{a=1}^ku_ag^a+g},
    \end{equation}
    where we do not write normalisation for brevity.
    In turn, we begin by checking the relation on the state $\ket{f}$, where $f = \sum_{a=1}^mu_ag^a$ for $u_a \in \mathbb{F}_q$. Indeed, we have
    \begin{align}
        U_\ell^{\beta\Gamma^A}\ket{f} &= \prod_{i=1}^n\exp\left(i\pi\tr(\beta\Gamma_i^Af_i^\ell)\right)\ket{f}\\
        &=\exp\left(i\pi\tr(\beta\sum_{i=1}^n\Gamma_i^Af_i^\ell)\right)\ket{f},
    \end{align}
    where the sum has been brought inside the trace because it is an $\mathbb{F}_2$-linear map.
    We may write
    \begin{equation}
        f_i^\ell = \sum_{a_1, \ldots, a_\ell=1}^mu_{a_1}\ldots u_{a_\ell}g_i^{a_1}\ldots g_i^{a_\ell}
    \end{equation}
    so that
    \begin{equation}
        \sum_{i=1}^n\Gamma_i^Af_i^\ell = u_A^\ell\tau_A,
    \end{equation}
    using the addressable $\ell$-orthogonality property with respect to $\tau$. We have found that
    \begin{equation}
        U_\ell^{\beta\Gamma^A}\ket{f} = \exp\left(i\pi\tr(\beta u_A^\ell\tau_A)\right)\ket{f}.
    \end{equation}
    This means that
    \begin{align}
        U_\ell^{\beta\Gamma^A}\ol{\ket{u}}&\propto \sum_{g \in \mathcal{G}_0}U_\ell^{\beta\Gamma^A}\Ket{\sum_{a=1}^ku_ag^a+g}\\
        &=\exp\left(i\pi\tr(\beta u_A^\ell\tau_A)\right)\ol{\ket{u}}\\
        &=\ol{U_\ell^{\beta\tau_A}[A]}\ol{\ket{u}}. \qedhere
    \end{align}
\end{proof}

\subsection{Single-Index Inter-Block CCZ Gates}\label{sec:address_ortho_single_index_CCZ}

We now turn to the most important multi-qudit gate in this work, the $\CCZ$ gate. Recall that, given three different code blocks, by $\CCZ^\beta[A]$ we mean the gate $\CCZ^\beta_{123}[A,A,A]$, i.e., the $\CCZ$ gate acting on the $A$'th physical qudit in each code block. We refer to such a gate as a single-index inter-block $\CCZ$ gate. As one would expect, we also denote
\begin{equation}
    \CCZ^\Gamma = \bigotimes_{i=1}^n\CCZ^{\Gamma_i}[i],
\end{equation}
where $\Gamma \in \mathbb{F}_q^n$. In Section~\ref{sec:intrablock_GRS}, we showed that by leveraging the particular structure of Reed-Solomon codes, qudit codes supporting an addressable $\CCZ$ gate could be constructed (meaning that any triple within one, two or three codeblocks can be addressed by a transversal circuit) and these will be converted into qubit codes with the same property in Section~\ref{sec:concatenation}. Here, we show that the weaker property of single-index addressability of $\CCZ$ gates can be obtained within the addressable orthogonality framework, and therefore arises for much more general codes than Reed-Solomon codes.

\begin{theorem}\label{thm:single_index_CCZ}
    Suppose $G$ has the addressable $3$-orthogonality property with respect to $\tau$. Then for all $A\in [k], \beta\in \FF_q$, we have
    \begin{equation}
        \ol{\CCZ^{\beta\tau_A}[A]} = \CCZ^{\beta\Gamma^A},
    \end{equation}
    with $\Gamma^A$ as defined in equation~\eqref{eq:address_ortho_prop}, in the associated quantum code $\qcode$. In words, we can address the $A$-th logical qudits in the three code blocks with a logical $\CCZ$ gate by acting transversally on the physical qudits.
\end{theorem}
\begin{proof}
    Just as in the proof of Theorem~\ref{thm:degree_ell_gate}, it will suffice to demonstrate this on three logical computational basis states, $\ol{\ket{u}\ket{v}\ket{w}}$. In turn, we consider $\ket{f^{(u)}}\ket{f^{(v)}}\ket{f^{(w)}}$, where $f^{(u)} = \sum_{a=1}^mu_ag^a$, $f^{(v)} = \sum_{a=1}^mv_ag^a$, and $f^{(w)} = \sum_{a=1}^mw_ag^a$ for some $u_a, v_a, w_a \in \mathbb{F}_q$. We have
    \begin{align}
        \CCZ^{\beta\Gamma^A}\ket{f^{(u)}}\ket{f^{(v)}}\ket{f^{(w)}} &= \prod_{i=1}^n\exp\left(i\pi\tr(\beta\Gamma_i^Af^{(u)}_if^{(v)}_if^{(w)}_i)\right)\ket{f^{(u)}}\ket{f^{(v)}}\ket{f^{(w)}}\\
        &=\exp\left(i\pi\tr(\beta\sum_{i=1}^n\Gamma_i^Af_i^{(u)}f_i^{(v)}f_i^{(w)})\right)\ket{f^{(u)}}\ket{f^{(v)}}\ket{f^{(w)}}.
    \end{align}
    We may write
    \begin{equation}
        f_i^{(u)}f_i^{(v)}f_i^{(w)} = \sum_{a_1, a_2, a_3=1}^mu_{a_1}v_{a_2}w_{a_3}g_i^{a_1}g_i^{a_2}g_i^{a_3}
    \end{equation}
    so that
    \begin{equation}
        \sum_{i=1}^n\Gamma_i^Af_i^{(u)}f_i^{(v)}f_i^{(w)} = \tau_Au_Av_Aw_A,
    \end{equation}
    where we have employed the addressable 3-orthogonal property here. Therefore,
    \begin{equation}
        \CCZ^{\beta\Gamma^A}\ket{f^{(u)}}\ket{f^{(v)}}\ket{f^{(w)}} = \exp\left(i\pi\tr(\beta\tau_Au_Av_Aw_A)\right)\ket{f^{(u)}}\ket{f^{(v)}}\ket{f^{(w)}}.
    \end{equation}
    Finally, considering the form of the logical computation basis states, as in Equation~\eqref{eq:qcode_comp_basis_logical}, we have, by linearity,
    \begin{align}
        \CCZ^{\beta \Gamma^A}\ol{\ket{u}\ket{v}\ket{w}} &= \exp(i\pi\tr(\beta\tau_Au_Av_Aw_A))\ol{\ket{u}\ket{v}\ket{w}}\\
        &=\ol{\CCZ^{\beta\tau_A}[A]\ket{u}\ket{v}\ket{w}},
    \end{align}
    as required.
\end{proof}
The following may be shown by a straightforward extension of the above proof.
\begin{remark}
    Suppose $G$ has the addressable $\ell$-orthogonality property, with $\Gamma^A$ as defined in equation~\eqref{eq:address_ortho_prop}. Then for all $A\in [k], \beta\in \FF_q$, we have $\ol{\CLZ^{\beta\tau_A}[A]} = \CLZ^{\beta\Gamma^A}$.
\end{remark}

\subsection{Designed Intra-Block 3-Qudit Gates}\label{sec:designed_intrablock}

Given a matrix $G$ with addressable $7$-orthogonality with respect to $\tau = \mathbf{1}^k$, we show how to construct a quantum CSS code from $G$ such that a particular non-Clifford 3-qudit gate 
can be implemented addressably and trasversally on predesigned intra-block triples of logical qudits using physical $U_7$ gates. This is much weaker than the result of Section~\ref{sec:intrablock_GRS}, which can address any triple of logical qudits in one, two, or three code blocks, but there the corresponding qubit code has parameters $\left[\left[n,\Omega\left(\frac{n}{\text{polylog}(n)}\right),\Omega\left(\frac{n}{\text{polylog}(n)}\right)\right]\right]$, whereas the qubit code corresponding to the present construction may be made asymptotically good in Section~\ref{sec:concatenation} via an instantiation with algebraic geometry codes in Section~\ref{sec:AG}.

Writing, as always, the $m \times n$ matrix $G$ as
\begin{equation}
    G = \begin{bmatrix}
        G_1\\G_0,
    \end{bmatrix}
\end{equation}
consider the $k\times n$ matrix $G_1$, and let us write $k = 7k'$.\footnote{Without loss of generality, we may assume that $k$ is divisible by 7, since if not we may add the leftover rows in $G_1$ into $G_0$, which amounts to removing some logical operators and adding some $X$-stabilisers to the code. Such an operation does not decrease the distance, and will decrease the dimension of the code by a negligible amount. Also, $G$ will retain its addressable 7-orthogonal property.} 
It makes sense to denote the rows of $G_1$ as $g^{7B+e}$ for $0\le B\le k'-1, e\in [7]$. 
Define the following $3\times 7$ matrix over $\FF_q$,

\begin{equation}\label{eq:matrix_M}
    M = 
    \begin{bmatrix}
        1 & 1 & 1 & 0 & 1 & 0 & 0 \\
        1 & 1 & 0 & 1 & 0 & 1 & 0 \\
        1 & 0 & 1 & 1 & 0 & 0 & 1
    \end{bmatrix}.
\end{equation}
Let $\cM$ be the $3k'\times k$ block-diagonal matrix which has $k'$ blocks of $M$ along its diagonal.
Then, consider the matrix
\begin{equation}
    H_0 = \cM G_1,
\end{equation}
meaning the multiplication of the matrix $\cM$ and $G_1$ over $\mathbb{F}_2$, and note that $H_0$ has dimensions $3k' \times n$. 
Denote the resulting matrix $H_1$, and define the matrix
\begin{equation}
    H = 
    \begin{bmatrix}
        H_1\\
        G_0
    \end{bmatrix}.
    \label{eq:def_H}
\end{equation}
In words, what has happened is that we have taken the matrix $G = \begin{bmatrix}
    G_1\\G_0
\end{bmatrix}$ and we have divided the rows of $G_1$ into groups of seven. Then, we have taken linear combinations of each group of seven rows, in accordance with the rows of $M$, to produce groups of three rows which form $H_1$. 

Now let $m' = m - 4k'$ be the number of rows of $H$. 
We denote the rows of $H$ as $h^{a}$ for $a\in [m']$. 
Specifically, the first $3k'$ rows are denoted as $h^{3B+e}$ for $0\le B\le k'-1, e\in [3]$, whereas the latter rows are simply denoted $h^a$ for $3k' < a \leq m'$. Note that these are simply rows of $G_0$, so that for these rows we in fact have $h^a = g^{a+4k'}$.

Since $G$ is addressable 7-orthogonal, Equation~\eqref{eq:address_ortho_prop} gives us certain vectors $(\Gamma^A)_{A \in [k]}$. With these, we have the following lemma.
\begin{lemma}\label{lem:extended_orthogonality}
    For $0\le A\le k'-1$, let $\Lambda^A = \sum_{j=1}^7 \Gamma^{7A+j}$. 
    Let $T(x_1, x_2, x_3) = x_1^{D_1}x_2^{D_2}x_3^{D_3}$ be a degree-7 monomial on three variables, i.e., $D_1 + D_2 + D_3 = 7$.
    For all $a, b, c\in [m']$, if $D_1,D_2,D_3$ are all non-zero, 
    we have
    \begin{align}
        \label{eq:intrablock_lem_1}
        \sum_{i=1}^n \Lambda^A_iT(h^{a}_i, h^b_i,h^c_i) = 
        \begin{cases}
            1 &\text{ if } \{a, b, c\} = \{3A+1, 3A+2, 3A + 3\}, \\
            0 &\text{ otherwise.}
        \end{cases}
    \end{align}
    If any of $D_1,D_2,D_3$ are zero, 
    we have
    \begin{align}
    \label{eq:intrablock_lem_2}
        \sum_{i=1}^n \Lambda^A_iT(h^{a}_i, h^b_i,h^c_i) = 0.
    \end{align}
\end{lemma}
\begin{proof}
    Let us expand the summation as 
    \begin{align}
        \sum_{i=1}^n \Lambda^A_iT(h^{a}_i, h^b_i,h^c_i)
        &= 
        \sum_{i=1}^n \sum_{j=1}^7 \Gamma^{7A+j}_i T(h^{a}_i, h^b_i,h^c_i) \\
        &= 
        \sum_{j=1}^7 \left(\sum_{i=1}^n \Gamma^{7A+j}_i T(h^{a}_i, h^b_i,h^c_i)\right) \\
        &=\sum_{j=1}^7 \left(\sum_{i=1}^n \Gamma^{7A+j}_i (h^{a}_i)^{D_1} (h^b_i)^{D_2}(h^c_i)^{D_3}\right).\label{eq:calculation_designed_intrablock_lemma}
    \end{align}
    The rows of $H$ are linear combinations of rows of $G$. 
    Recall the matrix $M$ in equation~\eqref{eq:matrix_M},
    we denote its rows as $M^1, M^2, M^3$, so that its $(i,j)$-th entry is $M^i_j$. 
    For $0\le B_a\le k'-1$, for $e_a\in [3]$,
    \begin{align}
        h^{3B_a+e_a} &= \sum_{j_a=1}^7 M^{e_a}_{j_a} g^{7B_a+j_a}.
    \end{align}
    For all $a > 3k'$, we simply have $h^a = g^{a + 4k'}$, as mentioned.
    
    We will now do a case analysis. Note that Equation~\eqref{eq:calculation_designed_intrablock_lemma} will be evaluated using the fact that $h^a$ are either linear combinations of rows of $G_1$ (if $a \leq 3k'$), or are simply rows of $G_0$ (if $a > 3k'$). It is therefore easy to see that if any of $a,b,c$ are greater than $3k'$, this sum must vanish using the addressable 7-orthogonality of $G$. For example suppose that $a, b, c > 3k'$ and all $D_1, D_2, D_3$ are non-zero. Then,
    \begin{align}
        \sum_{j=1}^7 \left(\sum_{i=1}^n \Gamma^{7A+j}_i (h^{a}_i)^{D_1} (h^b_i)^{D_2}(h^c_i)^{D_3}\right)
        &= \sum_{j=1}^7 \left(\sum_{i=1}^n \Gamma^{7A+j}_i 
        \left(g^{a+4k'}_i \right)^{D_1} 
        \left(g^{b+4k'}_i \right)^{D_2}
        \left(g^{c+4k'}_i \right)^{D_3}
        \right) \\
        &= \sum_{j=1}^7 \one_{a+4k'=b+4k'=c+4k'=7A+j} \label{eq:intrablock_lem_calc_indicator} \\
        &= \sum_{j=1}^7 0 = 0.
    \end{align}
    Note that equation~\eqref{eq:intrablock_lem_calc_indicator} is an application of the addressable $7$-orthogonality property of $G$.
    Moreover, if $D_3 = 0$, we simply replace the indicator function in equation~\eqref{eq:intrablock_lem_calc_indicator} with $\one_{a+4k'=b+4k'=7A+j}$ and the same result hold. Same is true if any of $D_1,D_2,D_3$ is zero. 
     
    The other cases, where we only assume that at least one of $a,b,c$ is more than $3k'$, follow for the same reason.

    The only non-trivial case is therefore $a,b,c \leq 3k'$. To this end, let us show Equations~\eqref{eq:intrablock_lem_1} and~\eqref{eq:intrablock_lem_2} with $a = 3B_a + e_a, b = 3B_b + e_b$, and $c = 3B_c + e_c$. 
    Beginning with the case that $D_1, D_2$ and $D_3$ are all non-zero, Equation~\eqref{eq:calculation_designed_intrablock_lemma} evaluates as follows. 
    \begin{align}
        &\sum_{j=1}^7 \left(\sum_{i=1}^n \Gamma^{7A+j}_i (h^{a}_i)^{D_1} (h^b_i)^{D_2}(h^c_i)^{D_3}\right) \\
        &= \sum_{j=1}^7 \left(\sum_{i=1}^n \Gamma^{7A+j}_i 
        \left(\sum_{j=1}^7 M^{e_a}_{j} g^{7B_a+j}_i \right)^{D_1} 
        \left(\sum_{j=1}^7 M^{e_b}_{j} g^{7B_b+j}_i \right)^{D_2} 
        \left(\sum_{j=1}^7 M^{e_c}_{j} g^{7B_c+j}_i \right)^{D_3}
        \right) \\
        &= \sum_{j=1}^7 
        \sum_{\substack{J_a\in [7]^{D_1}}}
        \sum_{\substack{J_b\in [7]^{D_2}}}
        \sum_{\substack{J_c\in [7]^{D_3}}}
        M^{e_a}_{J_a}M^{e_b}_{J_b}M^{e_c}_{J_c}
        \left(
        \sum_{i=1}^n \Gamma^{7A+j}_i g_i^{7B_a + J_a}g_i^{7B_b + J_b} g_i^{7B_c + J_c} \right) \label{eq:intrablock_lem_bigsum}
    \end{align}
    Here, we have introduced a short-hand to more easily expand products like $\left(\sum_{j=1}^7M_j^{e_a}g_i^{7B_a+j}\right)^{D_1}$. Indeed, given a tuple $J\in [7]^D$, we denote $M^e_J = \prod_{\delta = 1}^D M^e_{J_\delta}$ and $g^{7B+J}_i = \prod_{\delta = 1}^D g^{7B + J_\delta}_i$.
    Again by the addressable $7$-orthogonality property of $G$, $\sum_{i=1}^n \Gamma^{7A+j}_i g_i^{7B_a + J_a}g_i^{7B_b + J_b} g_i^{7B_c + J_c}$ is 1 when $B_a = B_b = B_c = A$, and when all the indices of $J_a, J_b$ and $J_c$ are equal to $j$.
    Therefore, equation~\eqref{eq:intrablock_lem_bigsum} simplifies to
    \begin{align}
        \sum_{j=1}^7 
        (M^{e_a}_{j})^{D_1} (M^{e_b}_{j})^{D_2} (M^{e_c}_{j})^{D_3}.
        \label{eq:intrablock_lem_mult_of_M}
    \end{align}
    In fact, one can check that Equation~\eqref{eq:calculation_designed_intrablock_lemma} evaluates to~\eqref{eq:intrablock_lem_mult_of_M} whenever $a=3B_a+e_a$, $b=3B_b+e_b$ and $c=3B_c+e_c$, irrespective of whether any of $D_1, D_2, D_3$ are zero.
    Let us now inspect the matrix $M$. 
    Observe that the entries in $M$ are all $0$ and $1$, which means 
    \begin{align}
        (M^{e_a}_j)^{D_1} = 
        \begin{cases}
            M^{e_a}_j &\text{ if $D_1\ne 0$,}\\
            1         &\text{ if } D_1 = 0.
        \end{cases}
        \label{eq:intrablock_lem_binary_exp}
    \end{align}
    Moreover, the three rows of $M$ satisfy the following properties: 
    \begin{align}
        \sum_{j=1}^7 M^{e}_j &= 0 \text{ for all $e\in [3]$, } \\
        \sum_{j=1}^7 (M^{e_1}\star M^{e_2})_j &= 0 \text{   for all $e_1,e_2\in [3]$,} \\
        \sum_{j=1}^7 (M^{e_1}\star M^{e_2}\star M^{e_3})_j &=
        \begin{cases}
            1 &\text{ if $\{e_1,e_2,e_3\} = \{1,2,3\}$,} \\
            0 &\text{ otherwise.}
        \end{cases}
        \label{eq:intrablock_lem_mult_prop_M}
    \end{align}
    We note that all summations in the above three equations are over the field $\FF_q$, and the above equations hold because $\FF_q$ has characteristic $2$.
    Combining equations~\eqref{eq:intrablock_lem_binary_exp} to~\eqref{eq:intrablock_lem_mult_prop_M}, we see that 
    \begin{align}
        \sum_{j=1}^7 
        (M^{e_a}_{j})^{D_1} (M^{e_b}_{j})^{D_2} (M^{e_c}_{j})^{D_3} = 1
    \end{align}
    if and only if $D_1,D_2,D_3$ are non-zero, and $\{e_a,e_b,e_c\} = \{1,2,3\}$. 
    This implies that $\{a,b,c\} = \{3A+1,3A+2,3A+3\}$, as desired. 
\end{proof}

We will now apply the usual matrix-to-code construction to $H$. Letting $\cH$ be the row space of $H$, we define the quantum CSS code $\qcode_M = \CSS(X, \cG_0; Z, \cH^\perp)$. 
\begin{proposition}
    $\qcode_M$ is a $[[n, 3k', d]]_q$ quantum code with $d\ge \min(d(\cG_0^\perp), d(\cG))$.
\end{proposition}
\begin{proof}
We may follow the argument in Section~\ref{sec:qcode_from_puncture} to establish these parameters, because the matrix $H$ satisfies Assumption~\ref{assump:independence}. Indeed, we know that our original matrix $G$, with the addressable 7-orthogonality property, satisfies Assumption~\ref{assump:independence} by Lemmas~\ref{lem:address_to_ortho} and~\ref{lem:ortho_to_assump}. Then, recalling that $H = \begin{pmatrix}
    H_1\\G_0
\end{pmatrix}$, we see that $H$ satisfies Assumption~\ref{assump:independence} as well, because the first three rows of $M$ are linearly independent.

Following the argument in Section~\ref{sec:qcode_from_puncture}, 
$\qcode_M$ encodes $3k'$ qudits. 
In addition, the distance of $\qcode_M$ satisfies $d\ge \min(d(\cG_0^\perp), d(\cH)) \ge \min(d(\cG_0^\perp),d(\cG))$, where we have used that $\mathcal{H}\subset\mathcal{G} \implies d(\mathcal{H}) \geq d(\mathcal{G})$.
\end{proof}

We claim that $\qcode_M$ support the logical implementation (using transversal physical $U_7$ gates) of a 3-qudit non-Clifford gate $W$ on a collection of triples of logical qubits. 
Specifically, $W$ is the gate
\begin{align}
    W \ket{\eta_1}\ket{\eta_2}\ket{\eta_3}
    &= \exp(i\pi \tr( \sum_{\{a,b,c\} = \{1, 2, 3\}}\eta_a^4\eta_b^2\eta_c
     ))
    \ket{\eta_1}\ket{\eta_2}\ket{\eta_3}, \\
    W^\beta \ket{\eta_1}\ket{\eta_2}\ket{\eta_3}
    &= 
    \exp(i\pi \tr( \beta \sum_{\{a,b,c\} = \{1, 2, 3\}}\eta_a^4\eta_b^2\eta_c
     ))
    \ket{\eta_1}\ket{\eta_2}\ket{\eta_3},
\end{align}
and it is straightforward to show that $W^\beta$ is in exactly the third level of the Clifford hierarchy for all large enough $t$ (recalling that the field $\mathbb{F}_q$ has size $q=2^t$) and for all $\beta \neq 0$.
For a logical computational basis state $\ol{\ket{u}}$ of $\qcode_M$, where $u \in \mathbb{F}_q^{3k'}$,
given 3 logical qudits $A, B, C\in [3k']$, we define the logical intra-block $W$ gate on logical qudits $A, B$ and $C$ to be 
\begin{equation}
    \ol{W^\beta [A,B,C]\ket{u}} = 
    \exp(i\pi \tr( \beta \sum_{\{a,b,c\} = \{A,B,C\}}u_a^4u_b^2u_c
     ))
     \ol{\ket{u}}. 
\end{equation}
Our main result of this section is the following theorem, which shows that the pre-designed triples of logical qudits may be addressed with the non-Clifford gate $W^\beta$.

\begin{theorem}
For all $0\le A\le k'-1$, $\beta\in \FF_q$,  we have
\begin{equation}
    \ol{W^\beta [3A+1, 3A+2, 3A+3]} = U_7^{\beta \Lambda^A}.
\end{equation}
\end{theorem}
\begin{proof}
For $u\in \FF_q^{3k'}$, the logical computational basis state of $\qcode_M$ corresponding to $u$ is
\begin{equation}
    \ol{\ket{u}} \propto \sum_{h\in \cG_0} \Ket{\sum_{a=1}^{3k'} u_{a}h^{a} + h}. 
\end{equation}
As always, to verify the action of the transversal gate $U_7^{\beta\Lambda^A}$ on the whole codepsace of $\qcode_M$, it suffices to verify its action on logical computational basis states $\ol{\ket{u}}$ by linearity. In turn, it is convenient to begin by considering some $v = \sum_{a=1}^{m'}v_ah^a$ for $v_a \in \mathbb{F}_q$, so that $v_i = \sum_{a=1}^{m'}v_ah^a_i$. We have
\begin{align}
    U_7^{\beta\Lambda^A}\ket{v} &= \prod_{i=1}^n\exp\left(i\pi\tr\left(\beta\Lambda_i^Av_i^7\right)\right)\ket{v}\\
    &=\exp\left(i\pi\tr(\beta\sum_{i=1}^n\Lambda_i^Av_i^7)\right)\ket{v}\label{eq:designed_intrablock_pre_calc}
\end{align}
where, as always, we have brought the sum inside the trace because $\tr$ is $\mathbb{F}_2$-linear. Then, since the field has characteristic 2, we have
\begin{align}\label{eq:seven_multinomial_one}
    v_i^7 = \left(\sum_{a \in [m']} v_{a}h^{a}_i \right)^7 &= \sum_{a\in [m']} v_a^7(h^a_i)^7 + \sum_{\substack{a,b \in [m']\\a \neq b}}[v_a^6v_b^1(h^a_i)^6(h^b_i)^1 + v_a^5v_b^2(h^a_i)^5(h^b_i)^2 + 
    v_a^4v_b^3(h^a_i)^4(h^b_i)^3] \\
    &+ \sum_{\substack{a,b,c \\\text{ pairwise distinct}}}v_a^4v_b^2v_c(h^a_i)^4(h^b_i)^2(h^c_i).\label{eq:seven_multinomial_two}
\end{align}
This statement may be seen either by evaluating the relevant multinomial coefficients and seeing which are odd and which are even or, alternatively, for a more complete proof, one may refer to Appendix A of~\cite{wills2024constant}.

From here, the sum inside the trace is
\begin{align}
    \sum_{i=1}^n\Lambda^A_i
     v_i^7
     &= \sum_{a\in [m']} v_a^7 \sum_{i=1}^n\Lambda^A_i(h^a_i)^7  \\
     &+ \sum_{\substack{a,b \in [m']\\a \neq b}}\left[v_a^6v_b^1 \sum_{i=1}^n\Lambda^A_i(h^a_i)^6(h^b_i)^1 + v_a^5v_b^2 \sum_{i=1}^n\Lambda^A_i(h^a_i)^5(h^b_i)^2 + 
    v_a^4v_b^3 \sum_{i=1}^n\Lambda^A_i(h^a_i)^4(h^b_i)^3\right] \\
    &+ \sum_{\substack{a,b,c \\\text{ pairwise distinct}}}v_a^4v_b^2v_c \sum_{i=1}^n\Lambda^A_i(h^a_i)^4(h^b_i)^2(h^c_i).
\end{align}
Invoking Lemma~\ref{lem:extended_orthogonality}, we get that 
\begin{align}
    \sum_{i=1}^n\Lambda^A_i
     v_i^7
     &=  \sum_{\substack{a,b,c \\\text{ pairwise distinct}}}v_a^4v_b^2v_c \one_{\{a,b,c\} = \{3A+1,3A+2,3A+3\}} \\
     &= \sum_{\{a,b,c\} = \{3A+1,3A+2,3A+3\}}v_a^4v_b^2v_c.\label{eq:designed_intrablock_mid_calc}
\end{align}
Combining Equations~\eqref{eq:designed_intrablock_pre_calc} and~\eqref{eq:designed_intrablock_mid_calc}, we find (leaving quantum states un-normalised for brevity) that
\begin{align}
    U_7^{\beta \Lambda^a}\ol{\ket{u}}
    &= \sum_{h\in \cG_0} 
     \exp(i\pi \tr( \beta \sum_{\{a,b,c\} = \{3A+1,3A+2,3A+3\}}u_a^4u_b^2u_c
     ))
    \Ket{\sum_{a=1}^{3k'} u_{a}h^{a} + h} \\
    &= \exp(i\pi \tr( \beta \sum_{\{a,b,c\} = \{3A+1,3A+2,3A+3\}}u_a^4u_b^2u_c
     )) \ol{\ket{u}}\\
     &= \ol{W^\beta [3A+1, 3A+2, 3A+3]} \ol{\ket{u}}. \qedhere
\end{align}
\end{proof}

\begin{remark}
    One may note that the only properties of $M$ that we used were that that its rows were linearly independent, that each rows has even weight, that each distinct pair has even overlap, and that the triple overlap of the three rows is odd. Note the similarity to the concept of triple intersection number in~\cite{zhu2023non}.
\end{remark}

\subsection{Single-Index Inter-Block Degree-\texorpdfstring{$\ell$}{} Polynomial Gates}\label{sec:address_ortho_poly_gates}

We conclude this section by proving that any polynomial gate $U_{q,P}^\beta$, as defined in Section~\ref{sec:prelim_non_clifford_addressability}, can be supported single-index addressably on $\qcode$ assuming the strong addressable orthogonal property for $G$. Such gates form an extremely flexible set, and are in fact complete for the set of diagonal qudit Clifford hierarchy gates with entries $\pm 1$ on the diagonal (see Appendix~\ref{sec:diagonal_Clifford}).

Indeed, let $P$ be an $e$-variate polynomial of degree-$\ell$ over $\mathbb{F}_q$.
We make the following definition as in the case of $\CLZ$ gates.
\begin{definition}[Single-index Inter-block Polynomial Gates]
    For vectors $h^1, \cdots, h^e\in \FF_q^n$, for $\beta\in \FF_q$ and address $A\in [n]$, we denote the single-index $U_P^\beta$ gate acting on qudit $A$ as 
    \begin{align}
        U_P^\beta[A]\ket{h^1}\cdots\ket{h^e} = \exp(i\pi\tr( \beta P(h^1_A, \cdots, h^e_A)))\ket{h^1}\cdots\ket{h^e}.
    \end{align}
    For a vector $\Gamma \in\FF_q^n$, we denote $U_P^\Gamma = \bigotimes_{i=1}^n U_P^{\Gamma_i}[i]$. 
\end{definition}

Let $G$ and the quantum code $\qcode$ be defined as in Section~\ref{sec:qcode_from_puncture}. For a general polynomial $P$ of degree $\ell$, we now require the strongly addressable $\ell$-orthogonality, as follows.

\begin{theorem}\label{thm:polynomial_gate_addressing}
Suppose $G$ has the strong addressable $\ell$-orthogonality property with respect to $\tau$. Then for all $A \in [k], \beta \in \mathbb{F}_q$, we have
\begin{align}
    \ol{U_P^{\beta\tau_A}[A]} = U_P^{\beta\Gamma^A},
\end{align}
with $\Gamma^A$ as in Equation~\eqref{eq:address_strong_mult_prop}, in the associated quantum code $\qcode$. In words, we can address the $A$-th logical qudits in the $e$ code blocks with a logical $U_P$ gate by acting transversally on the physical qudits.
\end{theorem}
The proof of this theorem follows essentially the same argument as in those of Theorems~\ref{thm:degree_ell_gate}.

\begin{proof}
    As in the proofs of Theorems~\ref{thm:degree_ell_gate} and~\ref{thm:single_index_CCZ},we start by consider a set of logical computational basis states of $\qcode$: $\ol{\ket{u_1}}, \cdots, \ol{\ket{u_e}}$. We also consider vectors $f^{(1)}, \ldots, f^{(e)}$, where
    \begin{align}
        f^{(p)}=\sum_{a=1}^mu_a^{(p)}g^a
    \end{align}
    for some $u_a^{(p)} \in \mathbb{F}_q$. Then, we have
    \begin{align}
        U_P^{\beta\Gamma^A}\ket{f^{(1)}}\ldots\ket{f^{(e)}} &=\exp\left(i\pi\tr(\beta\sum_{i=1}^n\Gamma_i^AP(f_i^{(1)}, \ldots, f_i^{(e)}))\right)\ket{f^{(1)}}\ldots\ket{f^{(e)}}.
    \end{align}
    Now we may write
    \begin{equation}
        P(f_i^{(1)}, \ldots, f_i^{(e)}) = \sum_{\substack{\text{monomials $T$} \\
    \deg(T)\le \ell} } c_T T(f_i^{(1)}, \ldots, f_i^{(e)})
    \end{equation}
    for some coefficients $c_T \in \mathbb{F}_q$, and using the strong addressable $\ell$-orthogonality with respect to $\tau$, we have, for any monomial $T$ of degree less than $\ell$,
    \begin{equation}
        \sum_{i=1}^n\Gamma_i^AT(f_i^{(1)}, \ldots, f_i^{(e)}) = \tau_AT(u_A^{(1)}, \ldots, u_A^{(e)}).
    \end{equation}
    This means that
    \begin{align}
        U_P^{\beta\Gamma^A}\ket{f^{(1)}}\ldots\ket{f^{(e)}} = \exp\left(i\pi\tr(\beta\tau_AP(u_A^{(1)}, \ldots, u_A^{(e)}))\right)\ket{f^{(1)}}\ldots\ket{f^{(e)}}
    \end{align}
    and, considering the forms of the logical computational basis states $\ol{\ket{u_p}}$, as given in Equation~\eqref{eq:qcode_comp_basis_logical}, by linearity this implies that
    \begin{align}
        U_P^{\beta\Gamma^A}\ol{\ket{u_1}\ldots\ket{u_e}} &= \exp\left(i\pi\tr(\beta\tau_AP(u_A^{(1)}, \ldots, u_A^{(e)}))\right)\ol{\ket{u_1}\ldots\ket{u_e}}\\
        &= \ol{U_P^{\beta\tau_A}[A]\ket{u_1}\ldots\ket{u_e}},
    \end{align}
    as required.
\end{proof}
\begin{remark}
    One situation in which the strong addressable $\ell$-orthogonality is not required is the case of a homogeneous polynomial $P$. Indeed, if $P$ is a degree-$\ell$ such polynomial, then the conclusion of Theorem~\ref{thm:polynomial_gate_addressing} holds even if $G$ is only addressable $\ell$-orthogonal.
\end{remark}

\section{Instantiations of the Addressable Orthogonality Framework}
\label{sec:instantiations_orthogonality}

We have seen in the previous section how matrices with various addressable orthogonality properties can give rise to quantum codes supporting interesting addressable transversal gates. We will now see how such matrices naturally arise as the generator matrices for punctured classical codes satisfying a \textit{multiplication property} defined below. Indeed, in Definition~\ref{def:multiplication_property}, we will define a regular, and strong $\ell$-multiplication property, and we will go on to see in Lemma~\ref{lem:mult_to_ortho} that if the classical code $C$ has each property, then the resulting generator matrix for the punctured code has the regular, and strong $\ell$-orthogonality property, respectively. Note that this multiplication property, and similar properties, have been well-studied in coding theory~\cite{randriambololona2015products}, and it is well-known that codes built from the evaluation of certain functions, such as Reed-Solomon and algebraic geometry codes, naturally satisfy such properties.

One of the key results of the last section is Lemma~\ref{lem:ortho_to_address}, which observes that a matrix with (strong) $(\ell+1)$-orthogonality also has addressable (strong) $\ell$-orthogonality. Immediately, this leads us to the realisation that to construct a quantum code supporting addressable $U_\ell^\beta$ gates, single-index inter-block $\CLZ$ gates, or, most generally, single-index $U_P^\beta$ gates, where $\deg(P) = \ell$, we need only construct matrices with $(\ell+1)$-orthogonality, which in turn requires classical codes with the $(\ell+1)$-multiplication property.

Since Reed-Solomon~\cite{reed1960polynomial} and algebraic geometry codes~\cite{goppa1982algebraico} naturally satisfy such multiplication properties, we are able to build a wide array of codes supporting these gates addressably. There is, however, one unfortunate feature of this, but we will be able to partially rectify it. In general, requiring a classical code to have the $\ell$-multiplication property for larger $\ell$ worsens its parameters, which will in turn worsen the parameters of the resulting quantum code $\qcode$. In particular, we know that the $\ell$-multiplication property on the classical code is enough to ensure transversality of gates such as $U_\ell^\beta$, $\CLZ$ and $U_P^\beta$, where $\deg(P) = \ell$, but to obtain their (single-index) addressability, if we are going to use Lemma~\ref{lem:ortho_to_address}, we must require $(\ell+1)$-multiplication of the classical code. This turns out to worsen the resulting parameters of the final quantum code by a factor of $\frac{\ell}{\ell+1}$, which is particularly unfortunate for the most important case of $\ell = 3$. However, we are able to amend this issue in the case of Reed-Solomon codes. By appealing to the internal structure of these codes, we in fact show in Section~\ref{sec:GRS} that a Reed-Solomon code with only $\ell$ multiplications can lead to a matrix with addressable strong $\ell$-orthogonality (rather than $\ell-1$), so that in the case of Reed-Solomon codes, we may bypass Lemma~\ref{lem:ortho_to_address}. This improves the parameters of the resulting quantum codes supporting the above gates for Reed-Solomon codes, which could be an important difference given that Reed-Solomon-based constructions could turn out to have some practical importance.

By contrast, we are not able to show a similar result for algebraic geometry codes. In Section~\ref{sec:AG}, we generalise the construction of~\cite{wills2024constant} to obtain algebraic geometry codes satisfying the $(\ell+1)$-multiplication property which, via Lemma~\ref{lem:ortho_to_address}, will suffice to construct codes with addressable gates such as $U_{\ell}^\beta$, $\CLZ$ and $U_P^\beta$, where $\deg(P) = \ell$. A result akin to what we achieve for Reed-Solomon codes, showing that an algebraic geometry code supporting only $\ell$ multiplications can be used to construct a matrix with addressable (strong) $\ell$-orthogonality, would likely improve the parameters of the resulting quantum code supporting the stated addressable gates. 
Further investigation into this is left to future work. Accordingly, the best achievable parameters for a quantum code supporting addressable gates such as $U_\ell^\beta$, $\CLZ$ and $U_P^\beta$ may be improvable in future. 

Note that by using the qudit code-to-qubit code conversion methods of Section~\ref{sec:concatenation}, all results obtained for codes on qudits (of dimension $2^t$) can be obtained for qubit codes, up to some loss in parameters. Note that, in certain places, if we claim improved parameters of a qudit quantum code, we are tacitly stating that the resultant qubit quantum code has improved parameters too, since converting the code to a qubit code via any given method would lead to improved parameters on the latter code as well. As usual, Reed-Solomon codes are only asymptotically good if one allows a growing field size, meaning the resulting qubit code will be not quite asymptotically good. By contrast, algebraic geometry codes are asymptotically good over a fixed finite field, leading also to asymptotically good codes over qubits.

It is interesting that the idea to obtain a matrix with addressable (strong) $\ell$-orthogonality via a code with (strong) $(\ell+1)$-multiplication is quite reminiscent of the folklore statement that a transversal gate in the $(\ell+1)$-th level of the Clifford hierarchy implies a transversal and addessable gate in the $\ell$-th level of the Clifford hierarchy; this folklore statement is discussed in Section~\ref{sec:prelim_non_clifford_addressability}. The reader may even wonder if these two observations are secretly the same thing. This is, however, not the case. Our observation that $(\ell+1)$-multiplications leads to the addressability of these gates is strictly stronger, and can lead to considerable savings in parameters. For example, consider the $U_\ell^\beta$ gate. As discussed in Appendix~\ref{sec:diagonal_Clifford}, for all large enough fields $\mathbb{F}_q$, the smallest $\ell$ for which $U_\ell^\beta$ is the second level of the Clifford hierarchy is $\ell = 3$, the smallest $\ell$ for which it is in the 3rd level is $\ell = 7$, the smallest to get into the 4th level is $\ell = 15$, and the smallest to get into the 5th level is $\ell = 31$. The values of $\ell$ required to get into successive levels of the Clifford hierarchy appear to be exponentially increasing, and we believe that further analysis may prove this to be true for all values of $\ell$. The consequence is that, by using the given folklore statement, building codes with addressable $U_\ell^\beta$ gates would incur far worse parameters than with our methods. Again, this is because building the classical code supporting $\ell$-multiplications for larger $\ell$ gives worse parameters on the classical code, and resultantly on the quantum code also. If one were to use this folklore statement, it would seem that the number of multiplications required to support an addressable $U_\ell^\beta$ gate in the $t$-th level of the Clifford hierarchy is $\ell = 2^{t+1}-1$. On the other hand, our method requires only $\ell=2^t-1$ multiplications, which leads to an improvement in the final parameters of the quantum code by a factor of $\frac{2^{t+1}+1}{2^t+1} \approx 2$. Our method of building codes with these addressable gates is therefore strictly stronger in this case, and we believe as the possibility of supporting more general diagonal Clifford hierarchy gates addressably opens up, it will in general be beneficial to consider our methods over the folklore result discussed in Section~\ref{sec:prelim_non_clifford_addressability}.

\subsection{Multiplication Property for Classical Codes}\label{sec:classical_code_multiplication}
Let us first discuss how we can obtain matrices $G$ with the (strong) $(\ell+1)$-orthogonality property, which will give us addressable (strong) $\ell$-orthogonality via Lemma~\ref{lem:ortho_to_address}. As observed in prior works~\cite{hastings2018distillation,krishna2019towards,wills2024constant,nguyen2024goodbinaryquantumcodes,golowich2024asymptoticallygoodquantumcodes,nguyen2024quantum}, such matrices can be constructed from classical codes with (strong) $(\ell+1)$-multiplication property, which we now define.

\begin{definition}\label{def:multiplication_property}
    We say that $C$ has the $\ell$-multiplication property with respect to $\sigma \in \bbF_q^N$ if, for all $c^1, \dots, c^\ell \in C$, it holds that
    \begin{align}
        \sum_{i = 1}^N \sigma_i c^1_i c^2_i \dots c^\ell_i = 0.
    \end{align}
    We say that $C$ has the strong $\ol{\ell}$-multiplication property with respect to $\sigma \in \bbF_q^N$ if, for all $\ell \le \ol{\ell}$, for all $c^1, \ldots, c^\ell \in C$, we have
    \begin{align}
        \sum_{i=1}^N \sigma_i c^1_i c^2_i \dots c^{\ell}_i = 0.
    \end{align}
    For convenience, we will say that $C$ has the (strong) $\ell$-multiplication property if there exists some $0 \neq \sigma \in \bbF_q^N$ for which the respective equations hold.
\end{definition}
\begin{remark}\label{rmk:all_ones_mult}
    In the situation that the code $C$ contains the all-ones vector, then the $\ell$-multiplication property with respect to $\sigma$ implies the strong $\ell$-multiplication property with respect to $\sigma$.
\end{remark}
Let us now show how classical codes with the (strong) multiplication properties can give rise to matrices with the (strong) orthogonality property, by puncturing their generator matrix.

\begin{lemma}\label{lem:mult_to_ortho}
    Consider a code $C \subseteq \mathbb{F}_q^N$ of dimension $m$. Given $k < m$, let $n = N-m$. If $C$ has the strong $\ol{\ell}$-multiplication property, for some $\ol{\ell} \geq 2$, with respect to $\sigma$, where $\sigma_1, \ldots, \sigma_k$ are not zero, then there is an $m \times n$ matrix $G$ with the strong $\ol{\ell}$-orthogonality property with respect to $(\Gamma,\tau)$, where $\Gamma = (\sigma_i)_{i=k+1}^N$ and $\tau = (\sigma_i)_{i=1}^k$. The same statement holds for the regular (not strong) $\ell$-multiplication property and the regular $\ell$-orthogonality property.
\end{lemma}

\begin{proof}
    Following the construction of Section~\ref{sec:qcode_from_puncture}, we pick a generator matrix $\tilde{G}$ of $C$ of the form
    \begin{align}
    \begin{pmatrix}
        I_k & G_1 \\
        0 & G_0.
    \end{pmatrix}
    \end{align}
    by row operations. Taking
    \begin{align}
        G = \begin{pmatrix}
            G_1 \\ 
            G_0
        \end{pmatrix},
    \end{align}
    let us check that this definition of $G$ satisfies the strong $\ol\ell$-orthogonality property. Denoting the rows of $\tilde{G}$ by $\tilde{g}^1, \dots, \tilde{g}^m$, and the rows of $G$ by $g^1, \dots, g^m$, the fact that $C$ has the strong $\ol{\ell}$-multiplication property, and that the rows of $\tilde{G}$ are codewords of $C$, means that for all $\ell \leq \ol{\ell}$ and all $a_1, a_2, \dots, a_\ell \in [m]$,
    \begin{align}
        \sum_{i=1}^N\sigma_i\tilde{g}_i^{a_1}\tilde{g}_i^{a_2}\ldots\tilde{g}_i^{a_l} = 0,
    \end{align}
    for some $\sigma \in \mathbb{F}_q^N$.
    Because the first $k$ columns of $\tilde{G}$ are of the form $\begin{pmatrix} I_k \\ 0 \end{pmatrix}$, this can be rewritten as
    \begin{align}
        \sigma_{a_1}\mathbbm{1}[1 \le a_1 = a_2 = \dots = a_\ell \le k] + \sum_{i = k+1}^{N}\sigma_i\tilde{g}^{a_1}_i \tilde{g}^{a_2}_i\cdots \tilde{g}^{a_\ell}_i  = 0,
    \end{align}
    or
    \begin{align}
        \sum_{i = 1}^{n} \sigma_{k+i}{g}^{a_1}_i {g}^{a_2}_i\cdots {g}^{a_\ell}_i = \begin{cases}
            \sigma_{a_1} & \text{ if $1 \le a_1 = a_2 = \dots = a_\ell \le m$} \\
            0 & \text{ otherwise}.
        \end{cases}
    \end{align}
    This latter manipulations assumes that the field $\mathbb{F}_q$ has characteristic 2, although the result naturally generalises with the appropriate minus signs when that is not the case.

    The proof for the regular (non-strong) case follows in the same way.
\end{proof}

\subsection{Generalised Reed-Solomon Codes}~\label{sec:GRS}

As discussed in the introduction to Section~\ref{sec:instantiations_orthogonality}, evaluation codes such as Reed-Solomon codes naturally satisfy (strong) $(\ell+1)$-multiplication. This can give rise to matrices with the (strong) $(\ell+1)$-orthogonality property via Lemma~\ref{lem:mult_to_ortho}, which can give rise to matrices with the addressable strong $\ell$-orthogonality property via Lemma~\ref{lem:ortho_to_address}. However, as mentioned, taking a Reed-Solomon code with $(\ell+1)$ multiplications rather than $\ell$ leads to a loss in the final parameters of the quantum code by a factor of $\frac{\ell}{\ell+1}$, which we would like to avoid. In this section, by leveraging the internal structure of Reed-Solomon codes, we will show that we can avoid this loss of parameters, and directly obtain addressable strong $\ell$-orthogonality, by requiring only $\ell$ multiplications in our Reed-Solomon code. Note that our construction bears some similarity to that of Section~\ref{sec:intrablock_GRS}.

\begin{theorem}\label{thm:GRS_gen_mat_address}
Fix any $\ol{\ell}, m, N, k\in \mathbb{N}$ and let $n \coloneq N-k$. Supposing that $m < 1+n/\ol\ell$ and $k \le m$, and that $q \geq N$ for some prime power $q$, we can construct an $m\times n$ matrix $G = \begin{pmatrix}
        G_1\\ G_0
    \end{pmatrix}$
with entries in $\FF_q$ such that
\begin{enumerate}[itemsep = 0pt]
    \item $G$ has addressable strong $\ol{\ell}$-orthogonality with respect to $\tau = \mathbf{1}^k$;
    \item $d(\cG)\ge n-m+1$ and $d(\cG_0^\perp) \ge m-k+1$.
\end{enumerate}
\end{theorem}
\begin{proof}
    Let $\bal$ be a collection of $N$ distinct points in $\FF_q$, consider the Reed-Solomon code $\RS_{N,m}(\bal)$. 
    Since the code has dimension $m\ge k$, we can take a generating matrix of the following form
    \begin{align}
        \tilde G = \begin{bmatrix}
            I_k & G_1 \\
            0 & G_0
        \end{bmatrix}.
    \end{align}
    As before, each row of $\tilde{G}$ corresponds to a polynomial $\tilde{g}^a$ of degree less than $m$.
    We may suppose that the columns of $\tilde{G}$ corresponds to evaluation points $\al_1, \cdots, \al_N$. 
    From Fact~\ref{fact:GRS-coeff}, for any $A\in [k]$, there exists vector $\Gamma^A\in \FF_q^{n}$ such that for all polynomials $f$ of degree less than $n$, we have
    \begin{align}\label{eq:GRS_interpolation_second_thm}
        f(\al_A) + \sum_{i=k+1}^N \Gamma^A_{i-k} f(\al_i) = 0.
    \end{align}
    As in equation~\eqref{eqn:A-addressable}, we take the coefficients of $f(\al_A)$ to be 1 by normalizing the other coefficients. Given any $\ell \le \ol\ell$ and any $a_1, \ldots, a_\ell \in [m]$, we have that $\tilde{g}^{a_1}\tilde{g}^{a_2}\ldots \tilde{g}^{a_\ell}$ is a polynomial of degree at most $\ell(m-1)\le\ol\ell(m-1) < n$. Thus, we may apply Equation~\eqref{eq:GRS_interpolation_second_thm}, to find that
    \begin{align}
        0 &= \tilde{g}^{a_1}(\alpha_A)\tilde{g}^{a_2}(\alpha_A)\ldots\tilde{g}^{a_\ell}(\alpha_A) + \sum_{i=k+1}^N\Gamma^A_{i-k}\tilde{g}^{a_1}(\alpha_{i})\tilde{g}^{a_2}(\alpha_i)\ldots\tilde{g}^{a_\ell}(\alpha_i)\\
        &=\mathbf{1}_{a_1 = a_2 = \ldots = a_\ell = A}+\sum_{i=1}^n\Gamma^A_i\tilde{g}^{a_1}(\alpha_{i+k})\tilde{g}^{a_2}(\alpha_{i+k})\ldots\tilde{g}^{a_\ell}(\alpha_{i+k}).
    \end{align}
    Therefore,
    \begin{equation}
        \sum_{i=1}^n\Gamma_i^Ag_i^{a_1}g_i^{a_2}\ldots g_i^{a_\ell} = \begin{cases}
            1 &\text{ if }a_1 = a_2 = \ldots = a_\ell = A\\
            0 &\text{ otherwise}
        \end{cases}.
    \end{equation}
    We have shown that $G$ has the addressable strong $\ol\ell$-orthogonality property with respect to the all-ones vector $\tau = \mathbf{1}^k$.
    
    Let $C$ denote the code with generator matrix $\tilde{G}$, and let $T$ denote the first $k$ indices. As in the case of equation~\eqref{eq:dist_qcode}, we have $\cG_0 = \shrT{C}$ and $\cG = \puncT{C}$. By Lemma~\ref{lem:short_punc}, Fact~\ref{fact:punctured_distance}, as well as Theorems~\ref{thm:GRS_dual} and~\ref{thm:GRS-rate-distance}, we have $d(\cG)\ge d(C) - k = n-m+1$ and $d(\cG_0^\perp) \ge d(C^\perp) - k = m-k+1$.
\end{proof}

With this done, our framework will finally bear fruit in the form of the following Reed-Solomon based construction of qudit quantum codes supporting transversal, addressable single qudit gates like $U_\ell^\beta$ and single-index addressable multi-qudit gates like $U_P^\beta$. As discussed, this result gives better parameters than would have been achieved otherwise if one had either simply used the folklore that transversal $\CCCZ$ implies single-index addressable $\CCZ$, or by simply using Lemma~\ref{lem:ortho_to_address}.

\begin{theorem}\label{thm:instantiation:GRS}
    Let $\ol\ell, m, n, k, q$ satisfy the assumptions of Theorem~\ref{thm:GRS_gen_mat_address}. Then there exists an $[[n,k,d]]_q$ quantum code $\qcode$ with $d \geq \min(n-m+1, m-k+1)$ such that, for any $e, \ell \leq \ol\ell$, $\qcode$ supports single-index addressable, transversal implementations of all $e$-variable, degree $\ell$ polynomial gates $U_P^\beta$ across $e$ code blocks. This includes single-index addressability of $\CLZ$, as well as addressability of single-qudit gates such as $U_\ell^\beta$. 

    Choosing $m = \left\lfloor n/\ol\ell\right\rfloor$ and $k=m/2$, we see that $\qcode$ may be chosen to be asymptotically good (over a growing qudit size) with rate and relative distance at least $1/(2\ol\ell)$.
\end{theorem}
\begin{remark}
    The surprising result here is that with Reed-Solomon codes, one may do \textit{as well} when it comes to addressability as for mere transversality, at least when it comes to the class of gates we study, namely diagonal Clifford hierarchy gates with entries $\pm 1$ on the diagonal. Indeed, one can check that these parameters match the parameters obtained for the transversality of the $\CCZ$ gate in~\cite{krishna2019towards,golowich2024asymptoticallygoodquantumcodes}. This is because we have used the interal structure of Reed-Solomon codes, in particular in Theorem~\ref{thm:instantiation:GRS}, and have thus improved the final parameters by a factor of $\frac{\ol\ell+1}{\ol\ell}$, as claimed, relative to what would have been achieved via either the use of Lemma 4.2, and by a potentiall much larger factor by using the folklore observation discussed in Section~\ref{sec:prelim_non_clifford_addressability}. As a key example, we note that if one is satisfied with a growing field size, then a single-index addressable $\CCZ$ gate may be obtained with rate and relative distance at least $1/6$.
\end{remark}
\begin{remark}
    Note that these are the same parameters as were obtained for the full addressability result for qudits, Theorem~\ref{thm:RS_addressable_CCZ}, and indeed the constructions are very similar, the present results offers greater generality in the construction at the expense of worse addressability properties (single-index addressability rather than addressability).
\end{remark}

\subsection{Algebraic Geometry Codes}~\label{sec:AG}

In this section, we describe how to use algebraic geometry codes to construct addressable strong $\ell$-orthogonal matrices over $\mathbb{F}_q$ for which the corresponding quantum code $\qcode$ is asymptotically good. Moreover, unlike for the Reed-Solomon codes, the field over which they are defined, $\mathbb{F}_q$, will be fixed, and will not grow with the field size, which means that the qubit codes that result from the conversions of Section~\ref{sec:concatenation} are asymptotically good.

As mentioned, however, unlike for the Reed-Solomon codes, we will not be able to obtain addressability using the internal structure of the code. Here, we will simply obtain $\ell$ orthogonality (for any $\ell$), and then appeal to Lemma~\ref{lem:ortho_to_address} to get matrices with addressable orthogonality. We will not make an attempt to optimise constants in this construction. We hope that future work, perhaps leveraging the internal structure of algebraic geometry codes, in a similar way to the previous section, will establish what parameters may be achieved for a quantum code supporting single-index addressable gates, such as the $\CCZ$ gates, on qubits, or on a field of fixed size.

We note that the body of the construction follows as a simple corollary of previous work. A general construction may be formulated for any algebraic function field with enough rational places as a simple extension of Section 5 of~\cite{wills2024constant}. On the other hand, here, to expedite the presentation, and because we make no attempt to optimise parameters, we simply appeal to the properties of particular function fields, in the same way as~\cite{nguyen2024goodbinaryquantumcodes}.
\begin{theorem}\label{thm:AG_instantiation_matrix}
    For any $\ol\ell \geq 2$, there exists some fixed $q=2^t$ and an explicit family of $[N,m,d]_q$ classical codes with the strong $\ol\ell$-multiplication property respect to the all-ones vector $\mathbf{1}^N$. Moreover, letting $d^\perp$ be the distance of the dual code, we have $m,d,d^\perp = \Theta(N)$, as well as $m-k, d-k, d^\perp-k = \Theta(N)$.
\end{theorem}
\begin{proof}
    The regular $\ol\ell$-multiplication property and the statement on the parameters follows as a simple generalisation of the argument of Theorem 3.6 of~\cite{nguyen2024goodbinaryquantumcodes}, which uses the results of~\cite{garcia1995tower,stichtenoth2009algebraic}, in particular by making the necessary replacements $3 \mapsto \ol\ell$. Moreover, since the code contains the all-ones vector, the regular $\ol\ell$-multiplication property implies the strong $\ol\ell$-multiplication property, as observed in Remark~\ref{rmk:all_ones_mult}.
\end{proof}
\begin{corollary}\label{cor:AG_instantiation_code}
    For any $\ol\ell \geq 2$, there exists some fixed $q=2^t$ and a family of matrices $G$ over $\mathbb{F}_q$ with the strong $\ol\ell$-orthogonality property with respect to all-ones vectors for which the corresponding quantum codes are explicit and asymptotically good. The same matrices satisfy addressable strong $(\ol\ell-1)$-orthogonality with respect to the all-ones vector.
\end{corollary}
\begin{proof}
    This follows from the construction of Section~\ref{sec:qcode_from_puncture} as well as Lemmas~\ref{lem:mult_to_ortho} and~\ref{lem:ortho_to_address}.
\end{proof}
Corollary~\ref{cor:AG_instantiation_code} may be used to instantiate any of the results of Section~\ref{sec:addressable_orthogonality_big_section} to produce asymptotically good qubit codes with various addressable gates. As an example that is of particular interest to us, we have the following.
\begin{theorem}\label{thm:designed_intrablock_AG_qudit}
    There exists some fixed $q = 2^t$ and a family of asymptotically good quantum codes over qudits of dimension $q$ such that a logical three-qudit non-Clifford gate may be executed on fixed intra-block (pairwise disjoint) triples of logical qudits via a depth-one physical circuit of $U_{q,7}^\beta$ gates.
\end{theorem}
\begin{proof}
    This is an immediate application of the construction of this section to the construction of Section~\ref{sec:designed_intrablock}.
\end{proof}

\section{Qubit Codes from Qudit Codes}
\label{sec:concatenation}

In this section, we will prove the main results of this work by converting the qudit codes constructed in Sections~\ref{sec:intrablock_GRS} to~\ref{sec:instantiations_orthogonality} to qubit codes, while preserving the addressability of their gates, and incurring only a small loss in their parameters. In large parts we follow the qudit-qubit conversion methods in~\cite{nguyen2024goodbinaryquantumcodes}, with some alterations for our particular situation.\footnote{We note that the indepedent work~\cite{golowich2024asymptoticallygoodquantumcodes} presents qudit-qubit conversion methods that may be viewed as a generalisation of those in~\cite{nguyen2024goodbinaryquantumcodes}. We leave to further work whether those constructions can apply in this situation, which could lead to a improvement in the parameters of the qubit code obtained.} Each logical/physical qudit will ultimately be embedded into a larger number of logical/physical qubits in such a way as to preserve the structure of the code that allows for gates to be implemented addressably and transversally in each case.

In general, there is a large amount of flexibility in these constructions, allowing one to build quantum codes with many different sets of addressable gates at different lengths with different implementation depths at the physical level for the addressable transversal gates. Any of the codes described in Sections~\ref{sec:intrablock_GRS} to~\ref{sec:instantiations_orthogonality} can be converted into qubit codes with similar addressability properties. However, to focus our attention, we will focus on achieving two things:
\begin{enumerate}
    \item \label{point:first_point_concatenation} We will convert the qudit code of Section~\ref{sec:intrablock_GRS} into a qubit code with the same addressability properties. Namely, any three logical qubits in one, two or three code blocks may be acted on with the qubit $\CCZ$ gate via a depth-one circuit of physical $\CCZ$ gates. Because the qudit codes were constructed from Reed-Solomon codes, and are therefore only asymptotically good with a growing field size, the resulting qubit codes with have parameters $\left[\left[n,\Omega\left(\frac{n}{\text{polylog}(n)}\right),\Omega\left(\frac{n}{\text{polylog}(n)}\right)\right]\right]_2$. Note that generalisations may easily be made to the other gates discussed such as $\CLZ$.
    \item We will convert the qudit codes of Theorem~\ref{thm:designed_intrablock_AG_qudit} to qubit codes with similar addressability properties. Namely, the resulting qubit code will have its logical qubits partitioned into sets of three. Each set of three may be addressed with a logical $\CCZ$ gate by acting with a depth-one circuit of physical $\mathsf{Z}$, $\mathsf{CZ}$ and $\mathsf{CCZ}$ gates. Again, generalisations may easily be made to other gates. While this addressability is much weaker than in Point~\ref{point:first_point_concatenation}, these qubit codes will be asymptotically good.
\end{enumerate}
These two conversions will be made in the coming sections. The latter conversion shares many ideas with the former and so will be stated more briefly. Therefore, we recommend that the coming sections are read in order.

\subsection{Near-Asymptotically Good Qubit Codes with Transversal and Addressable CCZ Gates}\label{sec:near_asymp_good_concatenation}

In this section, we will prove the main result of the paper, which is the following.

\begin{theorem}[Restatement of Theorem~\ref{thm:main_result_1}]\label{thm:main_addressable_CCZ}
    There exists an explicit family of quantum CSS codes over qubits with parameters
    \begin{equation}
        \left[\left[n,\Omega\left(\frac{n}{\text{polylog}(n)}\right),\Omega\left(\frac{n}{\text{polylog}(n)}\right)\right]\right]_2
    \end{equation}
    supporting a transversal, addressable $\CCZ$ gate. Specifically, any three logical qubits across one, two, or three blocks of the code may be addressed with the logical $\CCZ$ gate via a depth-one circuit of physical $\CCZ$ gates.
\end{theorem}

\subsubsection{Expanding Qudits into Qubits}

Let $\qcode_0$ denote the original code constructed in Section~\ref{sec:intrablock_GRS} for qudits of dimension $q=2^t$ with the addressable $\CCZ_q^\beta$ gate and denote its parameters as $[[n_0,k_0,d_0]]_q$, where $q=2^t = n_0^2$, and $k_0,d_0 = \Theta(n_0)$. The procedure begins by using the construction of Lemma~\ref{lem:self-dual-embedding} to convert $\mathcal{Q}_0$ into a qubit code $\mathcal{Q}_1$. We will use notation as in that lemma, as well as notation used throughout Section~\ref{sec:qudit_to_qubit_prelim}. In particular, we fix a self-dual basis $\{\alpha_i\}_{i=0}^{t-1}$ for $\mathbb{F}_{2^t}$ over $\mathbb{F}_2$, and let $\B : \mathbb{F}_{2^t}\to \mathbb{F}_2^t$ be the associated decomposition map. Let us call a block of $t$ qubits on $\qcode_1$ which corresponds to a $q$-dimensional qudit in $\qcode_0$ a \textit{$q$-block}.
Note that we have logical $q$-blocks and physical $q$-blocks.
As qudits are converted into qubits under this map,
the $\CCZ_q^\beta$ gate on qudits becomes the unitary $V^\beta \coloneq \theta^{\otimes 3}\left(\CCZ^\beta_q\right)$, acting on $3t$ qubits as
\begin{align}\label{eq:def_of_V}
     V^\beta \ket{\B(\eta_1)}\ket{\B(\eta_2)}\ket{\B(\eta_3)}
    &= (-1)^{\tr(\beta \eta_1\eta_2\eta_3)}\ket{\B(\eta_1)}\ket{\B(\eta_2)}\ket{\B(\eta_3)},\; \forall \eta_1,\eta_2,\eta_3,\beta\in \FF_q.
\end{align}
Therefore, if $\qcode_0$ is a qudit code on which logical $\CCZ$ is implemented by physical $\CCZ$ gates, then $\qcode_1$ is a qubit code on which logical $V$ gates acting on logical $q$-blocks are implemented by physical $V$ gates acting on physical $q$-blocks.

The key idea of~\cite{nguyen2024goodbinaryquantumcodes} is to modify the physical and logical spaces of the code $\qcode_1$ so that the qubit $V^\beta$ gates become products of qubit $\CCZ$ gates. At the logical level, this is achieved through a hard-coding of $t-1$ of the $t$ logical qubits in a logical $q$-block to $\ket{0}$, and with an appropriate choice of $\beta$, the action of $V^\beta$ will be shown to correspond to that of $\CCZ$, as is desired. At the physical level, the qubits are further encoded into a code that allows us to create the action of the $V^\beta$ gate using a depth-one circuit of $\CCZ$ gates. The code into which they are further encoded is created via multiplication-friendly embeddings, which we now introduce.

\subsubsection{Multiplication-Friendly Embeddings}

\begin{definition}[Definition~4.4 of~\cite{nguyen2024goodbinaryquantumcodes}]\label{def:MFE}
    For a field $\FF_q$ with $q = 2^t$, a degree-$\ell$ multiplication friendly embedding (MFE) of $\FF_q$ 
    into $\FF_2^r$ consists of a pair of functions $(\ME, \MEinv)$ and a collection of $r$-element permutations $\pi_1, \cdots, \pi_\ell$ such that
    \begin{enumerate}[itemsep = 0pt]
        \item $\ME: \FF_q\rightarrow \FF_2^r$ is an injective $\FF_2$-linear map, and $\MEinv: \FF_2^r\rightarrow \FF_q$ is a surjective $\FF_2$-linear map;
        \item For all $\al_1, \cdots, \al_\ell\in \FF_q$, we have     \begin{align}\label{eq:multiplication_friendly_encoding}
            \prod_{i=1}^\ell \al_i = \MEinv\left[ \pi_1(\ME(\al_1)) \star \pi_2(\ME(\al_2)) \star \cdots
            \star \pi_\ell(\ME(\al_\ell)) \right].
        \end{align}
    \end{enumerate}
\end{definition}

\begin{lemma}[Lemma~4.5 of~\cite{nguyen2024goodbinaryquantumcodes}]\label{lem:3MFE}
    For all $t$, there exists a degree-3 MFE of $\FF_{2^t}$ into $\FF_2^{r}$ where $r = t^3$.
\end{lemma}
The use of the following corollary will be evident by the end of the next section.
\begin{corollary}\label{cor:duplicated_MFE}
    Given a degree-$\ell$ MFE with maps $\ME, \MEinv$ and permutations $\pi_1,\cdots, \pi_\ell$, we can construct another degree-$\ell$ MFE by setting $\ME'(\gamma) = (\ME(\gamma),\ME(\gamma),\ME(\gamma),\ME(\gamma))$ and changing the other functions accordingly.
\end{corollary}
\begin{proof}
    We set $\MEinv':\FF_2^{4r}\rightarrow \FF_q$ as $\MEinv'(u) = \MEinv(u_1,\cdots, u_r)$, and we let $\pi_i'$ be $\pi_i$ acting on 4 copies of the embedding independently. Equation~\eqref{eq:multiplication_friendly_encoding} follows straightforwardly.
\end{proof}

\subsubsection{Logical and Physical Modifications to the Code  \texorpdfstring{$\mathcal{Q}_1$}{}}
Let us now describe how exactly we modify the code $\qcode_1$, which we note has parameters $[[n_1 = n_0t, k_1 = k_0t, d_1]]_2$. As mentioned, there will be a step that modifies its ``logical space'', and a step that modifies its ``physical space''. It will become more clear what this means.

Starting with the logical modification, we may enumerate the logical $Z$-operators of $\qcode_1$ as $\{\ol Z_{b,j}\}$ for $b\in [k_0], j\in [t]$, so that, for each $b$, the operators $\ol Z_{b,1}, \cdots, \ol Z_{b,t}$ are the logical $Z$ operators for the $b$-th block of $t$ logical qubits.
\begin{lemma}\label{lem:concatenation_logical}
    For every $b\in [k_0]$, let us add the $Z$-operators $\ol Z_{b, 2}, \cdots, \ol Z_{b,t}$ into the $Z$-stabilizers of $\qcode_1$ to create the code $\qcode_2$. 
    In the terminology of quantum codes, we ``gauge-fix'' the last $t-1$ logical qubits in every block of $t$ logical qubits. Then
    \begin{enumerate}[itemsep = 0pt]
        \item $\qcode_2$ has parameters  $[[n_2 = n_1, k_2 = k_0, d_2\ge d_1]]_2$;
        \item For every $\beta\in\FF_q$, suppose a physical circuit $\phyC{\beta}$ acting on the physical qubits of a code state of $\qcode_1$ implements a logical intra-block $V^\beta$ gate acting on logical $q$-blocks $A, B, C\in [k_0]$. In other words, suppose
        \begin{align}
            \ol{V^\beta[A, B, C]} = \phyC{\beta}.
        \end{align}
        Then the physical circuit $\phyC{\al_1^{-1}}$ acting on a code state of $\qcode_2$ implements a logical intra-block $\CCZ$ gate on logical qubits $A, B, C$ of $\qcode_2$. 
        Similar conclusions hold for logical inter-block $V$ gates.
    \end{enumerate}
\end{lemma}
\begin{proof}
    As mentioned, the described procedure simply hardcodes the last $t-1$ qubits of every logical $q$-block of $\qcode_1$ to $\ket{0}$. 
    The code distance does not decrease from such a procedure.
    Therefore, a logical $V$ gate defined as in equation~\eqref{eq:def_of_V} acting on three logical $q$-blocks evaluates as, for every $b_1,b_2,b_3\in \FF_2$,
    \begin{align}
        V^{\beta}
        \ket{b_1, 0, \cdots, 0}
        \ket{b_2, 0, \cdots, 0}
        \ket{b_3, 0, \cdots, 0}
        &= V^{\beta}\ket{\B(b_1\al_1)}\ket{\B(b_2\al_1)}\ket{\B(b_3\al_1)} \\
        &= (-1)^{\tr(\beta b_1\al_1b_2\al_1b_3\al_1)}\ket{\B(b_1\al_1)}\ket{\B(b_2\al_1)}\ket{\B(b_3\al_1)} \\
        &= (-1)^{b_1b_2b_3\tr(\beta \al_1\al_1\al_1)}\ket{\B(b_1\al_1)}\ket{\B(b_2\al_1)}\ket{\B(b_3\al_1)}.
    \end{align}
    By setting $\beta = \al_1^{-1}$, we see that 
    \begin{align}
        V^{\al_1^{-1}}
        \ket{b_1, 0, \cdots, 0}
        \ket{b_2, 0, \cdots, 0}
        \ket{b_3, 0, \cdots, 0}
        &= (-1)^{b_1b_2b_3\tr( \al_1\al_1)}\ket{\B(b_1\al_1)}\ket{\B(b_2\al_1)}\ket{\B(b_3\al_1)} \\
        &= (-1)^{b_1b_2b_3}\ket{\B(b_1\al_1)}\ket{\B(b_2\al_1)}\ket{\B(b_3\al_1)},
    \end{align}
    where the final equality follows from the self-duality of our chosen basis.
\end{proof}
\begin{remark}
    Clearly, hardcoding this many logical qubits to $0$ is quite wasteful. It would be interesting to see if it is possible to recover more addressable logical qubits per logical qudit, for example as was done for transversality by using reverse multiplication friendly encodings in~\cite{nguyen2024goodbinaryquantumcodes}, or via quantum multiplication friendly codes~\cite{golowich2024asymptoticallygoodquantumcodes}. Indeed, we note it is certainly possible to do this and retain a certain amount of addressability. However, whether it is possible to do this and retain full addressability of the $\CCZ$ gate, which is the primary aim of this work, is not clear.
\end{remark}

We now modify the physical space of $\qcode_2$. 
The high level idea is simple: we will construct, using a multiplication-friendly embedding, a quantum code $\qoutcode$ which encodes every $q$-block of $t$ qubits into $r$ qubits such that a $V$ gate acting on three $q$-blocks may be implemented by $\CCZ$ gates acting on $3r$ qubits. The use of the multiplication-friendly embedding allows a larger number of bits to emulate the multiplication in $\mathbb{F}_q$ via the star product.

Consider a degree-3 MFE as in Definition~\ref{def:MFE} with $\FF_2$-linear maps $(\ME, \MEinv)$ and a collection of $r$-element permutations $\pi_1, \pi_2, \pi_3$.
Recall that $\ME: \FF_q\rightarrow \FF_2^r$ is injective and $\MEinv: \FF_2^r\rightarrow \FF_q$ is surjective.
Let
\begin{align}\label{eq:def_of_enc}
    \enc = \ME\circ \B^{-1}: \FF_2^t\rightarrow \FF_2^r,
\end{align}
which is an injective $\FF_2$-linear map, and let $\coutcode = \im(\enc)$ be a classical code.
$\coutcode$ may also be considered a quantum code $\qoutcode$ with parameters $[[r,t,1]]_2$ via $\overline{\ket{u}}\coloneq\ket{\phi(u)}$. From here, we have the following lemma.

\begin{lemma}\label{lem:concatenation_physical}
    Consider three standard basis states $\ol{\ket{u}} = \ket{\phi(u)}$, $\ol{\ket{v}} = \ket{\phi(v)}$, $\ol{\ket{w}} = \ket{\phi(w)}$ of $\qoutcode$ for $u,v,w\in\FF_2^t$. 
    For all $\beta\in \FF_q$, the logical gate $V^\beta$ can be implemented by a depth-1 $\CCZ$ circuit.
\end{lemma}
\begin{proof}
    We are aiming to implement the logical gate $V^\beta$, where
    \begin{align}
        \ol{V^\beta\ket{u}\ket{v}\ket{w}}
        &= (-1)^{\tr(\beta\B^{-1}(u)\B^{-1}(v)\B^{-1}(w))} \ol{\ket{u}}\ol{\ket{v}}\ol{\ket{w}}.
    \end{align}
    We emphasise that the overlines here indicate logical gates and states of the code $\qoutcode$.
    From equation~\eqref{eq:multiplication_friendly_encoding}, we see that 
    \begin{align}
        \B^{-1}(u)\B^{-1}(v)\B^{-1}(w)
        &= \MEinv\left[ 
        \pi_1(\ME(\B^{-1}(u))) \star
        \pi_2(\ME(\B^{-1}(v))) \star
        \pi_3(\ME(\B^{-1}(w)))
        \right] \\
        &= \MEinv\left[ 
        \pi_1(\enc(u)) \star
        \pi_2(\enc(v)) \star
        \pi_3(\enc(w))
        \right].
    \end{align}
    Moreover, suppose $\beta = \sum_{i\in [t]} b_i\al_i$, let $\vec{b} = (b_1,\cdots, b_t)$. 
    Then for all $\gamma\in \FF_q$, since $\{\al_i\}$ is a self-dual basis we have $\tr[\beta\gamma] = \vec{b}\cdot\B(\gamma)$.
    Therefore,
    \begin{align}
        \tr(\beta \B^{-1}(u)\B^{-1}(v)\B^{-1}(w))
        &=  \vec{b}\cdot \left( \B(\MEinv\left[ 
        \pi_1(\enc(u)) \star
        \pi_2(\enc(v)) \star
        \pi_3(\enc(w))
        \right]) \right).
    \end{align}
    Note that $\B\circ \MEinv: \FF_2^r\rightarrow \FF_2^t$ is a surjective $\FF_2$-linear function. 
    Let $T$ be the matrix representation of this function, and let $\Lambda(\beta) = \vec{b}T\in \FF_2^r$ where $\vec{b}$ is considered a row vector and $\vec{b}T$ is simply matrix multiplication.
    Then we have
    \begin{align}
        \tr(\beta \B^{-1}(u)\B^{-1}(v)\B^{-1}(w))
        &=  \sum_{j=1}^r \Lambda(\beta)_j \left[ 
        \pi_1(\enc(u)) \star
        \pi_2(\enc(v)) \star
        \pi_3(\enc(w))
        \right]_j \\
        &=  \sum_{j=1}^r \Lambda(\beta)_j 
        \pi_1(\enc(u))_j 
        \pi_2(\enc(v))_j 
        \pi_3(\enc(w))_j
         \\
         &=  \sum_{j=1}^r \Lambda(\beta)_j 
        \enc(u)_{\pi_1^{-1}(j)}
        \enc(v)_{\pi_2^{-1}(j)}
        \enc(w)_{\pi_3^{-1}(j)}.
    \end{align}
    We can create this phase with a depth-1 $\CCZ$ circuit, namely
    \begin{align}
        \ol{V^\beta\ket{u}\ket{v}\ket{w}}
        &= \bigotimes_{j=1}^r \CCZ^{\Lambda(\beta)_j}[\pi_1^{-1}(j), \pi_2^{-1}(j), \pi_3^{-1}(j)]\ket{\phi(u)}\ket{\phi(v)}\ket{\phi(w)}.
    \end{align}
    Since $\Lambda(\beta)_j \in \FF_2$, $\CCZ^{\Lambda(\beta)_j}$ is simply $\CCZ$ if $\Lambda(\beta)_j = 1$ and identity otherwise. 
    This circuit has depth-1 since $\pi_1,\pi_2,\pi_3$ are permutations.
\end{proof}
\begin{corollary}\label{cor:duplicated_qoutcode}
Using the duplicated MFE as in Corollary~\ref{cor:duplicated_MFE}, a depth-4 logical circuit of $V$ gates acting on code states of $\qoutcode$ can be implemented by a depth-1 $\CCZ$ physical circuit.
\end{corollary}
\begin{proof}
    We construct the encoding function $\phi$ as in equation~\eqref{eq:def_of_enc}, and take the classical code $\coutcode$, which again may be viewed as a quantum code $\qoutcode$. 
    Following the same proof as above, the four $V$ gates acting on any $q$-block can each be implemented with a collection of non-overlapping CCZ circuit acting on $r$ of the $4r$ qubits of the encoded block.
    Therefore, we can arrange them all into a depth-1 $\CCZ$ circuit.
\end{proof}

We can now in a position to prove Theorem~\ref{thm:main_addressable_CCZ}.

\begin{proof}[Proof of Theorem~\ref{thm:main_addressable_CCZ}]
    Let $\qcode_0$ be the qudit CSS code from Theorem~\ref{thm:RS_addressable_CCZ} with parameters $[[n_0,k_0,d_0]]_q$.
    We know that $k_0,d_0 = \Theta(n_0)$, where $q=2^t = n_0^2$.
    Applying Lemma~\ref{lem:self-dual-embedding}, we get a qubit CSS code $\qcode_1$ with parameters $[[n_1 = n_0t, k_1 = k_0t, d_1\ge d_0]]_2$.
    $\qcode_1$ supports a depth-4 implementation of addressable intra-block $V$ gates, and a depth-1 implementation of addressable inter-block $V$ gates. 
        
    Next we apply Lemma~\ref{lem:concatenation_logical} to get a qubit CSS code $\qcode_2$ with parameters $[[n_2 = n_0t, k_2 = k_0, d_2\ge d_0]]_2$.
    On $\qcode_2$, the logical actions of $V$ gates are translated into $\CCZ$ gates, which means addressable intra-block $\CCZ$ gates and inter-block $\CCZ$ gates are implemented by circuits of $V$ gates of depth four and depth one, respectively. 
    
    Finally, if we want addressable inter-block $\CCZ$ gates to be implemented by a depth-1 $\CCZ$ circuit, we can encode every physical $q$-block of $\qcode_2$ into a $[[t^3,t,1]]_2$ quantum code $\qoutcode$ constructed from the degree-3 MFE in Lemma~\ref{lem:3MFE}.
    The resulting qubit CSS code $\qcode_3$ has parameters $[[n_3 = n_0t^3, k_2 = k_0, d_3\ge d_0]]_2$.
    The transversal implementation of addressable inter-block $\CCZ$ gates is guaranteed by Lemma~\ref{lem:concatenation_physical}.
    If we additionally want addressable intra-block $\CCZ$ gates to be implemented by a depth-1 $\CCZ$ circuit, we concatenate using the code $\qoutcode$ in Corollary~\ref{cor:duplicated_qoutcode} instead and obtain the resulting qubit code $\qcode_3'$ with parameters $[[n_3 = 4n_0t^3, k_2 = k_0, d_3\ge d_0]]_2$. Since $t = 2\log n_0$, we obtain the desired result.
\end{proof}

\begin{remark}
    Being able to execute the intra-block $\CCZ$ gates via a depth one circuit comes at the expense of a longer code. As is clear from Corollary~\ref{cor:duplicated_qoutcode}, and as we mention in the proof, one could instead obtain a shorter code at the expense of greater implementation depth. Now, as stated in Remark~\ref{rmk:high_degree_poly_gates}, the construction is generalisable to higher-degree polynomial gates such as $\CLZ$.\footnote{This is possible because, for example, the argument of Lemma 4.5 of~\cite{nguyen2024goodbinaryquantumcodes} can be easily generalised to obtain higher degree MFEs. In particular, there exists an explicit degree-$\ell$ MFE of $\mathbb{F}_{2^t}$ into $\mathbb{F}_2^r$ for which $r = t^\ell$. It is also possible to use the quantum multiplication friendly codes of~\cite{golowich2024asymptoticallygoodquantumcodes}.} In this case, the greater depth needed for the intra-block implementation could again be traded off into a longer code length, if desired.

    One can even take these thoughts a step further and note that if Lemma~\ref{lem:concatenation_physical} (and Corollary~\ref{cor:duplicated_qoutcode}) were not applied, then the logical, addressable $\CCZ$ gate would be implementable by a circuit of physical $V^\beta$ gates, implying a much shorter code, but with a higher (but still low) implementation depth. Indeed, the $V^\beta$ gates would be implementable by a circuit of $\mathsf{Z}, \mathsf{CZ}$ and $\CCZ$ gates.\footnote{This is true because $V^\beta = \theta^{\otimes 3}(\CCZ_q^\beta)$ must be a diagonal gate in exactly the third level of the Clifford hierarchy with $\pm 1$ on the diagonal. A diagonal gate with entries $\pm 1$ on the diagonal must be implementable via a circuit of $\CLZ$ gates~\cite{houshmand2014decomposition}, and so the fact that $V^\beta$ can be implemented by a circuit of (qubit) $\mathsf{Z}$, $\mathsf{CZ}$ and $\CCZ$ gates follows from the fact that $\CLZ$ gates are independent (meaning no non-trivial combination of them acts identically), and $\CLZ$ is in exactly the $\ell$-th level of the Clifford hierarchy.}
\end{remark}

\subsection{Asymptotically Good Qubit Codes with Transversal and Addressable CCZ Gates on \protect\linebreak Pre-Designed Intrablock Triples}

In this section, we will prove the second main result of the paper, which is the following.
\begin{theorem}[Restatement of Theorem~\ref{thm:main_result_2}]\label{thm:main_asymp_good}
    There exists an explicit family of quantum CSS codes over qubits with parameters
    \begin{equation}
        \left[\left[n,\Theta(n),\Theta(n)\right]\right]_2,
    \end{equation}
    meaning that it is asymptotically good, supporting a transversal, addressable $\CCZ$ gate on pre-designed (pairwise disjoint) intra-block triples of logical qubits. Specifically, the logical qubits in one block of the code are partitioned into sets of three, and one may address any such set with the $\CCZ$ gate via a depth-one circuit of physical $\mathsf{Z}$, $\mathsf{CZ}$ and $\mathsf{CCZ}$ gates.
\end{theorem}

The steps taken in this section are similar at a high level to those of Section~\ref{sec:near_asymp_good_concatenation}, with certain differences at a lower level. We denote by $\qcode_0$ the qudit codes of Theorem~\ref{thm:designed_intrablock_AG_qudit}, writing their parameters as $[[n_0,3k_0,d_0]]_q$, where $q=2^t$ is fixed, and $k_0, d_0 = \Theta(n_0)$. We note that the code $\qcode_0$ has $k_0$ sets of 3 logical qudits. Applying once more the construction of Lemma~\ref{lem:self-dual-embedding}, we obtain a code $\qcode_1$ with parameters $[[n_0t,3k_0t,d_1 \geq d_0]]_2$. This is done, as usual, via a self-dual basis $\{\alpha_i\}_{i=0}^{t-1}$ for $\mathbb{F}_{2^t}$ over $\mathbb{F}_2$, and with the decomposition map $\mathcal{B}:\mathbb{F}_{2^t}\to\mathbb{F}_2^t$. Recall that in the code $\qcode_0$, it was possible to address any of the $k_0$ sets of 3 logical qudits with the gate $W^\beta$ via a depth-one circuit of physical $U_7^\gamma$ gates, where
\begin{equation}
    W^\beta\ket{\eta_1}\ket{\eta_2}\ket{\eta_3} = \exp\left(i\pi\tr\left(\beta\sum_{\{a,b,c\} \in \{1,2,3\}}\eta_a^4\eta_b^2\eta_c\right)\right)\ket{\eta_1}\ket{\eta_2}\ket{\eta_3}.
\end{equation}
Under the embedding into qubits, these gates become a $3t$-qubit, and a $t$-qubit gate, respectively, which we denote
\begin{align}
    \Pi^\beta &\coloneq \theta^{\otimes 3}(W^\beta)\\
    \Sigma^\gamma &\coloneq \theta(U_7^\gamma).
\end{align}
The $k_0$ sets of 3 logical qudits in $\qcode_0$ become $k_0$ sets of 3 logical $q$-blocks in $\qcode_1$. Then, in the code $\qcode_1$, a physical circuit of $\Sigma$ gates acting transversally on the physical $q$-blocks executes a logical $\Pi^\beta$ gate between one of these sets of three logical $q$-blocks; one may address any set of three logical $q$-blocks with the logical $\Pi^\beta$ gate. Enumerating these sets of three $q$-blocks in $\qcode_1$ by $A \in [k_0]$, let $\overline{\Pi^\beta[A]}$ denote the logical $\Pi^\beta$ gate acting on the $A$-th set of 3 logical $q$-blocks. In the same way as last time, it is convenient to write
\begin{equation}
    \overline{\Pi^\beta[A]} = C(\beta),
\end{equation}
where $C(\beta)$ is the physical circuit of $\Sigma$ gates executing $\Pi^\beta$ gate on the $A$-th set of three logical $q$-blocks. Again in a similar way to last time, we will now hardcode $t-1$ of the $t$ qubits in each logical $q$-block to $\ket{0}$ to create a code $\qcode_2$ with parameters $[[n_0t,3k_0,d_2 \geq d_0]]_2$. However, this time, the hardcoding will go slightly differently. Given any set of three logical $q$-blocks labelled by $A \in [k_0]$, in the first of these $q$-blocks we hardcode all but the first logical qubit to $\ket{0}$, whereas in the second logical $q$-block we hardcode all but the second logical qubit to $\ket{0}$, and finally in the third logical $q$-block we hardcode all but the third logical qubit to $\ket{0}$. The gate $\Pi^\beta$ acts on such a set of three $q$-blocks as
\begin{align}
    \Pi^\beta&\ket{b_1, 0,0, \ldots, 0}\ket{0,b_2,0, \ldots, 0}\ket{0,0,b_3, \ldots, 0} = \Pi^\beta\ket{\mathcal{B}(b_1\alpha_1)}\ket{\mathcal{B}(b_2\alpha_2)}\ket{\mathcal{B}(b_3\alpha_3)}\\
    &=\exp\left(i\pi\tr\left(\beta\sum_{\{a,b,c\} \in \{1,2,3\}}(b_a\alpha_a)^4(b_b\alpha_b)^2(b_c\alpha_c)\right)\right)\ket{\mathcal{B}(b_1\alpha_1)}\ket{\mathcal{B}(b_2\alpha_2)}\ket{\mathcal{B}(b_3\alpha_3)}\\
    &=\exp\left(i\pi b_1b_2b_3\tr\left(\beta\sum_{\{a,b,c\} \in \{1,2,3\}}\alpha_a^4\alpha_b^2\alpha_c\right)\right)\ket{\mathcal{B}(b_1\alpha_1)}\ket{\mathcal{B}(b_2\alpha_2)}\ket{\mathcal{B}(b_3\alpha_3)},
\end{align}
where going into the third line we have used the fact that $x^2 = x$ for any $x \in \mathbb{F}_2$ and that $\tr$ is an $\mathbb{F}_2$-linear function. Given that we ultimately wish to execute the logical $\CCZ$ gate, for which we want the phase $(-1)^{b_1b_2b_3}$, we define $\hat{\beta}$ to be any element of $\mathbb{F}_q$ satisfying
\begin{equation}
    \tr\left(\hat{\beta}\sum_{\{a,b,c\} \in \{1,2,3\}}\alpha_a^4\alpha_b^2\alpha_c\right) = 1.\footnote{Note that this is always possible as long as $\sum_{\{a,b,c\} \in \{1,2,3\}}\alpha_a^4\alpha_b^2\alpha_c \neq 0$, because the kernel of $\tr$ is exactly half of $\mathbb{F}_q$. Strictly speaking, we need to demonstrate the existence of a self-dual basis $\{\alpha_i\}_{i=0}^{t-1}$ such that $\sum_{\{a,b,c\} \in \{1,2,3\}}\alpha_a^4\alpha_b^2\alpha_c \neq 0$, although in fact this is not necessary. The reason is that, while the self-dual basis is most convenient to expand in, one can in fact expand in terms of any basis with a little added complication --- see Section 8.1.2 of~\cite{gottesman2016surviving} --- and so all we actually need is three elements $\alpha_1, \alpha_2, \alpha_3 \in \mathbb{F}_q$ that are linearly independent over $\mathbb{F}_2$ such that $\sum_{\{a,b,c\} \in \{1,2,3\}}\alpha_a^4\alpha_b^2\alpha_c \neq 0$, and this is easily seen to be obtainable.}
\end{equation}
Note that a hard-coded logical $q$-block encodes exactly one logical qubit. Then, on a set of three hardcoded logical $q$-blocks, the logical gate $\Pi^{\hat{\beta}}$ acts as a (qubit) $\CCZ$ gate on the encoded logical qubits.

Having modified the logical space of $\qcode_1$ to obtain the code $\qcode_2$, it remains to modify the physical space. We want to be able to enact the gate $\Sigma^\gamma$ on physical $q$-blocks of $\qcode_2$ via a depth-one circuit of $\CLZ$ gates, for $\ell \leq 3$. The action of the gate $\Sigma^\gamma$ is as follows. With $u \in \mathbb{F}_2^t$, we note that $\mathcal{B}^{-1}(u) = \sum_{i=0}^{t-1} = u_i\alpha_i$, and
\begin{align}
    \Sigma^\gamma\ket{u} &= \exp\left[i\pi\tr(\gamma(\mathcal{B}^{-1}(u))^7)\right]\ket{u}
\end{align}
where
\begin{align}
    \left(\mathcal{B}^{-1}(u)\right)^7 = \sum_{i=0}^{t-1}u_i\alpha_i^7 &+ \sum_{i < j}u_iu_j\left(\alpha_i^6\alpha_j + \alpha_i^5\alpha_j^2+\alpha_i^4\alpha_j^3+\alpha_i^3\alpha_j^4+\alpha_i^2\alpha_j^5+\alpha_i\alpha_j^6\right)\\
    &+\sum_{i<j<k}u_iu_ju_k\left(\alpha_i^4\alpha_j^2\alpha_k+\alpha_i^4\alpha_j\alpha_k^2+\alpha_i^2\alpha_j^4\alpha_k+\alpha_i^2\alpha_j\alpha_k^4+\alpha_i\alpha_j^4\alpha_k^2+\alpha_i\alpha_j^2\alpha_k^4\right)
\end{align}
where we have used the same formula as in Equations~\eqref{eq:seven_multinomial_one} and~\eqref{eq:seven_multinomial_two}, albeit slightly re-written. Therefore, the desired phase is
\begin{equation}
    \exp\left[i\pi\tr(\gamma(\mathcal{B}^{-1}(u))^7)\right] = \exp\left[i\pi\left(\sum_{i=0}^{t-1}u_i\kappa_i+\sum_{i<j}u_iu_j\kappa_{ij}+\sum_{i<j<k}u_iu_ju_k\kappa_{ijk}\right)\right],
\end{equation}
where
\begin{align}
    \kappa_i &= \tr\left(\gamma\alpha_i^7\right)\\
    \kappa_{ij} &= \tr\left[\gamma\left(\alpha_i^6\alpha_j+\alpha_i^5\alpha_j^2+\alpha_i^4\alpha_j^3+\alpha_i^3\alpha_j^4+\alpha_i^2\alpha_j^5+\alpha_i\alpha_j^6\right)\right]\\
    \kappa_{ijk} &= \tr\left[\gamma\left(\alpha_i^4\alpha_j^2\alpha_k+\alpha_i^4\alpha_j\alpha_k^2+\alpha_i^2\alpha_j^4\alpha_k+\alpha_i^2\alpha_j\alpha_k^4+\alpha_i\alpha_j^4\alpha_k^2+\alpha_i\alpha_j^2\alpha_k^4\right)\right].
\end{align}
One could, therefore, enact the gate $\Sigma^\gamma$ on a physical $q$-block of qubits of $\qcode_2$ via the circuit
\begin{equation}\label{eq:designed_implementation_circuit}
    \prod_{i=0}^{t-1}\mathsf{Z}^{\kappa_i}[i]\;\prod_{i<j}\mathsf{CZ}^{\kappa_{ij}}[i,j]\;\prod_{i<j<k}\CCZ^{\kappa_{ijk}}[i,j,k]
\end{equation}
which is a non-transversal circuit of qubits $\CLZ$ gates with $\ell \leq 3$. However, we wish to obtain a depth-one circuit of $\CLZ$ gates. To obtain this, it is sufficient to simply modify the physical space of $\qcode_2$ by duplicating each physical qubit $\Omega$ times,\footnote{This simply means further encoding each physical qubit of $\qcode_2$ into the code given by $\ket{0} \mapsto \ket{0}^{\otimes \Omega}$; $\ket{1} \mapsto \ket{1}^{\otimes \Omega}$.} where $\Omega$ is the maximum number of times a qubit is involved in the circuit in Equation~\eqref{eq:designed_implementation_circuit}; indeed one can always choose $\Omega = 1 + (t-1) + \begin{pmatrix}
    t-1 \\2
\end{pmatrix} = 1+\frac{t(t-1)}{2}$. Then, these qubits may be acted on with a depth-one circuit of physical $\mathsf{Z}, \mathsf{CZ}$ and $\CCZ$ gates to create the desired phase $\exp\left[i\pi\tr\left(\gamma(\mathcal{B}^{-1}(u))^7\right)\right]$. From here, we may prove Theorem~\ref{thm:main_asymp_good}.
\begin{proof}[Proof of Theorem~\ref{thm:main_asymp_good}]
    The proof follows in essentially the same way as the proof of Theorem~\ref{thm:main_addressable_CCZ}. In this case, the code $\qcode_1$ has parameters $[[n_0t,3k_0t,d_1 \geq d_0]]_2$, and so the code $\qcode_2$ has parameters $[[n_0t,3k_0,d_2 \geq d_0]]_2$. After the duplication of each physical qubit $\Omega$ times, one obtains the final code $\qcode_3$ with parameters $[[n_0t\Omega,3k_0,d_3 \geq d_0]]_2$. Since $t$, and therefore $\Omega$, are constants, this code is asymptotically good, and the $k_0$ sets of 3 logical qubits may be addressed with logical $\CCZ$ gates in the way described by construction.
\end{proof}
\begin{remark}
    Note that the modification of the physical space here by duplicating qubits is morally the same as the use of the MFEs in Section~\ref{sec:near_asymp_good_concatenation}, although it is phrased differently here for convenience.
\end{remark}

\section*{Acknowledgements}

All authors are grateful to the participants of the MIT Quantum Fault Tolerance Reading Group, which led to this work, and in particular to Gefen Baranes for conversations regarding the implementability in neutral atom arrays.
We thank Chris Pattison, Mike Vasmer and Anqi Gong for valuable feedback on earlier versions of this paper.

\bibliographystyle{alpha}
\bibliography{main}

\appendix

\section{Addressable \texorpdfstring{$T$}{} Gates via Addressable Triorthogonality}\label{sec:transversal_T}
In this section, we prove that a binary matrix with the addressable 3-orthogonality property, as defined in Section~\ref{sec:addressable_orthogonality}, leads to a quantum code with an addressable $T$ gate, up to the action of a corrective Clifford. While we do not have an instantiation in mind, this directly extends the original notion of triorthogonal matrices/triorthogonal quantum codes defined by Bravyi and Haah~\cite{bravyi2012magic} (which coincides exactly with our definition of 3-orthogonality for a binary matrix) for which the associated quantum code can support a transversal $T$-gate, up to a corrective Clifford. Accordingly, we name the matrices, and the corresponding quantum codes ``addressable triorthogonal''. Although this is just our definition of addressable 3-orthogonality for a binary matrix, we enunciate this here.

\begin{definition}\label{def:addressable_triorthogonality}
    A matrix $G \in (\mathbb{F}_2)^{m \times n}$ with rows $(g^a)_{a=1}^m$ is called \textbf{addressable triorthogonal} if, for some $k \in [m]$, for every $A \in [k]$, there exists a vector $\Gamma^A \in \mathbb{F}_2^n$ such that, for all $a_1, a_2, a_3 \in [m]$,
    \begin{equation}
        \sum_{i=1}^n\Gamma_i^Ag_i^{a_1}g_i^{a_2}g_i^{a_3} = \begin{cases}
            1 &\text{ if } a_1 = a_2 = a_3 = A\\
            0 &\text{ otherwise}
        \end{cases}.
    \end{equation}
\end{definition}
One can check that this definition exactly coincides with the notion of addressable 3-orthogonality for a binary matrix given in Section~\ref{sec:addressable_orthogonality}. We also note that, given Lemma~\ref{lem:ortho_to_address}, a natural way to obtain addressable triorthogonal matrices is via the construction of binary matrices with 4-orthogonality, also defined in Section~\ref{sec:addressable_orthogonality}. We now show the claim that the quantum code $\qcode$ associated to a addressable triorthogonal matrix can support an addressable $T$ gate, up to corrective Clifford operations.
\begin{theorem}
    Suppose $G$ is an addressable triorthogonal matrix. Then, the following is true when acting on a code state of the associated quantum code $\qcode$, defined in Equation~\eqref{eq:repeated_qcode_defn}:
    \begin{equation}
        \overline{T[A]} = U_AT^{\Gamma^A},
    \end{equation}
    for some Clifford operation $U_A$.
\end{theorem}
\begin{proof}
    Our proof closely follows that of Lemma 2 in~\cite{bravyi2012magic}. As throughout the paper, arithmetic takes place in the finite field (the binary field in this case), unless we specify otherwise.

    Given $u \in \mathbb{F}_2^k$, the logical computational basis state $\overline{\ket{u}}$ in the quantum code $\qcode$ is
    \begin{equation}
        \overline{\ket{u}} \propto \sum_{g \in \mathcal{G}_0}\Ket{\sum_{a=1}^ku_ag^a+g},
    \end{equation}
    and we aim to show that there is a Clifford operation $U_A$ such that
    \begin{equation}
        U_AT^{\Gamma^A}\overline{\ket{u}} = \exp\left(\frac{i\pi}{4}u_A\right)\overline{\ket{u}}.
    \end{equation}
    We first consider some state $\ket{f}$, where $f = \sum_{a=1}^mu_ag^a$, for some $u_a \in \mathbb{F}_2$. Then,
    \begin{equation}\label{eq:T_action_f}
        T^{\Gamma^A}\ket{f} = \prod_{i=1}^n\exp\left(\frac{i\pi}{4}\Gamma_i^Af_i\right)\ket{f}.
    \end{equation}
    We then use the fact from the proof of Lemma 2 of~\cite{bravyi2012magic} that, for a binary vector $y \in \mathbb{F}_2^m$,
    \begin{equation}
        \exp\left(\frac{i\pi}{4}\sum_{a=1}^my_a\right) = \exp\underbrace{\left(\frac{i\pi}{4}\sum_{a=1}^my_a - \frac{i\pi}{2}\sum_{a<b}y_ay_b+i\pi\sum_{a<b<c}y_ay_by_c\right)}_{\text{Literal sums of integers}},
    \end{equation}
    where, here, sums on the right-hand side are the usual sums of integers, but we emphasise that the sum on the left-hand side takes place in the field $\mathbb{F}_2$ (it is taken modulo 2), as usual.

    Now, referring to the definition of addressable triorthogonality, we see that there are integers $\Lambda_a^A$, for $a \in [m]$, and $\Lambda_{ab}^A$ for $a,b \in [m]$, $a<b$, depending only on the matrix and the choice of $A \in [k]$, for which
    \begin{align}\label{eq:add_triorth_first}
        \sum_{i=1}^n\Gamma_i^Ag_i^a &= \begin{cases}
            2\Lambda_A^A+1 &\text{ if } a=A,\\
            2\Lambda_a^A &\text{ if } a \neq A,
        \end{cases}\\\label{eq:add_triorth_second}
        \sum_{i=1}^n\Gamma_i^Ag_i^ag_i^b &= 2\Lambda_{ab}^A,
    \end{align}
    where these sums are literal sums of integers. Note that Equation~\eqref{eq:add_triorth_first} is seen to be true by setting $a_1=a_2=a_3$ in Definition~\ref{def:addressable_triorthogonality}, whereas Equation~\eqref{eq:add_triorth_second} is seen to be true by taking $a_1 \neq a_2 = a_3$. Therefore, Equations~\eqref{eq:T_action_f} to~\eqref{eq:add_triorth_second} and Definition~\ref{def:addressable_triorthogonality} give
    \begin{equation}
        T^{\Gamma^A}\ket{f} = \exp\left(\frac{i\pi}{4}u_A\right)\prod_{a=1}^m\exp\left(\frac{i\pi}{2}u_a\Lambda^A_a\right)\prod_{a<b}\exp\left(i\pi\Lambda^A_{ab}u_au_b\right)\ket{f}.
    \end{equation}
    Note that the phase $\exp\left(\frac{i\pi}{4}u_A\right)$ is the phase that we want, whereas the remainder is the phase that we wish to cancel with a Clifford operation $U_A$.

    Now, by assumption our matrix $G$ is addressable 3-orthogonal with respect to the all-1's vector, $\mathbf{1}^k$.\footnote{Strictly speaking, this remark is redundant since the fact that $\tau \in (\mathbb{F}_2^*)^k$ implies that $\tau = \mathbf{1}^k$.} Therefore, $G$ is 3-orthogonal with respect to $(\Gamma,\mathbf{1}^k)$, for some $\Gamma$, by Lemma~\ref{lem:address_to_ortho}.  Then, by Lemma~\ref{lem:ortho_to_assump}, Assumption~\ref{assump:independence} is satisfied, meaning that the only linear dependence amongst the rows of $G$ is amongst the rows of $G_0$ which, as usual, is the latter $m-k$ rows of $G$. Moreover, we may say without loss of generality that the rows of $G_0$ are linearly independent as well, since eliminating rows of $G_0$ does not alter the code $\qcode$. As such, $f_i = \sum_{a=1}^mu_ag^a_i$ implies that $(u_a)_{a=1}^m$ are uniquely determined by $(f_i)_{i=1}^n$, and so there exists a matrix $B \in (\mathbb{F}_2)^{m \times n}$ such that $u_a = \sum_{i=1}^nB_{ai}f_i$. Using the fact that
    \begin{equation}
        \exp\left(\frac{i\pi}{2}\sum_{i=1}^nv_i\right) = \exp\underbrace{\left(\frac{i\pi}{2}\sum_{i=1}^nv_i+i\pi\sum_{i<j}v_iv_j\right)}_{\text{Literal sums of integers}},
    \end{equation}
    for a binary vector $v \in \mathbb{F}_2^n$, one may check that the un-wanted phase may be written
    \begin{equation}
        \prod_{a=1}^m\exp\left(\frac{i\pi}{2}u_a\Lambda^A_a\right)\prod_{a<b}\exp\left(i\pi\Lambda_{ab}^Au_au_b\right) = \prod_{i=1}^n\exp\left(\frac{i\pi}{2}\kappa_if_i\right)\prod_{i<j}\exp\left(i\pi\kappa_{ij}f_if_j\right),
    \end{equation}
    where
    \begin{align}
        \kappa_i &= \sum_{a=1}^mB_{ai}\Lambda_a^A + 2\sum_{a<b} \Lambda^A_{ab}B_{ai}B_{bi}\\
        \kappa_{ij} &= \sum_{a=1}^mB_{ai}B_{bj}\Lambda_a^A+\sum_{a<b}\Lambda_{ab}^A(B_{ai}B_{bj}+B_{aj}B_{bi}).
    \end{align}
    We therefore define the $n$-qubit Clifford operator
    \begin{equation}
        U_A = \prod_{i=1}^n(S^\dagger[i])^{\kappa_i}\prod_{i<j}(\mathsf{CZ}[i,j])^{\kappa_{ij}},
    \end{equation}
    and we find that
    \begin{equation}
        U_AT^{\Gamma^A}\ket{f} = \exp\left(\frac{i\pi}{4}u_A\right)\ket{f}.
    \end{equation}
    Finally, referring back to the form of the logical computational basis state $\overline{\ket{u}}$, we find by linearity that
    \begin{equation}
        U_AT^{\Gamma^A}\overline{\ket{u}} = \exp\left(\frac{i\pi}{4}u_A\right)\overline{\ket{u}},
    \end{equation}
    as required.
\end{proof}

\section{Diagonal Clifford Hierarchy for Qudits of Prime Power Dimension}
\label{sec:diagonal_Clifford}
In this section, we will discuss the importance of the following gates defined in Section~\ref{sec:prelim_non_clifford_addressability} in the context of the Clifford hierarchy, which is a very important mathematical structure in fault-tolerant quantum computing:
\begin{align}
    U_{q,P} = \sum_{\boldsymbol{\eta} \in \mathbb{F}_q^e}\exp(i\pi\tr(P(\boldsymbol{\eta})))\ket{\boldsymbol{\eta}}\bra{\boldsymbol{\eta}}.
\end{align}
These are diagonal gates defined for $e$ qudits of dimension $q=2^t$, where $P$ is an $e$-variate polynomial over $\mathbb{F}_q$.
Our main intention here is to convince the reader that this set of gates not only appears as a straightforward generalisation of the well-known $\CLZ$ gates, but is very natural in its own right. 

While determining the full structure of the Clifford hierarchy is still an open problem,~\cite{cui2017diagonal} gave a full characterisation of the diagonal Clifford hierarchy gates for qudits of prime dimension. In particular, for a set of $n$ qubits, it was found that any such gate may be generated from gates of the form
\begin{equation}
    \sum_{\mathbf{j}\in\mathbb{F}_2^n}\exp\left(\frac{2\pi i}{2^m}P(j_1, \ldots, j_n)\right)\ket{\mathbf{j}}\bra{\mathbf{j}},
\end{equation}
where $P$ is an $n$-variate polynomial over $\mathbb{F}_2$. The level in the Clifford hierarchy of the generated gate depends in a way made explicit in~\cite{cui2017diagonal} on the degree of the various polynomials $P$ and the ``precisions'' $m$ that appear. The same paper also observes, given the natural isomorphism between the states and gates of qudits of prime dimension (see Section~\ref{sec:qudit_to_qubit_prelim}), and a smaller number of qudits of prime-power dimension, and the associated isomorphism of Clifford hierarchies~\cite{gottesman2016surviving}, that a similar statement must apply to the diagonal Clifford hierarchy for qudits of prime-power dimension. Indeed, it is easily shown that all such gates for $e$ qudits of dimension $q=2^t$ may be generated from gates of the form
\begin{equation}
    \sum_{\boldsymbol{\eta} \in \mathbb{F}_q^e}\exp\left(\frac{2\pi i}{2^m}\tr(P(\eta_1, \ldots, \eta_e))\right)\ket{\boldsymbol{\eta}}\bra{\boldsymbol{\eta}}.
\end{equation}
The argument goes as follows. Because of the natural isomorphism,~\cite{cui2017diagonal} establishes for us that all diagonal entries must be some power of a phase $\exp\left(\frac{2\pi i}{2^m}\right)$, and so the gate may be written as
\begin{equation}\label{eq:general_diagonal_cliff_hierarchy}
    \sum_{\boldsymbol{\eta}\in \mathbb{F}_q^e}\exp\left(\frac{2\pi i}{2^m}\Theta(\boldsymbol{\eta})\right)\ket{\boldsymbol{\eta}}\bra{\boldsymbol{\eta}}
\end{equation}
for some map $\Theta:\mathbb{F}_q^e \to \mathbb{Z}_{2^m}$. Such a map may be written as $\Theta = \sum_{i=0}^{m-1}2^i\theta_i$, where $\theta_i:\mathbb{F}_q^e \to \mathbb{Z}_2 = \mathbb{F}_2$ is some other map. Therefore, the gate \eqref{eq:general_diagonal_cliff_hierarchy} may be generated from gates of the form
\begin{equation}
    \sum_{\boldsymbol{\eta}\in \mathbb{F}_q^e}\exp\left(\frac{2\pi i}{2^m}\theta(\boldsymbol{\eta})\right)\ket{\boldsymbol{\eta}}\bra{\boldsymbol{\eta}}
\end{equation}
where $\theta:\mathbb{F}_q^e \to \mathbb{F}_2$ is some map. However, we know that any map $\tilde{\theta}: \mathbb{F}_q^e \to \mathbb{F}_q$ may be written as some polynomial over $\mathbb{F}_q$ via, for example, Lagrange interpolation. Therefore, since $\ker(\tr_{\mathbb{F}_q/\mathbb{F}_2})$ is a proper subset of $\mathbb{F}_q$, $\theta$ may be written as the trace of some polynomial over $\mathbb{F}_q$.

In particular, we find that gates $U_{q,P}$ are actually complete for the set of diagonal Clifford hierarchy gates with entries $\pm 1$, and so they are certainly important to consider.

We would note, however, that whereas~\cite{cui2017diagonal} gives an explicit, and relatively simple description of the \textit{level} of the Clifford hierarchy of any diagonal gate in terms of the degrees of its polynomials, our argument does not preserve this simple description. In fact, it seems that any description must necessarily be not so simple. For example, a natural generating set for the set of diagonal Clifford hierarchy gates with $\pm 1$ on the diagonal for one qudit of dimension $q$ are the gates $U_{q,\ell}^\beta$ defined in Section~\ref{sec:prelim_non_clifford_addressability} as
\begin{equation}
    U_{q,\ell}^\beta = \sum_{\eta \in \mathbb{F}_q}\exp\left(i\pi\tr(\beta\eta^\ell)\right)\ket{\eta}\bra{\eta}.
\end{equation}
Then, for small values of $\ell$, one can check that the Clifford hierarchy level of $U_{q,\ell}^\beta$ behaves erratically with increasing $\ell$, as opposed to the case of prime-dimensional qudits~\cite{cui2017diagonal}, where the level of the Clifford hierarchy increases steadily with the polynomial degree. For example, one can check that (for all $q$ large enough) the smallest $\ell$ for which $U_\ell^\beta$ is in the 2nd level of the Clifford hierarchy is $\ell = 3$, the smallest $\ell$ for which $U_{q,\ell}^\beta$ is in the 3rd level of the Clifford hierarchy is $\ell=7$, the smallest 4th level gate has $\ell=15$, the smallest 5th level gate has $\ell=31$, and in each case the level is independent of $\beta$ as long as $\beta \neq 0$. The level also varies significantly with $\ell$ between these values. In light of this, one would reasonably conjecture that the smallest $\ell$ for which $U_{q,\ell}^\beta$ is in the $t$-th level of the Clifford hierarchy is $\ell = 2^t-1$, for all sufficiently large finite fields of characteristic 2, and in each case the Clifford hierarchy level is independent of $\beta \neq 0$. In general, a description of the diagonal Clifford hierarchy for prime-power dimensional qudits in terms of the degrees of the polynomials in the natural form \eqref{eq:general_diagonal_cliff_hierarchy} awaits further work.

\end{document}